\documentclass[11pt,journal,onecolumn]{IEEEtran}
\usepackage[T1]{fontenc}
\usepackage{microtype}
\usepackage{amsmath,amssymb,amsthm,mathtools}
\usepackage{booktabs,array,graphicx}
\usepackage{stfloats}

\usepackage{tikz}
\usetikzlibrary{arrows.meta}
\usepackage{cite}
\usepackage{xcolor}
\newenvironment{revision}{}{}
\newenvironment{finalrevision}{}{}
\DeclareRobustCommand{\reviewchange}[1]{#1}
\newcommand{\notationmark}[1]{{#1}}
\newcommand{\measeig}[1]{\reviewchange{\lambda_{#1}}}
\newcommand{\dloc}{\notationmark{d_{\mathrm{loc}}}}

\newcommand{\rtail}{\reviewchange{\widetilde m}}
\newcommand{\packbal}{\reviewchange{\widetilde m_0}}

\newcommand{\rbal}{\notationmark{r_{\mathrm{bal}}}}
\newcommand{\dcube}{\reviewchange{m'}}

\newcommand{\seed}{\notationmark{R}}
\newcommand{\seedval}{\notationmark{r}}
\newcommand{\seedlaw}{\notationmark{\pi}}
\newcommand{\obs}{\notationmark{Y}}

\newcommand{\paulix}{\notationmark{\mathsf X}}

\newcommand{\pauliz}{\notationmark{\mathsf Z}}

\newcommand{\cliffgate}{\notationmark{\mathsf C}}

\newcommand{\twirlfour}{\notationmark{\mathcal T^{(4)}}}

\newcommand{\sgnrow}{\notationmark{\mathsf s}}
\newcommand{\absrow}{\notationmark{\varpi}}
\newcommand{\sectorcoef}{\notationmark{\mathfrak a}}
\newcommand{\sectorflag}{\notationmark{\varepsilon}}
\newcommand{\sectorprod}{\notationmark{\mathsf R}}
\newcommand{\oparg}{\notationmark{\mathsf A}}

\newcommand{\utensor}{\notationmark{u_4}}
\newcommand{\momconst}{\notationmark{K_8}}
\newcommand{\covlo}{\notationmark{c_{\mathrm{cov}}}}
\newcommand{\covhi}{\notationmark{C_{\mathrm{cov}}}}
\newcommand{\covrate}{\notationmark{A_{\mathrm{cov}}}}
\newcommand{\covstar}{\notationmark{A_{\mathrm{cov},*}}}
\newcommand{\realj}{\notationmark{J_{\mathrm{cx}}}}

\newcommand{\eigbasis}{\notationmark{V_{\mathrm{eig}}}}
\newcommand{\actweights}{\notationmark{p^{\mathrm{act}}}}
\newcommand{\transfermat}{\notationmark{\mathsf T}}
\newcommand{\activity}{\notationmark{b^{\mathrm{act}}}}

\usepackage[hidelinks]{hyperref}
\usepackage{aliascnt}
\allowdisplaybreaks[2]
\newtheorem{theorem}{Theorem}
\newaliascnt{proposition}{theorem}
\newtheorem{proposition}[proposition]{Proposition}
\aliascntresetthe{proposition}
\newaliascnt{lemma}{theorem}
\newtheorem{lemma}[lemma]{Lemma}
\aliascntresetthe{lemma}
\newaliascnt{corollary}{theorem}
\newtheorem{corollary}[corollary]{Corollary}
\aliascntresetthe{corollary}
\newaliascnt{assumption}{theorem}

\aliascntresetthe{assumption}
\theoremstyle{definition}
\newaliascnt{definition}{theorem}

\aliascntresetthe{definition}
\newaliascnt{remark}{theorem}

\aliascntresetthe{remark}
\newcommand{\C}{\mathbb C}
\newcommand{\R}{\mathbb R}
\newcommand{\E}{\mathbb E}
\newcommand{\Prob}{\mathbb P}
\newcommand{\tr}{\operatorname{tr}}
\newcommand{\rank}{\operatorname{rank}}
\newcommand{\op}{\mathrm{op}}
\newcommand{\F}{\mathrm F}
\newcommand{\KL}{\mathrm{KL}}
\newcommand{\cD}{\mathcal D}
\newcommand{\cM}{\mathcal M}
\newcommand{\cL}{\mathcal L}
\newcommand{\cC}{\mathcal C}

\newcommand{\rhoCF}{\widehat\rho_{\mathrm{OMD}}}
\newcommand{\rhoMW}{\widehat\rho_{\mathrm{MW}}}
\newcommand{\rhoPLS}{\widehat\rho_{\mathrm{PLS}}}
\DeclareMathOperator*{\argmin}{arg\,min}
\hypersetup{pdftitle={Minimax Quantum State Tomography with Periodic Clifford Measurements},pdfauthor={Hongru Zhao},pdfsubject={Quantum tomography and minimax estimation}}

\title{Minimax Quantum State Tomography\\with Periodic Clifford Measurements}
\author{Hongru Zhao\thanks{Hongru Zhao is with the School of Statistics,
University of Minnesota, Minneapolis, MN, USA
(e-mail: \href{mailto:zhao1118@umn.edu}{zhao1118@umn.edu}).}}
\begin{document}
\maketitle

\begin{abstract}
Quantum state tomography provides a foundation for characterizing state preparation and predicting measurement outcomes.
We characterize the minimax expected trace norm risk of quantum state tomography over classes with polynomial spectral decay under randomized nonadaptive single copy measurements.
Under an explicit sufficient block size condition for the two layer periodic Clifford ensemble, we establish statewise expected trace norm oracle inequalities for operator norm minimum distance and measurement weighted projected least squares, with bounds that depend on each state's actual spectral tail, while the estimators require no prior knowledge of the state's structural characteristics, including its rank, eigenbasis, or spectral structure.
Over these polynomial spectral decay classes, both estimators attain the minimax expected trace norm rates up to constants.
As a direct consequence, the same estimators attain the minimax rate over every fixed rank constrained class, while remaining adaptive to the unknown structural characteristics of the state.
For admissible logarithmic block sizes, the periodic Clifford measurements attain these guarantees with logarithmic elementary gate depth.
The main technical ingredient is a uniform bound on the fourth moment of projectors for the periodic Clifford ensemble. The bound is sharp in dimension and holds for arbitrary test states, including those entangled across blocks.
Covariance concentration and positivity yield adaptive oracle bounds, while a shared rotation information argument on spectral tail packings establishes matching lower bounds.
We also provide polynomial time reconstruction algorithms, numerical comparisons, and Lean formalizations of all labeled theoretical results.
\end{abstract}
\section{Introduction}
\label{sec:introduction}
\begingroup
Quantum state tomography reconstructs an unknown density matrix to
characterize state preparation and predict measurement outcomes
\cite{JamesKwiatMunroWhite2001,EisertEtAl2020}. Structural assumptions can
reduce its statistical cost: low rank motivates compressed sensing methods
\cite{GrossLiuFlammiaBeckerEisert2010,FlammiaGrossLiuEisert2012}, while spatial or tensor network
structure can enable scalable tomography \cite{CramerEtAl2010}.
Spectral thresholding and matrix oracle inequalities address unknown
effective rank \cite{ButuceaGutaKypraios2015,ButuceaGutaKypraios2016Corrigendum,KoltchinskiiXia2015,Koltchinskii2011}.
Sparse Pauli models provide another form of structure
\cite{CaiKimWangYuanZhou2016}. We study spectral decay with an unrestricted
eigenbasis, without assuming spatial structure or sparsity in a fixed
operator basis. The reconstruction should exploit this spectral structure
without prior knowledge of the rank, eigenbasis, or decay parameters.

Measurement restrictions also determine copy complexity
\cite{HaahHarrowJiWuYu2017,LoweNayak2025}. We reconstruct a state $\rho$ on $n$ qubits, of dimension $d=2^n$, from $T$ independently prepared copies
under full trace norm loss $\|\widehat\rho-\rho\|_{\mathrm{tr}}$.
Each copy is measured separately, and the measurement settings are chosen
without using earlier outcomes. Classical and shallow shadows motivate
the randomized measurement construction
\cite{HuangKuengPreskill2020,BertoniEtAl2024}.
Cho and Kim obtain low rank tomography bounds using periodic overlapping
Clifford blocks, and full state bounds using repeated measurements of
selected bases, including a periodic two layer analysis
\cite[Theorems~1--2 and Appendices~D, F--G]{ChoKim2026}.

A new periodic Clifford fourth moment analysis yields logarithm free minimax tomography under shallow nonadaptive measurements, with reconstruction that adapts to unknown rank, eigenbasis, and spectral decay.\footnote{Here \emph{nonadaptive} means that measurement settings may be randomized, but their choice does not depend on earlier outcomes. Adaptation to unknown state structure refers to reconstruction: the estimators require no prior knowledge of the rank, eigenbasis, or spectral decay parameters.}
We use two logical layers of Clifford blocks on $k$ qubits, with a fresh
independent circuit on each copy. For even divisors $k$ of $n$ satisfying
the explicit sufficient block size condition~\eqref{eq:headline-block-condition},
our new reconstruction estimators remove the dimension logarithm from
the low rank and intermediate spectral upper bounds.
As in \cite{ChoKim2026}, the circuit can be implemented by a local quantum
circuit of depth $O(k)$; hence, when an admissible $k=O(\log n)$ exists,
the depth is $O(\log n)$.

We analyze the \emph{operator norm minimum distance (OMD) estimator} and
\emph{measurement weighted projected least squares (MW-PLS)}, defined in
Sections~\ref{sec:omd-estimator} and~\ref{sec:mw-estimator}. OMD compares the empirical score mean
with its expectation under a candidate state in operator norm;
see~\eqref{eq:raw-score}--\eqref{eq:forward-objective}. The exact weighted PLS fit
projects the inverse channel estimate in the metric of~\eqref{eq:mw-weighted-projection}, equivalently
minimizing the quadratic objective~\eqref{eq:mw-quadratic}. Its density contrast
interpretation in Supplement Section~\ref*{app:mw-pls} connects it to positive least
squares \cite[Sec.~7]{AcharyaKypraiosGuta2019}. Under the stated block
condition and optimization tolerances, both satisfy statewise trace loss oracle
inequalities for every density matrix. The bounds depend on the actual spectral
tail defined in~\eqref{eq:spectral-tail}, without rank or spectral decay
inputs to the estimators. Optimizing over the spectral classes~\eqref{eq:decay-class}
yields minimax rates, up to constants, matching the lower bound
over all randomized nonadaptive single copy measurement designs.

The main circuit estimate is a uniform eighth overlap moment bound, equivalently a
fourth projector moment bound, under the reference law of Section~\ref{sec:observations},
uniformly over all unit test vectors, including entangled ones. The local Clifford commutant and stabilizer state
moment formulas \cite{GrossNezamiWalter2021,HelsenWalter2023} supply the algebra. Our
contribution is the periodic four copy expansion and uniform sector contraction bound in
Lemmas~\ref{lem:ck-four-moment-expansion}--\ref{lem:ck-four-copy-boundary}, which, with the
coefficient estimates, yield Theorem~\ref{thm:ck-fourth-projector}. Related work develops
Clifford replica tensor networks \cite{MagniEtAl2025} and characterizes the Clifford commutant
for arbitrary qubit and replica numbers \cite{BittelEtAl2026}. Born reweighting and the
covariance conditioning bounds allow us to apply \cite[Theorem~2, arXiv~v3]{AbdallaZhivotovskiy2023},
obtaining the logarithm free expected operator norm bound in Proposition~\ref{prop:primitive-covariance}.
Combined with the oracle geometry, this gives the trace loss guarantees. The minimax
lower bound uses a Grassmann packing in the spectral tail, a shared hidden Haar
rotation, a one copy chi square information bound, and a nonadaptive Fano reduction.

Dense matrix algorithms accompany the statistical guarantees. We compare
OMD and weighted PLS with a projected least squares (PLS) baseline, which applies
Frobenius projection to linear inversion \cite{GutaKahnKuengTropp2020}.
The algorithms use spectral projection, ADMM, and Frank--Wolfe
\cite{Jaggi2013}; Section~\ref{sec:implementation} summarizes their
computation and stopping certificates in the explicit dimension $d$,
with detailed guarantees in Supplement Section~\ref*{supp:computational-proofs}.
Simulations compare reconstruction error, runtime, and memory using
shared measurement records. The scope of the available PLS guarantees
is discussed in Section~\ref{sec:comparisons}.

Section~\ref{sec:model} defines the model and measurement ensemble;
Section~\ref{sec:main-results} states the estimators and minimax
guarantees; Sections~\ref{sec:oracle}--\ref{sec:lower} prove them.
Sections~\ref{sec:implementation}--\ref{sec:discussion} cover computation,
experiments, and open directions. The appendices contain supporting
proofs and the scope of Lean verification. The separate Supplement
develops the local Clifford algebra, gives further proofs, and documents
the numerical algorithms and reproducibility details.
\par
\endgroup
\label{sec:introduction-end}

\section{Experiment and Parameter Classes}
\label{sec:model}

\subsection{States, loss, and spectral classes}

\begin{revision}
For every integer $d\geq1$, define
\end{revision}
\begin{equation}
 \cD_d=\{\rho\in\C^{d\times d}:\rho\succeq0,\ \tr\rho=1\}.
 \label{eq:state-space}
\end{equation}
\begin{revision}
For the periodic circuit experiment we specialize to $d=2^n$ with $n\geq1$.
\end{revision}
We use the operator, Frobenius, and full trace norms, denoted by
$\|\cdot\|_{\op}$, $\|\cdot\|_\F$, and
$\|\cdot\|_{\mathrm{tr}}$. The Frobenius inner product is
{\color{black}$\langle\oparg{},\mathsf B\rangle_\F=\operatorname{Re}\tr(\oparg{}^*\mathsf B)$}.
The conventional trace distance is half our loss:
$D_{\mathrm{tr}}(\rho,\sigma)=\tfrac12\|\rho-\sigma\|_{\mathrm{tr}}
=\sup_M\operatorname{TV}(p_\rho^M,p_\sigma^M)$ \cite{Holevo2011}.
Here the supremum is over all POVMs $M$, and $p_\rho^M$ and $p_\sigma^M$
are the induced classical measurement outcome distributions; total
variation uses the convention of half the $\ell_1$ distance for discrete
outcomes. Thus TV compares the classical outcome laws, while trace distance
compares the density matrices. In particular,
$\|\rho-\sigma\|_{\mathrm{tr}}\leq2$ for states. Throughout, $\ln$ denotes the natural logarithm;
base two logarithms are written $\log_2$ explicitly.

For decreasing eigenvalues $\lambda_j(\rho)$, put
\begin{equation}
 \tau_s(\rho)=\sum_{j>s}\lambda_j(\rho),\qquad 1\leq s\leq d.
 \label{eq:spectral-tail}
\end{equation}
This is the error of the best rank $s$ approximation when trace one is
\emph{not} required of the approximating matrix. The best rank $s$
density matrix approximation has full trace norm error $2\tau_s$.
For $\alpha>1$ and $L\geq1$, define
\begin{equation}
 \cC_{\alpha,d}(L)=\{\rho\in\cD_d:
  \tau_s(\rho)\leq Ls^{1-\alpha}\ \text{for all }s\}.
 \label{eq:decay-class}
\end{equation}
The eigenbasis is unrestricted. Polynomial eigenvalue decay implies this
tail condition, and conversely it implies
$\lambda_{2s}(\rho)\leq Ls^{-\alpha}$ when $2s\leq d$.
For $L\geq d^{\alpha-1}$ the class is all of $\cD_d$.

\subsection{Observation law and sufficient record}
\label{sec:observations}
For each of $T$ fresh copies, independently draw a public unitary $U_t$
from a known ensemble $\mathcal E$, apply it, and measure the computational
basis. The observation is $(U_t,b_t)$, with
\begin{equation}
 B_{U,b}=U^*|b\rangle\!\langle b|U,\qquad
 p_\rho(b\mid U)=\tr(\rho B_{U,b}).
 \label{eq:born-observation}
\end{equation}
Let $B=B_{U,b}$ denote a matrix valued random element, specifically a
random rank one Hermitian projector in $\C^{d\times d}$. Write
$B_t=B_{U_t,b_t}$ for the projector recorded on copy $t$, and let $\nu$
be the reference probability law of $B$ when the unitary $U\sim\mathcal E$ and $b$
is independently uniform.
Then $\E_\nu B=I/d$, whereas the observed projector has law
\begin{equation}
 \mathbb{P}_\rho(\mathrm{d}B)=d\tr(\rho B)\nu(\mathrm{d}B).
 \label{eq:rank-one-statistical-law}
\end{equation}
In these measure expressions, $B$ also serves as the integration variable,
ranging over rank one Hermitian projectors. All expectations with subscript
$\rho$ use this Born law.

Write $B^T=(B_1,\ldots,B_T)$ for the full projector record. It is a
sufficient statistic for $\rho$: given $B^T$, the conditional law of the
full data record $(U_t,b_t)_{t=1}^T$ can be chosen independently of $\rho$.
Indeed, relative to the convenient reference state {\color{black}$\rho_{\mathrm{mix}}=I/d$}, the
full record likelihood ratio is
{\color{black}\begin{equation}
 \frac{\mathrm{d}\mathbb{P}_\rho}{\mathrm{d}\mathbb{P}_{\rho_{\mathrm{mix}}}}((U_t,b_t)_{t=1}^T)
 =\prod_{t=1}^T d\tr(\rho B_t).
 \label{eq:projector-sufficiency}
\end{equation}}%
The likelihood ratio in \eqref{eq:projector-sufficiency} depends on the
full record only through $B^T$, so sufficiency follows from the factorization
criterion for dominated families \cite{HalmosSavage1949}; see also
\cite[Thm.~6.2.6]{CasellaBerger2002}. The conditional law of the full record
given $B^T$ under {\color{black}$\rho_{\mathrm{mix}}$}
therefore reconstructs it under every $\rho$. On standard Borel spaces
this is a state independent Markov kernel, including on null sets after an
arbitrary extension. Hence randomized estimators based on the two records
can have the same statewise output laws. 
\subsection{Reconstruction rules and minimax risks}
\label{sec:risks}
\begin{finalrevision}
Let $\cC\subseteq\cD_d$ be a class of states. A reconstruction rule is a
measurable Markov kernel from the observed record to $\cD_d$;
deterministic rules are included as point mass kernels. For a prescribed
measurement ensemble $\mathcal E$, the fixed ensemble minimax risk is
\begin{equation}
 \mathfrak R_T^{\mathcal E}(\cC)
 =\inf_\Gamma\sup_{\rho\in\cC}
   \E_{\rho,\mathcal E,\Gamma}\|\widehat\rho-\rho\|_{\mathrm{tr}}.
 \label{eq:fixed-ensemble-risk}
\end{equation}
Here $\Gamma$ acts on the projector record $B^T$, and the expectation
includes the measurement draws, outcomes, and reconstruction randomness.
The ensemble law is fixed in advance; its individual draws may be random.

For the unrestricted comparison, a randomized nonadaptive single copy
design \cite[Sec.~2.2]{LiuAcharya2024} uses a public seed
$\seed\sim\seedlaw$, independent of $\rho$, to select the entire POVM
sequence $(M_t(\seed))_{t=1}^T$ before any outcome is observed.
Seed and outcome spaces are standard Borel, with jointly measurable
measurement kernels. Conditional on $(\rho,\seed)$, outcomes are
independent and obey the Born laws
\[
 \Prob_\rho\{\obs_t\in\mathcal B\mid\seed=\seedval\}
 =\tr\bigl(\rho M_t(\seedval)(\mathcal B)\bigr).
\]
The estimator may use the seed and all outcomes. Define
\begin{equation}
 \mathfrak R_T(\cC)=\inf_{(\mathrm{design},\Gamma)}
 \sup_{\rho\in\cC}\E\|\widehat\rho-\rho\|_{\mathrm{tr}}.
 \label{eq:minimax-risk}
\end{equation}
The design and reconstruction rule may depend on $\cC$ and $T$, but not
on $\rho$. The infimum includes randomized measurement designs and
reconstruction rules, and excludes outcome adaptive and collective
measurements. Every fixed random basis ensemble is admissible, hence
$\mathfrak R_T(\cC)\leq\mathfrak R_T^{\mathcal E}(\cC)$.
\end{finalrevision}

\subsection{The periodic measurement ensemble}
\label{sec:periodic-ensemble}

Let $k>0$ be even and divide $n$, and put $m=n/k$. On the qubit ring define
\begin{equation}
 A_j=\{jk,\ldots,jk+k-1\},\quad
 A'_j=A_j+k/2\pmod n.
 \label{eq:periodic-supports}
\end{equation}
Write {\color{black}$\mathcal P_k=\{i^{\nu_{\mathrm{ph}}} {\paulix{}}^{\mathbf x}{\pauliz{}}^{\mathbf z}:{\nu_{\mathrm{ph}}}\in\{0,1,2,3\},\ {\mathbf x},{\mathbf z}\in\mathbb F_2^k\}$}
for the Pauli group, where {\color{black}${\paulix{}}^{\mathbf x}{\pauliz{}}^{\mathbf z}=\bigotimes_{r=1}^k{\paulix{}}^{{\mathbf x}_r}{\pauliz{}}^{{\mathbf z}_r}$}.
The full Clifford group is its unitary normalizer,
{\color{black}$\operatorname{Cliff}_{\mathrm{unit}}(k)=\{{\cliffgate{}}\in\mathrm U(2^k):{\cliffgate{}}\mathcal P_k {\cliffgate{}}^*=\mathcal P_k\}$}.
The projective group
{\color{black}$\operatorname{Cliff}(k):=\operatorname{Cliff}_{\mathrm{unit}}(k)/\{e^{i{\phi_{\mathrm{ph}}}}I:{\phi_{\mathrm{ph}}}\in\mathbb R\}$}
is finite; $\operatorname{Unif}(\operatorname{Cliff}(k))$
means the uniform law on that projective group, with fixed unitary
representatives.
\reviewchange{We write $\dloc$ for the Hilbert space dimension of a local
Clifford block; here $\dloc=2^k$.}
Independently sample one gate on every shifted support
$A'_j$ for layer 1 and every unshifted support $A_j$ for layer 2, with
independence also between layers. Embedding each gate as identity on
complementary qubits gives
{\color{black}\begin{equation}
 U=U_2U_1,\qquad
 U_1=\prod_{j=0}^{m-1}\cliffgate{}^{(1)}_{A'_j},\quad
 U_2=\prod_{j=0}^{m-1}\cliffgate{}^{(2)}_{A_j}.
 \label{eq:periodic-unitary}
\end{equation}}%
The probability law of the output $U$ is the \emph{periodic ensemble of Cho and Kim} \cite{ChoKim2026}; we write $U\sim\mathcal E_{\rm per}$.
Each block matrix {\color{black}$\cliffgate{}_A^{(\ell)}$} is uniform over the finite set of
representatives embedded on its support $A$; the superscript $\ell$
labels the layer. Repeat the construction with fresh independent samples
for each copy. Changing the fixed representatives affects only global
phases and leaves all measurement laws unchanged.

Within each layer, the product equals a tensor product after grouping the
qubits by support; adjacent layers overlap on half a block.
Figure~\ref{fig:periodic-circuit} illustrates the geometry.
Supplement Section~\ref*{app:clifford-background} records the gate conventions,
tensor ordering, and exact sampling procedure.

Uniform $k$ qubit Clifford blocks can be sampled exactly and synthesized
with expected classical runtime and elementary gate count polynomial
in $k$ \cite{KoenigSmolin2014,VanDenBerg2021}. The two layers in
\eqref{eq:periodic-unitary} count $k$ qubit blocks.

\begin{figure}[t]
 \centering
 \resizebox{\columnwidth}{!}{\begin{tikzpicture}[x=1cm,y=1cm,font=\small,>=Stealth]
 \renewcommand{\cliffgate}{\mathsf C}
 \definecolor{wrapred}{RGB}{169,46,53}
 \definecolor{shiftgold}{RGB}{163,105,25}
 \definecolor{layerblue}{RGB}{38,86,145}
 \definecolor{layergreen}{RGB}{39,112,91}
 \node[font=\small\bfseries] at (3.0,.94) {Layer 1};
 \node[font=\footnotesize] at (3.0,.58) {shifted blocks};
 \node[font=\small\bfseries] at (5.8,.94) {Layer 2};
 \node[font=\footnotesize] at (5.8,.58) {unshifted blocks};
 \node[align=center,font=\footnotesize] at (7.85,.78) {Measure\\$\pauliz$ basis};
 \foreach \q in {0,...,11} {
   \draw[black!70,line width=.5pt] (.6,-.5*\q)--(8.22,-.5*\q);
   \node[anchor=east] at (.42,-.5*\q) {$\q$};
 }
 \node[anchor=east,font=\footnotesize] at (.42,.58) {Qubit};
 \filldraw[fill=wrapred!9,draw=wrapred,line width=.9pt,rounded corners=2pt]
   (1.9,.20) rectangle (4.1,-.70);
 \filldraw[fill=wrapred!9,draw=wrapred,line width=.9pt,rounded corners=2pt]
   (1.9,-4.80) rectangle (4.1,-5.70);
 \node[text=wrapred] at (3.0,-.25) {$\cliffgate{}^{(1)}_{\{10,11,0,1\}}$};
 \node[text=wrapred] at (3.0,-5.25) {$\cliffgate{}^{(1)}_{\{10,11,0,1\}}$};
 \draw[wrapred,densely dashed,line width=.8pt] (1.9,-.25)--(1.53,-.25)--(1.53,-5.25)--(1.9,-5.25);
 \filldraw[fill=shiftgold!9,draw=shiftgold,line width=.9pt,rounded corners=2pt]
   (1.9,-.80) rectangle (4.1,-2.70);
 \node[text=shiftgold] at (3.0,-1.75) {$\cliffgate{}^{(1)}_{\{2,3,4,5\}}$};
 \filldraw[fill=layerblue!9,draw=layerblue,line width=.9pt,rounded corners=2pt]
   (4.7,.20) rectangle (6.9,-1.70);
 \node[text=layerblue] at (5.8,-.75) {$\cliffgate{}^{(2)}_{\{0,1,2,3\}}$};
 \filldraw[fill=layergreen!9,draw=layergreen,line width=.9pt,rounded corners=2pt]
   (4.7,-1.80) rectangle (6.9,-3.70);
 \node[text=layergreen] at (5.8,-2.75) {$\cliffgate{}^{(2)}_{\{4,5,6,7\}}$};

 \filldraw[fill=shiftgold!9,draw=shiftgold,line width=.9pt,rounded corners=2pt]
   (1.9,-2.80) rectangle (4.1,-4.70);
 \node[text=shiftgold] at (3.0,-3.75) {$\cliffgate{}^{(1)}_{\{6,7,8,9\}}$};
 \filldraw[fill=layerblue!9,draw=layerblue,line width=.9pt,rounded corners=2pt]
   (4.7,-3.80) rectangle (6.9,-5.70);
 \node[text=layerblue] at (5.8,-4.75) {$\cliffgate{}^{(2)}_{\{8,9,10,11\}}$};
 \foreach \q in {0,...,11} {
  \begin{scope}[shift={(7.85,-.5*\q)}]
   \filldraw[fill=white,draw=black!75,line width=.6pt] (-.22,-.18) rectangle (.22,.18);
   \draw[black!80,line width=.6pt] (-.14,-.07) arc[start angle=180,end angle=0,radius=.14];
   \draw[black!80,-{Stealth[length=2.5pt]},line width=.6pt] (0,-.07)--(.10,.13);
  \end{scope}
 }
 \draw[-{Stealth[length=4pt]},black!65] (2.5,-6.13)--(6.4,-6.13);
 \node[fill=white,inner sep=2pt,font=\footnotesize] at (4.45,-6.13) {circuit time};
 \node[text=black,font=\footnotesize] at (4.45,-6.53) {Dashed link: one four-qubit gate.};
\end{tikzpicture}
}
 \caption{Periodic Clifford measurement circuit for $n=12$ qubits and block size $k=4$, with $m=n/k=3$ blocks per layer. The first layer acts on shifted supports and the second on unshifted supports, followed by computational basis measurement. The two red pieces joined by the dashed annotation represent a single gate on $\{10,11,0,1\}$; they are not independent two qubit gates. All six four qubit Clifford blocks are sampled independently and uniformly from the projective Clifford group. Circuit time runs from left to right, so $U=U_2U_1$. This example lies outside the sufficient parameter range in \eqref{eq:headline-block-condition}.}
 \label{fig:periodic-circuit}
\end{figure}
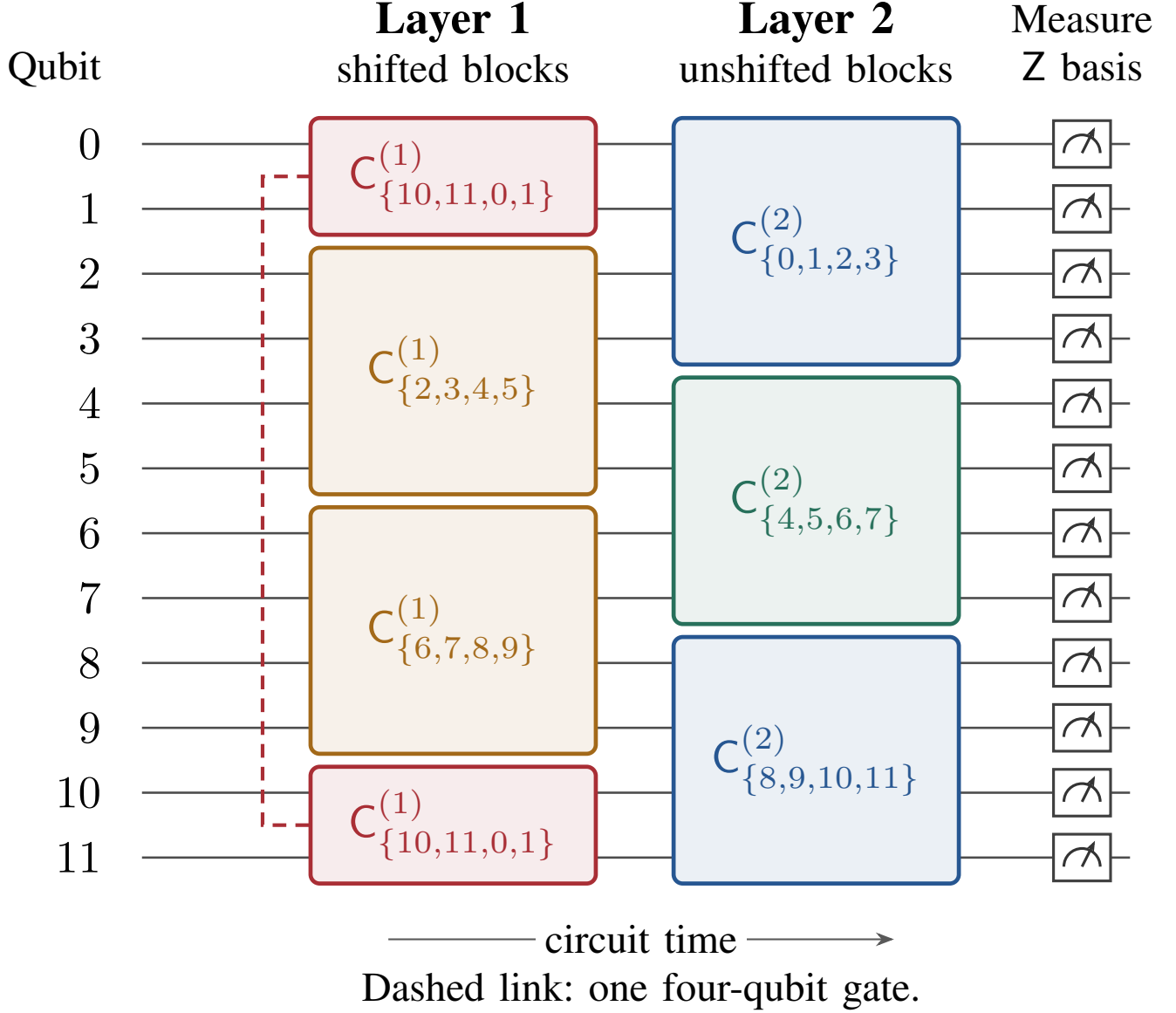

Our sufficient condition is
\begin{equation}
 \frac{92}{\ln2}\,n\leq k2^{k/2}.
 \label{eq:headline-block-condition}
\end{equation}
It implies $k\geq16$ and $d\geq65536$. The smallest admissible instance
is $n=k=16$; the smallest with more than one block per layer is
$(n,k)=(36,18)$. These constants are conservative sufficient conditions,
not an optimal threshold. The theorem covers the stated even divisors;
padding and other boundary conditions are not presumed to preserve it.

\section{Estimators and Main Results}
\label{sec:main-results}
All three estimators use the known measurement channel and calibrated map
\begin{align}
 \cM(A)&=\E_U\sum_b\tr(AB_{U,b})B_{U,b},
   \label{eq:measurement-channel}\\
 \cL(A)&=(d+1)\cM(A)-\tr(A)I.
   \label{eq:calibrated-map}
\end{align}
For the periodic ensemble, both $\cM$ and $\cL$ are diagonal in the Pauli basis.
These are the usual shadow channel objects
\cite{HuangKuengPreskill2020,BertoniEtAl2024,ChoKim2026}; exact brickwork
channel calculations also appear in \cite{ArienzoEtAl2023}.

\subsection{Projected Least Squares (PLS)}
\label{sec:pls-estimator}
We start with the established Projected Least Squares (PLS) estimator
\cite{GutaKahnKuengTropp2020}:
\begin{equation}
 \rhoPLS=\Pi_{\cD_d}^{\F}\left(
       \frac1T\sum_{t=1}^T\cM^{-1}(B_t)\right).
 \label{eq:candidate-one}
\end{equation}
Write $A_T=T^{-1}\sum_{t=1}^T\cM^{-1}(B_t)$ for the average of the
single sample inverse channel estimators.
Here $\Pi_{\cD_d}^{\F}(A)=\allowbreak\argmin_{\sigma\in\cD_d}\allowbreak\|A-\sigma\|_\F^2$
is the \emph{Frobenius metric projection} onto $\cD_d$, namely the unique
density matrix nearest to $A$ in Frobenius norm.
PLS projects the average analyzed by Cho and Kim
\cite{ChoKim2026} in Frobenius norm
\cite{GutaKahnKuengTropp2020}, without a rank input.
See also \cite[Sec.~8]{AcharyaKypraiosGuta2019}.

\subsection{Operator Norm Minimum Distance (OMD)}
\label{sec:omd-estimator}
For OMD, use the score and its exact mean
\begin{equation}
 Y_t=(d+1)B_t-I,\qquad \E_\rho Y_t=\cL(\rho).
 \label{eq:raw-score}
\end{equation}
The score is standard \cite{SchusterHaferkampHuang2025,ChoKim2026}.

With $\overline Y_T=T^{-1}\sum_tY_t$, set
\begin{equation}
 f_T(\sigma)=\|\overline Y_T-\cL(\sigma)\|_{\op},\qquad
 f_T^*=\min_{\sigma\in\cD_d}f_T(\sigma).
 \label{eq:forward-objective}
\end{equation}
Select a measurable feasible approximate minimizer $\rhoCF$ satisfying
\begin{equation}
 f_T(\rhoCF)\leq f_T^*+\delta_{\rm alg},\qquad
 \delta_{\rm alg}\leq\min\{1,\sqrt{d/T}\}.
 \label{eq:candidate-two-program}
\end{equation}
The compact feasible set gives an exact minimizer; finite ordered
approximations or the deterministic solver in Supplement
Section~\ref*{sec:projected-solver-proof}
give measurable approximate selections.

\subsection{Measurement weighted PLS (MW-PLS)}
\label{sec:mw-estimator}
Writing $\overline B_T=T^{-1}\sum_tB_t$, the average after channel inversion is
$A_T=\cM^{-1}(\overline B_T)$, so that $\cL(A_T)=\overline Y_T$.
\reviewchange{For every admissible even divisor $k$ of $n$, $\cL$ is
positive definite on Hermitian matrices, by the activity formula in
Supplement Section~\ref*{supp:transfer}. It determines the projection metric}
\begin{equation}
 \|H\|_{\cL}^2:=\langle H,\cL(H)\rangle_\F,
 \qquad
 \rhoMW^\star:=\argmin_{\sigma\in\cD_d}\|\sigma-A_T\|_{\cL}^2.
 \label{eq:mw-weighted-projection}
\end{equation}
This defines the exact \emph{measurement weighted PLS (MW-PLS)} fit. Its minimizer
exists by compactness and is unique by strict convexity on the
affine space of matrices with trace one. It is a measurable function of the data; for
example, uniqueness and compactness give continuous dependence on $\overline Y_T$.
Expanding the square gives the equivalent computation
{\color{black}\begin{equation}
 \begin{aligned}
 \rhoMW^\star&=\argmin_{\sigma\in\cD_d}\mathcal Q_T(\sigma),\\
 \mathcal Q_T(\sigma)&=\tfrac12\langle\sigma,\cL(\sigma)\rangle_\F
                     -\langle \overline Y_T,\sigma\rangle_\F,\\
 \nabla \mathcal Q_T(\sigma)&=\cL(\sigma)-\overline Y_T.
 \end{aligned}
 \label{eq:mw-quadratic}
\end{equation}}%
For iterative computation, put {\color{black}$\mathcal Q_T^*=\min_{\cD_d}\mathcal Q_T$} and define the
first order gap
{\color{black}\[
 g_{\mathrm{MW}}(\sigma):=\max_{\omega\in\cD_d}
 \langle\nabla \mathcal Q_T(\sigma),\sigma-\omega\rangle_\F.
\]}%
Select a measurable feasible fit $\rhoMW$ with nonnegative tolerance
$\delta'_{\rm alg}$ satisfying
{\color{black}\begin{equation}
 \begin{aligned}
 \mathcal Q_T(\rhoMW)&\leq \mathcal Q_T^*+g_{\mathrm{MW}}(\rhoMW),\\
 g_{\mathrm{MW}}(\rhoMW)&\leq\delta'_{\rm alg}\leq\min\{1,d/T\}.
 \end{aligned}
 \label{eq:mw-approximate-fit}
\end{equation}}%
The first inequality follows from convexity, so this certificate makes
$\rhoMW$ an approximate minimizer to accuracy $\delta'_{\rm alg}$ of {\color{black}$\mathcal Q_T$}.
The choice $\delta'_{\rm alg}=0$ recovers $\rhoMW^\star$.
\begin{finalrevision}
The weights are determined by the known measurement channel.
Supplement Section~\ref*{app:mw-pls} gives the equivalent least squares density contrast,
connecting this objective to positive least squares
\cite[Sec.~7]{AcharyaKypraiosGuta2019}. The quadratic formulation avoids
channel inversion.
\end{finalrevision}

For a feasible state $\sigma$, the first order gap
used in \eqref{eq:mw-approximate-fit} has the spectral representation
{\color{black}\begin{equation}
 g_{\mathrm{MW}}(\sigma)=\langle\nabla \mathcal Q_T(\sigma),\sigma\rangle_\F
             -\lambda_{\min}(\nabla \mathcal Q_T(\sigma)).
 \label{eq:mw-gap}
\end{equation}}%
Minimizing a linear functional over density matrices gives
{\color{black}$g_{\mathrm{MW}}(\sigma)=\max_{\omega\in\cD_d}
\langle\nabla \mathcal Q_T(\sigma),\sigma-\omega\rangle_\F$}.
In particular, {\color{black}$g_{\mathrm{MW}}\geq0$} and
{\color{black}$\mathcal Q_T(\sigma)-\min_{\cD_d}\mathcal Q_T\leq g_{\mathrm{MW}}(\sigma)$}.

The statewise guarantees for both fits are stated next. Their different
optimization tolerances account for how each certificate enters its
trace loss bound.

\subsection{Statewise adaptive oracle inequalities}
\label{sec:statewise-oracle}
The fundamental upper bounds are statewise: each estimator controls
trace loss for every density matrix, without a spectral decay assumption.
We first state the OMD oracle.

\begin{theorem}[Statewise adaptive trace loss oracle inequality]
\label{thm:statewise-oracle}
Let $n,T\geq1$, $d=2^n$, and let $k$ be an even divisor of $n$ satisfying
\eqref{eq:headline-block-condition}. Let $\rhoCF$ be a measurable feasible
fit satisfying \eqref{eq:candidate-two-program}.
There is a universal constant $B_0$ such that, for $T\geq2d$,
\begin{equation}
 \E_\rho\|\rhoCF-\rho\|_{\mathrm{tr}}
 \leq4\tau_s(\rho)+B_0s\sqrt{d/T},\qquad 1\leq s\leq d,
 \label{eq:polynomial-expected-oracle}
\end{equation}
simultaneously for every $\rho\in\cD_d$. For every $T\geq1$, the loss
is also bounded by $2$.
\end{theorem}
\begin{proof}
Taking expectations in \eqref{eq:polynomial-forward-oracle}, using the
covariance bound \eqref{eq:forward-covariance-rate} with coefficient
\eqref{eq:audited-covariance-constant} and the fitting tolerance
\eqref{eq:candidate-two-program}, gives
\eqref{eq:polynomial-expected-oracle}. Since both $\rhoCF$ and $\rho$ are
density matrices, their trace norm distance is at most $2$.
\end{proof}

The inequality holds simultaneously for all truncation levels $s$, so its
right side can be minimized over $s$ without supplying a rank,
spectral decay parameter, or eigenbasis to the estimator. This is the
statewise adaptation property.

The key ingredients are statewise oracle geometry and the logarithm free
expected operator norm bound \eqref{eq:forward-covariance-rate} for the raw
score average. The argument in Proposition~\ref{prop:primitive-covariance} uses a
well conditioned covariance representation of effective rank at most
$2d$ and the uniform reference law eighth moment bound in
Theorem~\ref{thm:ck-fourth-projector}, \eqref{eq:ck-u4-34}, valid for arbitrary,
possibly entangled directions. Born reweighting gives the sixth moment
condition \eqref{eq:ck-born-sixth} needed for concentration.

The next proposition gives the weighted PLS counterpart, including its
dependence on optimization error. Its curvature assumption
\eqref{eq:curvature} is verified for the periodic ensemble in
Section~\ref{sec:fourth-moment}.

\phantomsection\label{sec:mw-statistical-result}
\begin{proposition}[Weighted PLS trace loss oracle]
\label{prop:mw-oracle}
Assume \eqref{eq:curvature} with $a>0$. For every density matrix $\rho$,
if a feasible $\sigma$ satisfies
{\color{black}$g_{\mathrm{MW}}(\sigma)\leq\varepsilon$}, then, for every $1\leq s\leq d$,
\begin{equation}
 \begin{aligned}
 \|\sigma-\rho\|_{\mathrm{tr}}
 &\leq4\tau_s(\rho)+\frac{4s}{a}\|\overline Y_T-\cL(\rho)\|_{\op}\\
 &\quad+4\sqrt{\frac{s\varepsilon}{a}}.
 \end{aligned}
 \label{eq:mw-oracle}
\end{equation}
The exact minimizer satisfies the same inequality with $\varepsilon=0$.
For the periodic experiment under \eqref{eq:headline-block-condition},
a measurable feasible rule with {\color{black}$g_{\mathrm{MW}}\leq\varepsilon_T:=\min\{1,d/T\}$}
satisfies, for $T\geq2d$,
\begin{equation}
 \E_\rho\|\sigma-\rho\|_{\mathrm{tr}}
 \leq4\tau_s(\rho)+B_{\rm MW}s\sqrt{d/T},
 \label{eq:mw-expected-oracle}
\end{equation}
with universal $B_{\rm MW}$.
\end{proposition}
\begin{proof}
Section~\ref{sec:mw-oracle-proof} proves the deterministic inequality
\eqref{eq:mw-oracle} from \eqref{eq:mw-basic-inequality} and positivity
\eqref{eq:cone-bound}. The covariance estimate
\eqref{eq:forward-covariance-rate} gives the expected oracle.
\end{proof}

Both estimators have trace loss at most $2$ for every $T$.
We next specialize the statewise bounds of Theorem~\ref{thm:statewise-oracle}
and Proposition~\ref{prop:mw-oracle} to spectral decay classes.

\subsection{Spectral class minimax consequence}
\label{sec:spectral-minimax}
Each oracle is stronger than its spectral class upper bound consequence:
it retains the actual tail $\tau_s(\rho)$ of each state. On
$\cC_{\alpha,d}(L)$, substituting $\tau_s(\rho)\leq Ls^{1-\alpha}$ and
optimizing over $s$ gives a uniform upper rate. Combining that consequence
with the matching lower bound over all randomized nonadaptive single copy
designs yields the following minimax theorem.
Define the spectral class rate
\begin{equation}
 \notationmark{\mathsf{r}}_{\alpha,L}(d,T)=\min\left\{1,
 L^{1/\alpha}\left(\frac dT\right)^{\frac{\alpha-1}{2\alpha}},
 \sqrt{\frac{d^3}{T}}\right\}.
 \label{eq:headline-rate}
\end{equation}

\begin{theorem}[Adaptive minimax tomography over spectral decay classes]
\label{thm:headline-ck}
Let $n,T\geq1$, $d=2^n$, and let $k$ be an even divisor of $n$ satisfying
\eqref{eq:headline-block-condition}. There is one OMD
estimator \eqref{eq:candidate-two-program}, depending only on the known
experiment and $(d,T)$, such that, for every $\alpha>1$ and $L\geq1$,
\begin{equation}
 \begin{aligned}
 c_\alpha\notationmark{\mathsf{r}}_{\alpha,L}(d,T)
 &\leq\mathfrak R_T(\cC_{\alpha,d}(L))\\
 &\leq\mathfrak R_T^{\mathcal E_{\rm per}}(\cC_{\alpha,d}(L))\\
 &\leq\sup_{\rho\in\cC_{\alpha,d}(L)}
           \E_\rho\|\rhoCF-\rho\|_{\mathrm{tr}}\\
 &\leq C\notationmark{\mathsf{r}}_{\alpha,L}(d,T).
 \end{aligned}
 \label{eq:headline-minimax-sandwich}
\end{equation}
Here $C$ is universal and $c_\alpha>0$ depends only on $\alpha$.
\end{theorem}
\begin{proof}
Optimize Theorem~\ref{thm:statewise-oracle} using
$\tau_s(\rho)\leq Ls^{1-\alpha}$, $s=d$, and the loss cap.
Section~\ref{sec:minimax-proof} and Appendix~\ref{sec:upper-rate-optimization}
give the integer rounding and universal upper coefficient
\eqref{eq:audited-upper-constant}. The block condition \eqref{eq:headline-block-condition} gives
$d\geq65536$, so Theorem~\ref{thm:lower}, specifically \eqref{eq:lower-main}, supplies the lower bound.
The remaining risk inequalities follow from admissibility of the
periodic experiment and selected estimator.
\end{proof}

The matching minimax conclusion is for the specified spectral decay
classes $\cC_{\alpha,d}(L)$. The statewise oracle itself remains valid on
all of $\cD_d$; it does not, by itself, establish a matching minimax lower
bound for every other state class.

Every weighted PLS fit satisfying \eqref{eq:mw-approximate-fit} satisfies the
same statistical minimax sandwich
\eqref{eq:headline-minimax-sandwich}, with $\rhoCF$ replaced by $\rhoMW$
and a possibly different universal upper constant; see
Proposition~\ref{prop:mw-oracle}. Its statewise oracle
\eqref{eq:mw-expected-oracle} gives the spectral class upper bound,
while the unrestricted lower bound is estimator independent.
Thus OMD and weighted PLS attain this rate without knowing the rank, eigenbasis,
or spectral decay parameters.

The reconstruction procedures and their dense arithmetic costs are summarized
in Section~\ref{sec:implementation}. Supplement
Section~\ref*{supp:computational-proofs} gives the computational guarantees
and proofs.

\subsection{Comparison of the three estimators}
\label{sec:comparisons}
On the same records and feasible set $\cD_d$, PLS, OMD, and weighted PLS use
Frobenius, operator norm, and measurement weighted quadratic geometries,
respectively, and can return different estimates
\cite{AcharyaKypraiosGuta2019}.

Table~\ref{tab:estimator-comparison} summarizes the available bounds,
with $\notationmark{\mathsf{r}}_{\alpha,L}$ defined in \eqref{eq:headline-rate}.
The OMD and weighted PLS spectral class minimax guarantees are proved here.
The PLS spectral bound follows from Cho and Kim's analysis of inverse channel
estimators \cite{ChoKim2026} and the projection argument in
Appendix~\ref{sec:pls-comparison}, while its bound over the full state space
is already free of logarithmic factors. The logarithm in this available upper bound
does not establish that PLS requires it.

\begin{table*}[!b]
\centering\small
\caption{Expected full trace loss bounds on common periodic records, up to constants.
The PLS bound follows from \cite{ChoKim2026} and Frobenius projection
(Appendix~\ref{sec:pls-comparison}); its logarithm is not proved necessary.}
\label{tab:estimator-comparison}
\begin{tabular}{@{}p{.22\textwidth}p{.50\textwidth}p{.22\textwidth}@{}}
\toprule
Reconstruction & Available risk bound & Reconstruction after averaging\\
\midrule
PLS &
$\min\{1,\ L^{1/\alpha}(d\ln d/T)^{(\alpha-1)/(2\alpha)},\ \sqrt{d^3/T}\}$
& Channel inversion and one spectral projection\\[3pt]
OMD & $\notationmark{\mathsf{r}}_{\alpha,L}(d,T)$
& Convex operator norm fit\\[3pt]
MW-PLS & $\notationmark{\mathsf{r}}_{\alpha,L}(d,T)$
& Convex quadratic fit\\
\bottomrule
\end{tabular}
\end{table*}

The PLS rate over the full state space in Table~\ref{tab:estimator-comparison} follows from the
explicit uniform bound
\begin{equation}
 \sup_{\rho\in\cD_d}\E_\rho\|\rhoPLS-\rho\|_{\mathrm{tr}}
 \leq\min\{2,\,(d+1)\sqrt{d/T}/a_*\},
 \label{eq:pls-full-bound}
\end{equation}
proved in Appendix~\ref{sec:pls-comparison}, with $a_*=32767/32768$. Thus the demonstrated improvement is
in the low rank and intermediate spectral regimes.

For global Clifford measurements ($k=n$), $\cL=\operatorname{Id}$,
so the exact weighted PLS fit $\rhoMW^\star$ equals PLS; the common estimate is also an exact OMD minimizer
(Supplement Lemma~\ref*{lem:supp-global-shortcut}). If $A_T\succeq0$, all three
objectives admit the same exact reconstruction $A_T$. Thus overlap
between their numerical curves is compatible with different objectives.

\begin{samepage}
\subsection{Consequences and interpretation}
\label{sec:consequences}
\begin{corollary}[Rank class minimax rate]
\label{cor:exact-rank}
Under the conditions of Theorem~\ref{thm:statewise-oracle}, for every
$T\geq1$ and $1\leq r\leq d$, put
$\cD_{d,r}=\{\rho\in\cD_d:\rank(\rho)\leq r\}$. Then
\begin{equation}
 \begin{aligned}
 c_{\rm rank}\min\{1,r\sqrt{d/T}\}
 &\leq\mathfrak R_T(\cD_{d,r})\\
 &\leq\sup_{\rho\in\cD_{d,r}}
 \E_\rho\|\rhoCF-\rho\|_{\mathrm{tr}}\\
 &\leq\min\{2,C_{\rm rank}r\sqrt{d/T}\}.
 \end{aligned}
 \label{eq:exact-rank-risk}
\end{equation}
The constants are universal and positive. The same inequalities hold
with $\rhoCF$ replaced by the weighted PLS fit $\rhoMW$ satisfying
\eqref{eq:mw-approximate-fit}, with the same $c_{\rm rank}$ and a possibly
different universal upper constant $C_{\rm rank}^{\rm MW}$.
\end{corollary}
\end{samepage}
\begin{proof}[Proof sketch]
The upper bounds follow from $\tau_r(\rho)=0$ in
\eqref{eq:polynomial-expected-oracle} for OMD and
\eqref{eq:mw-expected-oracle} for weighted PLS, together with the loss cap $2$.
The lower bound follows the chi square/Fano strategy of Lowe and Nayak
\cite[arXiv v3, Lemma~4.1 and Thm.~4.11]{LoweNayak2025}, whose $\Omega(dr^2/\varepsilon^2)$
sample bound covers $1\leq r\leq\lfloor d/3\rfloor$, small
$\varepsilon$, and constant success probability in their finite outcome model.

We give a new, self contained proof of the expected risk lower bound in
the randomized nonadaptive single copy model~\eqref{eq:minimax-risk}.
For $3\leq r\leq\lfloor d/3\rfloor$,
Sections~\ref{sec:lower-packing}--\ref{sec:lower-fano} implement the
above chi square/Fano strategy \cite[Thm.~2.10.1]{CoverThomas2006},
using our hard family with the Grassmann packing of
\cite[Lemma~2, Eq.~(19)]{FlammiaGrossLiuEisert2012}.
Our hard family has rank {\color{black}${\rtail{}}+2$}, so our proof treats $r=1,2$ separately:
rank one uses the pure state hypercube \eqref{eq:pure-cube-states},
Assouad's binary testing argument \cite[Lemmas~1--2]{Yu1997}, and
the quantum testing comparison \cite{FuchsVanDeGraaf1999}, proved here
in \eqref{eq:pure-cube-binary-effect};
rank two follows from $\cD_{d,1}\subseteq\cD_{d,2}$, with a universal
constant loss. The cited theorem itself already includes these two ranks.

For $d/3<r\leq d$, the inclusion
$\cD_{d,\lfloor d/3\rfloor}\subseteq\cD_{d,r}$ and
$\lfloor d/3\rfloor\geq r/6$ extend the expected risk lower bound
up to a universal constant.
Appendix~\ref{app:rank-proof} gives the complete proof for every
$1\leq r\leq d$.
\end{proof}

The rank dependent scaling is already known; PLS attains it with the
uniform POVM \cite[Thm.~2]{GutaKahnKuengTropp2020}.
Earlier spectral thresholding \cite{ButuceaGutaKypraios2015,ButuceaGutaKypraios2016Corrigendum},
trace regression oracles \cite{KoltchinskiiXia2015}, and PLS tail bounds
also address unknown or approximate rank
\cite[Prop.~1 and Thm.~4]{GutaKahnKuengTropp2020}.
Here, under \eqref{eq:headline-block-condition}, OMD and weighted PLS attain
minimax rates over every class $\cC_{\alpha,d}(L)$ with polynomial
spectral tails, using the same periodic Clifford experiment without
rank, eigenbasis, or decay inputs (Theorem~\ref{thm:headline-ck} and Proposition~\ref{prop:mw-oracle}).
\begin{revision}
In particular, Corollary~\ref{cor:exact-rank} removes the logarithmic factor
in Cho and Kim's low rank upper bound \cite[Thm.~1]{ChoKim2026}
by using OMD or weighted PLS under the sufficient block size condition
\eqref{eq:headline-block-condition}.
\end{revision}

\section{Statistical Upper Bounds for OMD and Weighted PLS}
\label{sec:oracle}
The proof separates into a common trace loss conversion, a common
covariance bound, and the two fitting certificates. Section~\ref{sec:fourth-moment}
verifies the circuit conditions; Appendix~\ref{sec:upper-rate-optimization}
optimizes both oracles over the spectral tail.

\subsection{Deterministic curvature and positivity}
For traceless Hermitian $H$, the previously defined map satisfies
$\cL(H)=(d+1)\cM(H)$. Suppose
\begin{equation}
 \langle H,\cL(H)\rangle_\F\geq a\|H\|_\F^2
 \quad\text{for all }H=H^*,\ \tr H=0,
 \label{eq:curvature}
\end{equation}
where $a>0$. For the periodic ensemble under
\eqref{eq:headline-block-condition}, Proposition~\ref{prop:periodic-low-order} in
Section~\ref{sec:periodic-curvature} provides such a constant, with
$a\geq a_*:=32767/32768$.
The following inequality uses the positive semidefinite
geometry underlying low rank tomography and matrix oracle bounds
\cite{GutaKahnKuengTropp2020,FlammiaGrossLiuEisert2012,KoltchinskiiXia2015}.

\begin{lemma}[Shared trace loss conversion]
\label{lem:trace-conversion}
For $\rho,\sigma\in\cD_d$, write $\Delta=\sigma-\rho$.
If $a>0$, $e,\varepsilon\geq0$, and
\[
 a\|\Delta\|_\F^2\leq e\|\Delta\|_{\mathrm{tr}}+\varepsilon,
\]
then, for every $1\leq s\leq d$,
\[
 \|\Delta\|_{\mathrm{tr}}
 \leq4\tau_s(\rho)+\frac{4s}{a}e
                  +4\sqrt{\frac{s\varepsilon}{a}}.
\]
In particular, under \eqref{eq:curvature},
\begin{equation}
 \|\Delta\|_{\mathrm{tr}}
 \leq4\tau_s(\rho)+\frac{4s}{a}\|\cL(\Delta)\|_{\op}.
 \label{eq:trace-conversion}
\end{equation}
\end{lemma}
\begin{proof}
Let {\color{black}$\rho={\rho_{\leq s}}+{\rho_{>s}}$}, where {\color{black}${\rho_{\leq s}}$} is the leading rank $s$ spectral part and
{\color{black}${\rho_{>s}}\succeq0$} has trace $\tau_s(\rho)$. Let $E_-$ be the negative spectral
projection of {\color{black}$\sigma-{\rho_{\leq s}}$}. Its rank is at most $s$: on its range the
compression of {\color{black}${\rho_{\leq s}}=\sigma-(\sigma-{\rho_{\leq s}})$} is positive definite, so that
range cannot have dimension greater than {\color{black}$\rank {\rho_{\leq s}}$}.
The negative part identity and $\tr\Delta=0$ give
{\color{black}\begin{align*}
 \|\Delta\|_{\mathrm{tr}}&\leq\|\sigma-{\rho_{\leq s}}\|_{\mathrm{tr}}+\tr {\rho_{>s}}\\
  &=2\tau_s(\rho)-2\tr(E_-(\sigma-{\rho_{\leq s}}))\\
  &\leq2\tau_s(\rho)-2\tr(E_-\Delta)\\
  &\leq2\tau_s(\rho)+2\sqrt{s}\|\Delta\|_\F.
\end{align*}}%
The third line uses {\color{black}$\tr(E_-{\rho_{>s}})\geq0$}; the last uses
$\|E_-\|_\F=\sqrt{\rank E_-}\leq\sqrt{s}$.
Thus
\begin{equation}
 \|\Delta\|_{\mathrm{tr}}
 \leq2\tau_s(\rho)+2\sqrt{s}\|\Delta\|_\F.
 \label{eq:cone-bound}
\end{equation}
Substitute the assumed Frobenius bound into \eqref{eq:cone-bound} and use
the subadditivity of the square root. The term involving trace loss is
absorbed by
\[
 2\sqrt{\frac{se}{a}\|\Delta\|_{\mathrm{tr}}}
 \leq\frac12\|\Delta\|_{\mathrm{tr}}+\frac{2se}{a}.
\]
For \eqref{eq:trace-conversion}, curvature and Schatten duality give
\[
 \begin{aligned}
 a\|\Delta\|_\F^2
 &\leq\langle\Delta,\cL(\Delta)\rangle_\F\\
 &\leq\|\cL(\Delta)\|_{\op}\|\Delta\|_{\mathrm{tr}}.
 \end{aligned}
\]
Apply the first conclusion with $e=\|\cL(\Delta)\|_{\op}$ and
$\varepsilon=0$.
\end{proof}

\subsection{Primitive ensemble conditions for covariance control}
Fourth overlap moments and covariance conditioning are accessible through
standard Clifford design methods
\cite{Zhu2017,SchusterHaferkampHuang2025}.
Our circuit contribution is a uniform eighth overlap moment bound for
the periodic ensemble, including arbitrary entangled test directions.
This is a fourth projector moment, obtained by controlling the periodic
contraction of known local four copy formulas
\cite{GrossNezamiWalter2021,HelsenWalter2023}.
Under \eqref{eq:headline-block-condition},
Theorem~\ref{thm:ck-fourth-projector}, \eqref{eq:ck-u4-34}, verifies
\eqref{eq:primitive-fourth} with {\color{black}${\momconst{}}=34$};
Proposition~\ref{prop:periodic-low-order}, \eqref{eq:forward-covariance-conditioning},
verifies \eqref{eq:primitive-conditioning}.
The following proposition converts these inputs, via Born reweighting
and the sixth moment bound \eqref{eq:ck-born-sixth}, into the logarithm free
covariance estimate in operator norm \eqref{eq:forward-covariance-rate} using
\cite[Theorem~2, arXiv v3]{AbdallaZhivotovskiy2023}.
Its $O(\sqrt{d/T})$ expected operator error for $T\geq c_0d$ supplies
the stochastic term in both the OMD and weighted PLS trace loss oracles.

\begin{proposition}[From reference moments to Born covariance]
\label{prop:primitive-covariance}
Let $B=|\psi\rangle\!\langle\psi|$ have reference law $\nu$ as in
\eqref{eq:rank-one-statistical-law}. Suppose constants {\color{black}${\momconst{}},{\covlo{}},{\covhi{}}>0$} satisfy
{\color{black}\begin{align}
 \sup_{\langle u|u\rangle=1}\E_\nu|\langle u|\psi\rangle|^8
 &\leq\frac{{\momconst{}}}{d(d+1)(d+2)(d+3)},
 \label{eq:primitive-fourth}\\
 {\covlo{}}I\preceq(d+1)\E_\rho B&\preceq {\covhi{}}I\quad(\rho\in\cD_d).
 \label{eq:primitive-conditioning}
\end{align}}%
There are constants {\color{black}$c_0,{\covrate{}}$}, depending only on {\color{black}$({\momconst{}},{\covlo{}},{\covhi{}})$}, such that
independent Born observations satisfy
{\color{black}\begin{equation}
 T\geq c_0d\quad\Longrightarrow\quad
 \E_\rho\|\overline Y_T-\cL(\rho)\|_{\op}\leq {\covrate{}}\sqrt{d/T}.
 \label{eq:forward-covariance-rate}
\end{equation}}
\end{proposition}
\begin{proof}
Independently choose {\color{black}${\phi_{\mathrm{ph}}}$} uniformly from the four phases
$0,\pi/2,\pi,3\pi/2$, and define the complex random vector
{\color{black}\begin{equation}
 \begin{aligned}
 |X\rangle&=\sqrt{d+1}\,e^{i{\phi_{\mathrm{ph}}}}|\psi\rangle,\\
 \Sigma_{\C}&=\E_\rho|X\rangle\!\langle X|=I+\cL(\rho).
 \end{aligned}
 \label{eq:ck-phase-lift}
\end{equation}}%
Then $\E|X\rangle=0$ and $\langle X|X\rangle=d+1$. For a unit direction $|v\rangle$,
convexity in a spectral decomposition of $\rho$ and
\eqref{eq:primitive-fourth} bound both $\E_\nu\tr(\rho B)^4$ and
$\E_\nu|\langle v|\psi\rangle|^8$ by
{\color{black}${\momconst{}}/[d(d+1)(d+2)(d+3)]$}. The Born change of measure is
\begin{equation}
 \E_\rho f(B)=d\E_\nu[\tr(\rho B)f(B)].
 \label{eq:ck-born-change-measure}
\end{equation}
Apply H\"older's inequality with exponent $4$ to $\tr(\rho B)$ and
exponent $4/3$ to $|\langle v|\psi\rangle|^6$.
The fourth and eighth moment bounds above then give
{\color{black}\begin{equation}
 \begin{aligned}
 &\E_\rho|\langle v|X\rangle|^6
 =d(d+1)^3\E_\nu[\tr(\rho B)|\langle v|\psi\rangle|^6]\\
 &\quad\leq d(d+1)^3
       \bigl(\E_\nu\tr(\rho B)^4\bigr)^{1/4}\\[-2pt]
 &\qquad\quad\times
       \bigl(\E_\nu|\langle v|\psi\rangle|^8\bigr)^{3/4}\\
 &\quad\leq {\momconst{}}\frac{(d+1)^2}{(d+2)(d+3)}\leq {\momconst{}}.
 \end{aligned}
 \label{eq:ck-born-sixth}
\end{equation}}%
The phases leave every $|X\rangle\!\langle X|=(d+1)B$ unchanged; they are a proof device.

Let $Z\in\R^{2d}$ stack the real and imaginary coordinates of $|X\rangle$,
and write {\color{black}$\mathcal{R}$} for matrix realification. Averaging over the four phases yields
{\color{black}\begin{equation}
 \Sigma_{\R}=\E ZZ^{\mathsf T}=\tfrac12\mathcal{R}(\Sigma_{\C}).
 \label{eq:forward-real-covariance}
\end{equation}}%
For every real direction $w$,
{\color{black}\begin{align}
 \E|\langle w,Z\rangle|^2&\geq({\covlo{}}/2)\|w\|^2,\\
 \E|\langle w,Z\rangle|^6&\leq({\momconst{}}/2)\|w\|^6.
 \label{eq:forward-real-sixth}
\end{align}}%
For the second inequality, the phase average of the sixth power of the real
part of {\color{black}${a_{\mathrm{re}}}+i{a_{\mathrm{im}}}$} is {\color{black}$({a_{\mathrm{re}}}^6+{a_{\mathrm{im}}}^6)/2\leq|{a_{\mathrm{re}}}+i{a_{\mathrm{im}}}|^6/2$}.
Thus the sixth power $L_6$ to $L_2$ ratio is at most {\color{black}$4{\momconst{}}/{\covlo{}}^3$};
interpolation gives the required intermediate moments. Also
{\color{black}$\|\Sigma_{\R}\|_{\op}\leq {\covhi{}}/2$} and its effective rank
$r_{\rm eff}=\tr(\Sigma_{\R})/\|\Sigma_{\R}\|_{\op}$ is at most $2d$.

For independent real observations $Z_1,\ldots,Z_T$,
\cite[Theorem~2, arXiv v3]{AbdallaZhivotovskiy2023}, applied at $p=6$,
and the truncation check below give
\begin{equation}
 \begin{aligned}
 &\E_\rho\left\|\frac1T\sum_{t=1}^T Z_tZ_t^{\mathsf T}
                  -\Sigma_{\R}\right\|_{\op}
 \leq C_{\rm BY}\|\Sigma_{\R}\|_{\op}\sqrt{r_{\rm eff}/T},\\
 &\qquad T\geq\max\{2d,c_{\rm BY}r_{\rm eff}\}.
 \end{aligned}
 \label{eq:forward-real-sample-covariance}
\end{equation}
The constants depend only on the moment ratio. The theorem's radial
cutoff is $(T\tr(\Sigma_{\R})\|\Sigma_{\R}\|_{\op})^{1/4}$.
Here $\tr(\Sigma_{\R})=\|Z\|^2=d+1$, and truncation is inactive since
\begin{equation}
 (d+1)^2\leq T(d+1)\|\Sigma_{\R}\|_{\op}\quad(T\geq2d),
 \label{eq:forward-truncation-inactive}
\end{equation}
The last inequality follows because the operator norm is at least the trace
divided by $2d$. Finally, let {\color{black}${\realj{}}$} be the real linear map on $\R^{2d}$
corresponding to multiplication by $i$ under the chosen realification. Then
{\color{black}\begin{equation}
 \mathcal{R}(|X\rangle\!\langle X|)=ZZ^{\mathsf T}+{\realj{}}ZZ^{\mathsf T}{\realj{}}^{\mathsf T}.
 \label{eq:forward-realification-identity}
\end{equation}}%
Realification preserves operator norm, so the complex error is at most
twice the real covariance error. Taking
$c_0=2\max\{1,c_{\rm BY}\}$ and {\color{black}${\covrate{}}=\sqrt2 {\covhi{}} C_{\rm BY}$} proves the claim.
\end{proof}

\subsection{OMD: from the fitting certificate to the oracle}
\label{sec:oracle-proof}
Feasibility of the true state and \eqref{eq:candidate-two-program} give
\begin{equation}
 \begin{aligned}
 \|\cL(\rhoCF-\rho)\|_{\op}
 &\leq f_T(\rhoCF)+f_T(\rho)\\
 &\leq2\|\overline Y_T-\cL(\rho)\|_{\op}+\delta_{\rm alg}.
 \end{aligned}
 \label{eq:forward-fit-triangle}
\end{equation}
Apply Lemma~\ref{lem:trace-conversion} with the periodic curvature bound
$a\geq a_*$ from Proposition~\ref{prop:periodic-low-order}:
\begin{equation}
 \begin{aligned}
 \|\rhoCF-\rho\|_{\mathrm{tr}}
 &\leq4\tau_s(\rho)\\
 &\quad+\frac{8}{a_*}s\|\overline Y_T-\cL(\rho)\|_{\op}
       +\frac{4}{a_*}s\delta_{\rm alg}.
 \end{aligned}
 \label{eq:polynomial-forward-oracle}
\end{equation}
The rank index $s$ appears only in the analysis.

\begin{proof}[Proof of Theorem~\ref{thm:statewise-oracle}]
Proposition~\ref{prop:periodic-low-order} and Theorem~\ref{thm:ck-fourth-projector} verify the
primitive conditions of Proposition~\ref{prop:primitive-covariance}.
Appendix~\ref{app:quantitative} specializes the resulting fixed norm
covariance bound to
$\E\|\overline Y_T-\cL(\rho)\|_{\op}\leq \covstar{}\sqrt{d/T}$ for $T\geq2d$.
Taking expectations in \eqref{eq:polynomial-forward-oracle}, with
$\delta_{\rm alg}\leq\sqrt{d/T}$, gives
\eqref{eq:polynomial-expected-oracle} and $B_0=(8\covstar{}+4)/a_*$.
The trace distance between two feasible states gives the loss cap.
\end{proof}

More generally, the same argument under \eqref{eq:curvature},
\eqref{eq:primitive-fourth}, and \eqref{eq:primitive-conditioning}
gives the statewise bound
{\color{black}\[
 \E_\rho\|\rhoCF-\rho\|_{\mathrm{tr}}
 \leq C_{\mathrm{OMD}}\min\bigl\{1,\inf_{1\leq s\leq d}
 [\tau_s(\rho)+s\sqrt{d/T}]\bigr\},
\]}%
with {\color{black}$C_{\mathrm{OMD}}$} depending only on the primitive parameters. For the finitely
many sample sizes below the covariance theorem's threshold, use the loss
cap and enlarge {\color{black}$C_{\mathrm{OMD}}$}.

\subsection{Weighted PLS: from the first order certificate to the oracle}
\label{sec:mw-oracle-proof}
Only the certificate changes for the quadratic fit.

\begin{proof}[Proof of Proposition~\ref{prop:mw-oracle}]
Write $\Delta=\sigma-\rho$. Since $\rho$ is feasible, the gap definition yields
{\color{black}\begin{equation}
 \begin{aligned}
 a\|\Delta\|_\F^2&\leq\langle\Delta,\cL(\Delta)\rangle_\F\\
 &=\langle\overline Y_T-\cL(\rho),\Delta\rangle_\F
      +\langle\nabla \mathcal Q_T(\sigma),\Delta\rangle_\F\\
 &\leq\|\overline Y_T-\cL(\rho)\|_{\op}\|\Delta\|_{\mathrm{tr}}
      +\varepsilon.
 \end{aligned}
 \label{eq:mw-basic-inequality}
\end{equation}}%
For the exact minimizer, first order optimality gives {\color{black}$g_{\mathrm{MW}}=0$}.
Apply Lemma~\ref{lem:trace-conversion} to \eqref{eq:mw-basic-inequality}
to obtain \eqref{eq:mw-oracle}.
The periodic bounds give $a\geq a_*$ and
$\E\|\overline Y_T-\cL(\rho)\|_{\op}\leq \covstar{}\sqrt{d/T}$ for $T\geq2d$,
so one can take
$B_{\rm MW}=4\covstar{}/a_*+4/\sqrt{a_*}$, using $\sqrt{s}\leq s$.
\end{proof}

For a feasible $\sigma$, the first order gap certificate
{\color{black}$g_{\mathrm{MW}}(\sigma)\leq\varepsilon$} and the primal objective gap certificate
{\color{black}$\mathcal Q_T(\sigma)-\mathcal Q_T^*\leq\varepsilon$} both yield a weighted PLS trace loss oracle.
The former gives \eqref{eq:mw-oracle}. For the latter, expansion at $\rho$ gives
\[
 \frac a2\|\sigma-\rho\|_\F^2
 \leq\|\overline Y_T-\cL(\rho)\|_{\op}\|\sigma-\rho\|_{\mathrm{tr}}
     +\varepsilon.
\]
The same lemma, with $a/2$ in place of $a$, yields
\begin{equation}
 \begin{aligned}
 \|\sigma-\rho\|_{\mathrm{tr}}
 &\leq4\tau_s(\rho)+\frac{8s}{a}\|\overline Y_T-\cL(\rho)\|_{\op}\\
 &\quad+4\sqrt{\frac{2s\varepsilon}{a}}.
 \end{aligned}
 \label{eq:mw-primal-oracle}
\end{equation}
Under \eqref{eq:headline-block-condition}, either certificate with
$\varepsilon\leq\min\{1,d/T\}$ therefore yields the expected oracle
\eqref{eq:mw-expected-oracle} for $T\geq2d$, with possibly different
universal constants.

\subsection{Spectral optimization and the minimax theorem}
\label{sec:minimax-proof}
\begin{proof}[Proof of Theorem~\ref{thm:headline-ck}]
Apply the optimization lemma in Appendix~\ref{sec:upper-rate-optimization}
with $\eta=(B_0/4)\sqrt{d/T}$ to Theorem~\ref{thm:statewise-oracle}; the loss
cap covers $T<2d$. The same substitution with $B_{\rm MW}$ treats
Proposition~\ref{prop:mw-oracle}. The block condition
\eqref{eq:headline-block-condition} gives $d\geq65536$, so
Theorem~\ref{thm:lower} supplies the common lower bound. Admissibility of the
periodic experiment gives the intervening infima. Neither experiment nor
estimator uses the class parameters.
\end{proof}

One may replace $L$ by $\min\{L,d^{\alpha-1}\}$ without changing either
the class or the rate. In the normalized range, the transitions are
\begin{equation}
 T_0=dL^{2/(\alpha-1)},\qquad T_1=d^{2\alpha+1}/L^2.
 \label{eq:exact-change-points}
\end{equation}
The constant, spectral, and full dimensional branches apply below $T_0$,
between $T_0$ and $T_1$, and above $T_1$. These are changes in the bound;
the OMD estimator does not switch at either point.

\section{Periodic Circuit Estimates}
\label{sec:fourth-moment}
We verify the ensemble conditions used in Section~\ref{sec:oracle}.
The reference law samples a fresh circuit and a uniformly random outcome;
the observed law is its Born reweighting
\eqref{eq:rank-one-statistical-law}. The principal circuit estimate controls
the fourth projector moment for every test state. Low order design bounds
then supply curvature, covariance conditioning, and scalar variance.

\subsection{Uniform fourth moment of projectors}
\label{sec:fourth-moment-proof}
We average the random projector itself before choosing a test state.
Lemma~\ref{lem:ck-four-moment-expansion} expresses its fourth tensor moment
as a finite sum of deterministic operators. Lemma~\ref{lem:ck-four-copy-boundary}
bounds their contractions against any test state; the theorem then reduces
to a scalar estimate of the coefficients.

\begin{finalrevision}
\emph{Replica and site ordering.}
For an ordered block $D=\{j_1,\ldots,j_\ell\}$, write
$\mathbf b_D^{(a)}=(b_{j_1}^{(a)},\ldots,b_{j_\ell}^{(a)})$ for the
bits in replica $a$, and $\mathbf b_{j_i}=(b_{j_i}^{(1)},\ldots,b_{j_i}^{(4)})$
for the four bits at site $j_i$. The unitary permutation
\[
 J_D\left(\bigotimes_{a=1}^4|\mathbf b_D^{(a)}\rangle\right)
 =\bigotimes_{i=1}^{\ell}|\mathbf b_{j_i}\rangle,
 \qquad \oparg_{\rm site}=J_D\oparg_{\rm rep}J_D^*
\]
changes replica order to site order and extends linearly to entangled
states. The local twirl, $B^{\otimes4}$, and $|u\rangle^{\otimes4}$
are initially written in replica order. The sector product
\eqref{eq:ck-sector-operator}, Gram identity \eqref{eq:ck-local-gram}, and
overlap factorization \eqref{eq:ck-four-overlap} use site order.
Supplement Section~\ref*{supp:replica-order} gives coordinate examples.
\end{finalrevision}

\begin{revision}
For a block of at least three qubits, the local Clifford four copy
commutant has a basis of 30 sector operators
\cite[Theorem~4.3]{GrossNezamiWalter2021}.
Their labels form the set $\Sigma_4$ of binary subspaces
$\Xi\subseteq\mathbb F_2^4\oplus\mathbb F_2^4$ defined in
Appendix~\ref{app:local-algebra}. For $\Xi\in\Sigma_4$, let $R_D(\Xi)$
denote the corresponding four replica sector operator on the physical
block $D$ of $\ell$ qubits. In column/site order, it is defined by
{\color{black}\begin{equation}
 R_D(\Xi)=\left(\sum_{({\mathbf b},{\mathbf c})\in\Xi}|{\mathbf b}\rangle\langle {\mathbf c}|\right)^{\otimes\ell},
 \label{eq:ck-sector-operator}
\end{equation}}%
where {\color{black}$|{\mathbf b}\rangle=|{\mathbf b}_1{\mathbf b}_2{\mathbf b}_3{\mathbf b}_4\rangle$} for
{\color{black}${\mathbf b}=({\mathbf b}_1,{\mathbf b}_2,{\mathbf b}_3,{\mathbf b}_4)\in\mathbb F_2^4$}. Its row/replica form is
$R_D^{\rm rep}(\Xi)=J_D^*R_D(\Xi)J_D$. We use site order by default
in both the article and Supplement. To display each change of coordinates,
write
{\color{black}\[
 {\oparg{}}^{[4]}:=J_D{\oparg{}}^{\otimes4}J_D^*,\qquad
 |v^{[4]}\rangle:=J_D|v\rangle^{\otimes4}
\]}%
for an operator {\color{black}${\oparg{}}$} or vector $v$ on $D$. The underlying support is
specified by {\color{black}${\oparg{}}$} or $v$; on the whole system use $J_n$.
Thus $B^{[4]}=J_nB^{\otimes4}J_n^*$, and contractions with sector
operators use $|u^{[4]}\rangle$. Tensor products over named physical supports
are embedded in this canonical site order, with the fixed permutation
for regrouping disjoint supports understood. Write {\color{black}$R_{\dloc{}}(\Xi)$} when only the block
dimension {\color{black}${\dloc{}}=2^\ell$} matters. The Gram matrix and, for {\color{black}${\dloc{}}>4$}, its inverse are
{\color{black}\begin{equation}
 \begin{aligned}
 G_{\dloc{}}(\Xi,\Xi')&:=\tr\bigl(R_{\dloc{}}(\Xi)^*R_{\dloc{}}(\Xi')\bigr)
          ={\dloc{}}^{\dim(\Xi\cap\Xi')},\\
 W_{\dloc{}}&:=G_{\dloc{}}^{-1}\qquad({\dloc{}}>4).
 \end{aligned}
 \label{eq:ck-local-gram}
\end{equation}}%
The inverse gives the coefficients of the local Clifford averaging
formula \cite[Eq.~(18), arXiv~v2]{HelsenWalter2023},
written out in \eqref{eq:ck-local-average}; its explicit entries are in
\eqref{eq:ck-four-copy-inverse}. Its signed row sum, independent of $\Xi$, is
{\color{black}\begin{equation}
\begin{aligned}
 \sgnrow_{\dloc}&:=\sum_{\Xi'\in\Sigma_4}W_{\dloc}(\Xi,\Xi')\\
 &=\frac1{\dloc(\dloc+1)(\dloc+2)(\dloc+4)}.
 \end{aligned}
 \label{eq:known-signed-row}
\end{equation}}%
This is the scalar coefficient appearing in the fourth moment expansion
of a Clifford randomized stabilizer state in
\cite[Eq.~(19), arXiv~v2]{HelsenWalter2023}.

For even $k\geq4$ dividing $n$, use the blocks $A_j,A'_j$ and layers $U_1,U_2$ from
\eqref{eq:periodic-supports}--\eqref{eq:periodic-unitary}, and put
\begin{equation}
 q=2^{k/2},\qquad Q=q^2=2^k,\qquad m=n/k,\qquad d=q^{2m}.
 \label{eq:ck-four-dimensions}
\end{equation}
Here $G_Q$ is {\color{black}$G_{\dloc{}}$} in \eqref{eq:ck-local-gram} evaluated at
{\color{black}${\dloc{}}=Q=2^k$} for a full $k$ qubit block; $G_q$ is the same family at
{\color{black}${\dloc{}}=q=2^{k/2}$} for a $k/2$ qubit half block. Both are $30\times30$
matrices indexed by $\Sigma_4$; their subscripts specify Hilbert space
dimensions. For $m\geq2$, the half blocks are the two disjoint overlaps
$A'_j\cap A_j$ and $A'_j\cap A_{j+1}$, with
$|A'_j\cap A_j|=|A'_j\cap A_{j+1}|=k/2$.
Each overlap contributes a $G_q$ factor. The local Clifford average acts
on all of $A'_j$, so it uses $W_Q=G_Q^{-1}$; $G_q$ is not inverted.
For $m=1$, use the two halves of the single block, as in the proof below.
Let {\color{black}${\xi},{\upsilon},\zeta\in\Sigma_4^m$}:
{\color{black}${\xi}_j$} labels the sector on $A_j$ after averaging $U_2$;
{\color{black}${\upsilon}_j$} is the summation label for testing the next local average
on $A'_j$; and $\zeta_j$ is its output sector.
All block indices run from $0$ to $m-1$ and are read modulo $m$.
For each $\zeta$, define the real coefficient
{\color{black}\begin{equation}
 \sectorcoef{}_\zeta:=\sgnrow{}_Q^m\sum_{{\xi},{\upsilon}\in\Sigma_4^m}
 \prod_{j=0}^{m-1}
 G_q({\upsilon}_j,{\xi}_j)G_q({\upsilon}_j,{\xi}_{j+1})
 W_Q({\upsilon}_j,\zeta_j).
 \label{eq:ck-four-expansion-coefficients}
\end{equation}}%
These coefficients depend only on the circuit geometry and block size;
no positivity is assumed. For the disjoint shifted partition
$A'=(A'_0,\ldots,A'_{m-1})$ and a unit test vector $|u\rangle$, define
\begin{equation}
 c_{u,A'}(\zeta):=\langle u^{[4]}|
 \left(\bigotimes_j R_{A'_j}(\zeta_j)\right)|u^{[4]}\rangle.
 \label{eq:ck-four-boundary-coefficient}
\end{equation}
This definition also applies to any disjoint block partition.
Lemma~\ref{lem:ck-four-copy-boundary} below proves
$|c_{u,A'}(\zeta)|\leq1$ for every unit $|u\rangle$, including states
entangled across the blocks.

\begin{lemma}[Exact fourth moment expansion]
\label{lem:ck-four-moment-expansion}
For the periodic reference law $\nu$ with even $k\geq4$ dividing $n$,
let $b=(b_0,\ldots,b_{n-1})\in\{0,1\}^n$ be uniform and independent of $U$.
Write $|b\rangle=|b_0\cdots b_{n-1}\rangle$ for the corresponding
computational basis state, and set
$|\psi\rangle=U^*|b\rangle=U_1^*U_2^*|b\rangle$.
With $B=|\psi\rangle\langle\psi|$, the following operator identity holds:
{\color{black}\begin{equation}
 \E_\nu B^{[4]}
 =\sum_{\zeta\in\Sigma_4^m}\sectorcoef{}_\zeta
   \bigotimes_{j=0}^{m-1}R_{A'_j}(\zeta_j).
 \label{eq:ck-four-operator-expansion}
\end{equation}}%
Consequently, for every unit vector $|u\rangle$,
{\color{black}\begin{equation}
 \E_\nu|\langle u|\psi\rangle|^8
 =\sum_{\zeta\in\Sigma_4^m}\sectorcoef{}_\zeta c_{u,A'}(\zeta).
 \label{eq:ck-four-exact-expansion}
\end{equation}}
\end{lemma}
\begin{proof}
Condition on $b$ and average the independent $U_2$ gates first.
Let {\color{black}${\cliffgate{}}_j$} be the local $k$ qubit matrix of the gate {\color{black}$\cliffgate{}^{(2)}_{A_j}$}
in \eqref{eq:periodic-unitary}. After grouping by the disjoint supports
$A_j$, we have
{\color{black}\[
 U_2=\bigotimes_{j=0}^{m-1}{\cliffgate{}}_j,\qquad
 |b\rangle=\bigotimes_{j=0}^{m-1}|b_{A_j}\rangle,
\]}%
where the {\color{black}${\cliffgate{}}_j$} are independent uniform Clifford gates and
$|b_{A_j}\rangle$ is the computational basis state on $A_j$.
For each block,
\[
 \langle b_{A_j}^{[4]}|R_Q(\Xi)^*|b_{A_j}^{[4]}\rangle=1
 \qquad\text{for every }\Xi\in\Sigma_4.
\]
Indeed, copying the same basis string four times gives, at each physical
site, the column state $|0000\rangle$ or $|1111\rangle$.
Both $(\mathbf0,\mathbf0)$ and $(\mathbf1,\mathbf1)$ belong to every
$\Xi$: the former by linearity, and the latter by definition.
Thus \eqref{eq:ck-sector-operator} gives a diagonal entry of one at
each site, and their product is one (also for the adjoint).
This is a statement about these diagonal entries; $R_Q(\Xi)$ need not be
unitary or fix $|b_{A_j}^{[4]}\rangle$.
For each block, the fourth moment formula in
\cite[Eq.~(19), arXiv~v2]{HelsenWalter2023} therefore gives
{\color{black}\[
 \E_{{\cliffgate{}}_j}({\cliffgate{}}_j^*|b_{A_j}\rangle\langle b_{A_j}|{\cliffgate{}}_j)^{[4]}
 =\sgnrow{}_Q\sum_{{\xi}_j\in\Sigma_4}R_{A_j}({\xi}_j).
\]}%
Taking four copies and regrouping the tensor factors by physical block,
independence yields
{\color{black}\[
 \begin{aligned}
 &\E_{U_2}(U_2^*|b\rangle\langle b|U_2)^{[4]}\\
 &\quad=\bigotimes_{j=0}^{m-1}
   \E_{{\cliffgate{}}_j}({\cliffgate{}}_j^*|b_{A_j}\rangle\langle b_{A_j}|{\cliffgate{}}_j)^{[4]}\\
 &\quad=\bigotimes_{j=0}^{m-1}
   \left(\sgnrow{}_Q\sum_{{\xi}_j\in\Sigma_4}R_{A_j}({\xi}_j)\right).
 \end{aligned}
\]}%
Expanding this tensor product of sums and writing
{\color{black}${\xi}=({\xi}_0,\ldots,{\xi}_{m-1})\in\Sigma_4^m$} gives
{\color{black}\begin{equation}
 \E_{U_2}(U_2^*|b\rangle\langle b|U_2)^{[4]}
 =\sgnrow{}_Q^m\sum_{{\xi}\in\Sigma_4^m}\bigotimes_j R_{A_j}({\xi}_j).
 \label{eq:ck-four-first-layer}
\end{equation}}%
Next average the independent $U_1$ gates using the inverse Gram formula
\eqref{eq:ck-local-average}. Fix {\color{black}${\xi}$} in \eqref{eq:ck-four-first-layer}
and write {\color{black}$\sectorprod{}_\xi:=\bigotimes_{r=0}^{m-1}R_{A_r}({\xi}_r)$}.
In column/site order, \eqref{eq:ck-sector-operator} lets us regroup its
site factors by the shifted blocks:
{\color{black}$\sectorprod{}_\xi=\bigotimes_j(\sectorprod{}_\xi|_{A'_j})$}, where {\color{black}$\sectorprod{}_\xi|_{A'_j}$} denotes
the tensor factor on four replicas of $A'_j$.
For $m\geq2$, fix $j$. The block $A'_j$ consists of the two disjoint
halves $A'_j\cap A_j$ and $A'_j\cap A_{j+1}$, each of Hilbert space
dimension $q$.
The input has sector {\color{black}${\xi}_j$} on the first half and {\color{black}${\xi}_{j+1}$} on
the second, whereas the testing operator {\color{black}$R_{A'_j}({\upsilon}_j)$} has
the same sector {\color{black}${\upsilon}_j$} on both. Thus, by trace multiplicativity
and the Gram definition \eqref{eq:ck-local-gram},
{\color{black}\[
 \begin{aligned}
 \sectorprod{}_\xi|_{A'_j}
 &=R_{A'_j\cap A_j}({\xi}_j)
   \otimes R_{A'_j\cap A_{j+1}}({\xi}_{j+1}),\\
 R_{A'_j}({\upsilon}_j)
 &=R_{A'_j\cap A_j}({\upsilon}_j)
   \otimes R_{A'_j\cap A_{j+1}}({\upsilon}_j),\\[2pt]
 \tr\bigl[R_{A'_j}({\upsilon}_j)^*\sectorprod{}_\xi|_{A'_j}\bigr]
 &=\tr\bigl[R_{A'_j\cap A_j}({\upsilon}_j)^*
             R_{A'_j\cap A_j}({\xi}_j)\bigr]\\
 &{}\times\tr\bigl[R_{A'_j\cap A_{j+1}}({\upsilon}_j)^*
                      R_{A'_j\cap A_{j+1}}({\xi}_{j+1})\bigr]\\
 &=G_q({\upsilon}_j,{\xi}_j)G_q({\upsilon}_j,{\xi}_{j+1}).
 \end{aligned}
\]}%
This is the local Hilbert--Schmidt overlap required by
\eqref{eq:ck-local-average} on $A'_j$. Multiplying these overlaps
over the disjoint shifted blocks, using trace multiplicativity and then
substituting the local identity above, gives
{\color{black}\begin{equation}
 \begin{aligned}
 &\tr\left[\left(\bigotimes_jR_{A'_j}({\upsilon}_j)\right)^*
                 \left(\bigotimes_jR_{A_j}({\xi}_j)\right)\right]\\
 &\hspace{8mm}=\tr\left[\bigotimes_j
   \left(R_{A'_j}({\upsilon}_j)^*(\sectorprod{}_\xi|_{A'_j})\right)\right]\\
 &\hspace{8mm}=\prod_j\tr\bigl[R_{A'_j}({\upsilon}_j)^*
                                     (\sectorprod{}_\xi|_{A'_j})\bigr]\\
 &\hspace{8mm}=\prod_jG_q({\upsilon}_j,{\xi}_j)G_q({\upsilon}_j,{\xi}_{j+1}).
 \end{aligned}
 \label{eq:ck-four-overlap}
\end{equation}}%
For each $j$, \eqref{eq:ck-local-average} multiplies the local overlap
by {\color{black}$W_Q({\upsilon}_j,\zeta_j)$} and outputs $R_{A'_j}(\zeta_j)$.
Independence of the $U_1$ gates makes these local coefficients multiply.
For a fixed output tuple $\zeta$, summing over {\color{black}${\upsilon}$}, then over the
input tuple {\color{black}${\xi}$} with the factor {\color{black}$\sgnrow{}_Q^m$} from
\eqref{eq:ck-four-first-layer}, gives exactly {\color{black}$\sectorcoef{}_\zeta$} in
\eqref{eq:ck-four-expansion-coefficients}. This proves
\eqref{eq:ck-four-operator-expansion}. The result is independent of $b$,
so averaging the uniform outcome introduces no additional factor.
Contracting \eqref{eq:ck-four-operator-expansion} against $\langle u^{[4]}|$ and
$|u^{[4]}\rangle$ proves \eqref{eq:ck-four-exact-expansion}.
For $m=1$, the two halves belong to the same block:
$G_q(\Xi,\Xi')^2=G_Q(\Xi,\Xi')$ and $W_QG_Q=I$ reduce
\eqref{eq:ck-four-expansion-coefficients} to {\color{black}$\sectorcoef{}_\zeta=\sgnrow{}_Q$}.
Hence the identity also includes the single block case.
\end{proof}

The local four copy formulas \cite{GrossNezamiWalter2021,HelsenWalter2023}
and the general framework for Clifford replica contractions
\cite{MagniEtAl2025} are established. The contribution here is the explicit
periodic expansion \eqref{eq:ck-four-operator-expansion} with the signed
coefficients \eqref{eq:ck-four-expansion-coefficients}, which separates
circuit geometry from the test state. The coefficients {\color{black}$\sectorcoef{}_\zeta$} depend
only on the block size and periodic geometry; all dependence on $|u\rangle$ is in
$c_{u,A'}(\zeta)$. This separation, combined with
Lemma~\ref{lem:ck-four-copy-boundary} and the coefficient estimates below,
yields the uniform bound of Theorem~\ref{thm:ck-fourth-projector} for
arbitrary test states, including states entangled across the blocks.

\begin{lemma}[Uniform sector contraction bound]
\label{lem:ck-four-copy-boundary}
For any disjoint block partition $A'$, any $\zeta\in\Sigma_4^m$, and
any unit vector $|u\rangle$,
\begin{equation}
 |c_{u,A'}(\zeta)|\leq1.
 \label{eq:ck-four-boundary}
\end{equation}
\end{lemma}
\begin{proof}
Let $D=\{j_1,\ldots,j_\ell\}\subseteq\{0,\ldots,n-1\}$ be a set of
$\ell$ physical qubits, as in the replica/site notation above.
Its Hilbert space dimension is $d_D=2^\ell=2^{|D|}$. Write
$\Omega_D=d_D^{-1}\sum_{P\in\mathcal P(D)}P^{[4]}$,
where $\mathcal P(D)$ contains one Hermitian representative of each
Pauli label on $D$ \cite[Eq.~(16), arXiv~v2]{HelsenWalter2023}.
The classification into 24 replica permutation sectors and six exceptional
sectors \cite[pp.~7--8, arXiv~v2]{HelsenWalter2023} gives, on each block,
{\color{black}$R_{A'_j}(\zeta_j)=\Omega_{A'_j}^{\sectorflag{}_j}V_j$},
with {\color{black}$\sectorflag{}_j\in\{0,1\}$},
where $V_j$ is the site representation of a replica permutation and {\color{black}$\sectorflag{}_j=1$} exactly for
an exceptional sector. When {\color{black}$\sectorflag{}_j=0$}, $\Omega_{A'_j}^{0}=I$,
so $R_{A'_j}(\zeta_j)=V_j$. Set {\color{black}${D_{\mathrm{ex}}}=\bigcup_{\sectorflag{}_j=1}A'_j$} and
$V=\bigotimes_jV_j$. The blocks are disjoint, so their Pauli sums factor:
{\color{black}\[
 \bigotimes_jR_{A'_j}(\zeta_j)
 =\left(\bigotimes_{\sectorflag{}_j=1}\Omega_{A'_j}\right)V
 =\Omega_{D_{\mathrm{ex}}}V.
\]}%
Operators on {\color{black}${D_{\mathrm{ex}}}$} are extended by identity outside {\color{black}${D_{\mathrm{ex}}}$}; if
{\color{black}${D_{\mathrm{ex}}}=\varnothing$}, then {\color{black}$\Omega_{D_{\mathrm{ex}}}=I$} and {\color{black}$d_{D_{\mathrm{ex}}}=1$}. Moreover,
{\color{black}\begin{equation}
 \Omega_{D_{\mathrm{ex}}}=\frac1{d_{D_{\mathrm{ex}}}}\sum_{P\in\mathcal P({D_{\mathrm{ex}}})}P^{[4]}
          =d_{D_{\mathrm{ex}}}\Pi_{D_{\mathrm{ex}}},
 \label{eq:ck-pauli-code-projector}
\end{equation}}%
where {\color{black}$\Pi_{D_{\mathrm{ex}}}$} is an orthogonal projector. Each $V_j$ permutes the four
identical copies of a Pauli operator, so $V$ is unitary and commutes with
{\color{black}$\Omega_{D_{\mathrm{ex}}}$} and {\color{black}$\Pi_{D_{\mathrm{ex}}}$}.
For {\color{black}$|{\utensor{}}\rangle=|u^{[4]}\rangle$}, we have
{\color{black}$c_{u,A'}(\zeta)=\langle {\utensor{}}|\Omega_{D_{\mathrm{ex}}}V|{\utensor{}}\rangle$}.
Cauchy--Schwarz inside the range of {\color{black}$\Pi_{D_{\mathrm{ex}}}$} gives
{\color{black}$|\langle {\utensor{}}|\Omega_{D_{\mathrm{ex}}}V|{\utensor{}}\rangle|
\leq\langle {\utensor{}}|\Omega_{D_{\mathrm{ex}}}|{\utensor{}}\rangle$}.
Let {\color{black}$\rho_{D_{\mathrm{ex}}}$} be the reduced state of $|u\rangle\!\langle u|$ on {\color{black}${D_{\mathrm{ex}}}$}.
Since each $P$ is Hermitian and unitary, {\color{black}$\tr(P\rho_{D_{\mathrm{ex}}})\in[-1,1]$}.
Pauli Parseval therefore yields
{\color{black}\[
 \begin{aligned}
 \langle {\utensor{}}|\Omega_{D_{\mathrm{ex}}}|{\utensor{}}\rangle
 &=\frac1{d_{D_{\mathrm{ex}}}}\sum_P\{\tr(P\rho_{D_{\mathrm{ex}}})\}^4\\
 &\leq\frac1{d_{D_{\mathrm{ex}}}}\sum_P\{\tr(P\rho_{D_{\mathrm{ex}}})\}^2\\
 &=\tr(\rho_{D_{\mathrm{ex}}}^2)\leq1.
 \end{aligned}
\]}%
This proves \eqref{eq:ck-four-boundary} without a product state assumption on
$|u\rangle$.
\end{proof}

The moment bound now follows from the two lemmas and the scalar estimates
in Appendix~\ref{app:local-algebra}.
\begin{theorem}[Uniform fourth projector moment]
\label{thm:ck-fourth-projector}
Under the reference law of the periodic ensemble, let $|\psi\rangle=U^*|b\rangle$,
where $|b\rangle$ is a computational basis state.
Then condition \eqref{eq:headline-block-condition} implies
\begin{equation}
 \sup_{\langle u|u\rangle=1}\E_\nu|\langle u|\psi\rangle|^8
 \leq\frac{34}{d(d+1)(d+2)(d+3)}.
 \label{eq:ck-u4-34}
\end{equation}
The supremum includes arbitrary entangled test states.
\end{theorem}
\begin{proof}
The block condition \eqref{eq:headline-block-condition} gives $q\geq256$,
so Lemma~\ref{lem:ck-four-moment-expansion} applies.
By Lemma~\ref{lem:ck-four-copy-data}, specifically
\eqref{eq:ck-four-absolute-row}, the absolute row sum of $W_Q$ is
{\color{black}\[
 \absrow{}_Q:=\sum_{\Xi'\in\Sigma_4}|W_Q(\Xi,\Xi')|
     =\frac1{Q(Q-1)(Q-2)(Q-4)}.
\]}%
Apply \eqref{eq:ck-four-exact-expansion} and
Lemma~\ref{lem:ck-four-copy-boundary} to obtain
{\color{black}\[
 \E_\nu|\langle u|\psi\rangle|^8
 \leq\sum_{\zeta\in\Sigma_4^m}|\sectorcoef{}_\zeta|.
\]}%
For fixed {\color{black}${\upsilon}$}, the output labels $\zeta_j$ are summed independently:
{\color{black}\[
 \sum_{\zeta\in\Sigma_4^m}\prod_j|W_Q({\upsilon}_j,\zeta_j)|
 =\prod_j\sum_{\zeta_j\in\Sigma_4}|W_Q({\upsilon}_j,\zeta_j)|=\absrow{}_Q^m.
\]}%
Next, symmetry of $G_q$ gives, for each fixed $j$,
{\color{black}\[
 \sum_{{\upsilon}_j\in\Sigma_4}
 G_q({\upsilon}_j,{\xi}_j)G_q({\upsilon}_j,{\xi}_{j+1})
 =(G_q^2)_{{\xi}_j,{\xi}_{j+1}},
\]}%
where $G_q^2$ is the ordinary matrix square. Taking absolute values in
\eqref{eq:ck-four-expansion-coefficients} and using {\color{black}$\sgnrow{}_Q>0$} and $G_q\geq0$
entrywise, sum first over $\zeta$ and then independently over the
{\color{black}${\upsilon}_j$}. The two identities above yield
{\color{black}\begin{equation}
 \begin{aligned}
 \sum_{\zeta\in\Sigma_4^m}|\sectorcoef{}_\zeta|
 &\leq \sgnrow{}_Q^m\absrow{}_Q^m\sum_{{\xi}\in\Sigma_4^m}\sum_{{\upsilon}\in\Sigma_4^m}\\[-2pt]
 &\qquad{}\prod_{j=0}^{m-1}
 G_q({\upsilon}_j,{\xi}_j)G_q({\upsilon}_j,{\xi}_{j+1})\\
 &=\sgnrow{}_Q^m\absrow{}_Q^m\sum_{{\xi}\in\Sigma_4^m}
       \prod_{j=0}^{m-1}(G_q^2)_{{\xi}_j,{\xi}_{j+1}}.
 \end{aligned}
 \label{eq:ck-four-absolute-expansion}
\end{equation}}%
The cyclic sum is
{\color{black}\begin{equation}
 \sum_{{\xi}\in\Sigma_4^m}
 \prod_{j=0}^{m-1}(G_q^2)_{{\xi}_j,{\xi}_{j+1}}
 =\tr(G_q^{2m}).
 \label{eq:ck-periodic-fourth-cycle}
\end{equation}}%
The matrix $G_q$ has nonnegative entries and constant row sum
\cite[pp.~7--8, arXiv~v2]{HelsenWalter2023}: for every $\Xi\in\Sigma_4$,
\[
 \sum_{\Xi'\in\Sigma_4}G_q(\Xi,\Xi')=g_q
 :=q(q+1)(q+2)(q+4).
\]
Hence every row of $G_q^{2m}$ sums to
$g_q^{2m}$. Bounding each diagonal entry by its row sum gives
$\tr(G_q^{2m})\leq30g_q^{2m}$, and substitution gives
\begin{equation}
 \begin{aligned}
 \E_\nu|\langle u|\psi\rangle|^8
 &\leq30d^{-4}\mathcal A(q)^m,\\[-2pt]
 \mathcal A(q)&=\prod_{c\in\{1,2,4\}}
 \frac{(1+c/q)^2}{1-c^2/q^4}.
 \end{aligned}
 \label{eq:ck-four-preconstant}
\end{equation}
Appendix~\ref{app:local-algebra} proves
\begin{equation}
 \ln\mathcal A(q)\leq14/q\qquad(q\geq8).
 \label{eq:ck-four-local-inflation}
\end{equation}
Condition \eqref{eq:headline-block-condition} also gives
$m/q\leq\ln2/92$ and $d\geq65536$, hence
$\mathcal A(q)^m\leq2^{7/46}$. Finally,
\[
 30\,2^{7/46}(1+d^{-1})(1+2d^{-1})(1+3d^{-1})<34
\]
for $d\geq65536$, proving \eqref{eq:ck-u4-34} uniformly in $|u\rangle$.
\end{proof}
\end{revision}
\subsection{Curvature and low order conditioning}
\label{sec:periodic-curvature}
The ingredients of the next proposition have established antecedents.
Cho and Kim analyze the same periodic architecture through its Pauli diagonal
measurement channel and lower bounds on its Pauli eigenvalues
\cite[Appendix~B]{ChoKim2026}. Exact third moments of multiqubit Clifford
gates were proved by Webb and Zhu \cite{Webb2016,Zhu2017}; our local
argument uses the same Pauli orbit method. The comparison with a
global Haar third moment follows the overlapping block gluing method of
Schuster, Haferkamp, and Huang
\cite[Appendices~B.3 and D.2, arXiv~v2]{SchusterHaferkampHuang2025}.
Proposition~\ref{prop:periodic-low-order} assembles these ingredients for
the forward channel and records the constants used here.

\begin{proposition}[Curvature, conditioning, and scalar variance]
\label{prop:periodic-low-order}
Under \eqref{eq:headline-block-condition}, the periodic channel satisfies
\eqref{eq:curvature} with $a\geq a_*=32767/32768$, and, for every state,
\begin{equation}
 \frac{19627}{50000}I\preceq I+\cL(\rho)\preceq\frac{80373}{25000}I.
 \label{eq:forward-covariance-conditioning}
\end{equation}
Moreover,
\begin{equation}
 \sup_{\rho,\langle v|v\rangle=1}\operatorname{Var}_\rho(\langle v|Y_1|v\rangle)
 \leq\kappa_*:=\frac{241119}{25000}<13.
 \label{eq:scalar-variance}
\end{equation}
\end{proposition}
\begin{revision}
The proof is given in Appendix~\ref{app:prop10-proof}.
\end{revision}

Theorem~\ref{thm:ck-fourth-projector} supplies the fourth projector moment
bound \eqref{eq:primitive-fourth}, and Proposition~\ref{prop:periodic-low-order}
supplies the covariance bounds \eqref{eq:primitive-conditioning} required
by Proposition~\ref{prop:primitive-covariance}.
The remaining
curvature and scalar variance conclusions of
Proposition~\ref{prop:periodic-low-order} are used in the reconstruction
bounds through Lemma~\ref{lem:trace-conversion} and the finite sample estimates
in Appendix~\ref{app:quantitative}, respectively.

\begin{revision}
\section{Nonadaptive Single Copy Lower Bound}
\label{sec:lower}
The state space {\color{black}$\cD_d$} and the risk in Section~\ref{sec:risks}
are defined for every integer {\color{black}$d$}; only the circuit construction
specializes to {\color{black}$d=2^n$}. The lower bound therefore applies in arbitrary
integer dimensions and does not require the block condition \eqref{eq:headline-block-condition}.

\begin{theorem}[Spectral decay lower bound]
\label{thm:lower}
For every {\color{black}$\alpha>1$}, {\color{black}$d\geq514$}, {\color{black}$T\geq1$}, and {\color{black}$L\geq1$},
{\color{black}\begin{equation}
 \mathfrak R_T(\cC_{\alpha,d}(L))\geq c_\alpha\notationmark{\mathsf{r}}_{\alpha,L}(d,T),
 \label{eq:lower-main}
\end{equation}}%
where {\color{black}$c_\alpha>0$} depends only on {\color{black}$\alpha$}. Moreover, for every randomized
nonadaptive single copy design and estimator as specified in
Section~\ref{sec:risks} (see \eqref{eq:minimax-risk}),
there is a state {\color{black}$\rho\in\cC_{\alpha,d}(L)$} for which the error is at least a
constant multiple of the right side with probability at least {\color{black}$1/2$}.

\end{theorem}
\begin{proof}
Sections~\ref{sec:lower-packing}--\ref{sec:lower-fano} construct a fixed
admissible state with the probability and expectation bounds
\eqref{eq:fixed-hard-state-probability}--\eqref{eq:fixed-hard-state-expectation}.
Section~\ref{sec:lower-envelope} optimizes its separation using the packing
and integer lemmas in Appendix~\ref{app:lower-details}, yielding the stated
rate with error probability at least {\color{black}$11/16$}.
\end{proof}

\subsection{Grassmann packing and the hard family}
\label{sec:lower-packing}

We use the metric Grassmann packing of
{\color{black}\cite[Lemma~2, Eq.~(19)]{FlammiaGrossLiuEisert2012}}. Its explicit specialization
Lemma~\ref{lem:grassmann-packing}, proved in Appendix~\ref{app:lower-details},
gives {\color{black}$N$} rank {\color{black}${\rtail}$} projectors with {\color{black}$\ln N\geq\gamma_0{\rtail} d_\perp$},
{\color{black}$\gamma_0=25/4608$}, and pairwise trace separation at least {\color{black}${\rtail}/2$}.

Put {\color{black}$d_\perp=d-2$} and split {\color{black}$\C^d=H\oplus K$} with {\color{black}$\dim H=2$} and
{\color{black}$\dim K=d_\perp$}.  For {\color{black}$1\leq {\rtail}\leq\lfloor d_\perp/3\rfloor$}, apply
Lemma~\ref{lem:grassmann-packing} on {\color{black}$K$}.
Choose
{\color{black}\begin{equation}
 0<\beta\leq\min\{1/4,\,2^{1-\alpha}L{\rtail}^{1-\alpha}\}
 \label{eq:b-class-condition}
\end{equation}}%
and define
{\color{black}\begin{equation}
\rho_x=\frac{1-\beta}{2}\Pi_H+\frac{\beta}{{\rtail}}P_x.
\label{eq:hard-states}
\end{equation}}%
Since {\color{black}$\beta\leq1/4\leq {\rtail}/({\rtail}+2)$}, the ordered eigenvalues are two copies of
{\color{black}$(1-\beta)/2$}, followed by {\color{black}${\rtail}$} copies of {\color{black}$\beta/{\rtail}$}, followed by zeros.  Hence
{\color{black}\begin{equation}
 \tau_s(\rho_x)=
 \begin{cases}
  (1+\beta)/2,&s=1,\\
  \beta({\rtail}+2-s)/{\rtail},&2\leq s\leq {\rtail}+1,\\
  0,&{\rtail}+2\leq s\leq d.
 \end{cases}
 \label{eq:hard-exact-tails}
\end{equation}}%
For {\color{black}$s=1$}, \eqref{eq:hard-exact-tails} is at most {\color{black}$5/8\leq L$}.  For
{\color{black}$2\leq s\leq {\rtail}+1$},
{\color{black}\[
 \tau_s(\rho_x)\leq \beta
 \leq L(2{\rtail})^{1-\alpha}\leq Ls^{1-\alpha},
\]}%
The last inequality uses {\color{black}$s\leq {\rtail}+1\leq2{\rtail}$} and {\color{black}$1-\alpha<0$}; the remaining tails vanish.  Hence {\color{black}$\rho_x\in\cC_{\alpha,d}(L)$}.  Also,
{\color{black}\begin{equation}
 \|\rho_x-\rho_{x'}\|_{\mathrm{tr}}
 =\frac{\beta}{{\rtail}}\|P_x-P_{x'}\|_{\mathrm{tr}}\geq\frac{\beta}{2}.
 \label{eq:hard-separation}
\end{equation}}

Let {\color{black}$W\sim\mu_{\mathrm{Haar},d}$} be one shared hidden rotation, independent
of {\color{black}$x$}, and present the measurement design with the state
{\color{black}\begin{equation}
 \rho_{x,W}=W\rho_xW^\dagger.
 \label{eq:rotated-hard-state}
\end{equation}}%
For comparison define
{\color{black}\begin{equation}
 \rho_0=\frac{1-\beta}{2}\Pi_H+\frac \beta{d_\perp}\Pi_K,
 \qquad
 A_x=\frac1{\rtail}P_x-\frac1{d_\perp}\Pi_K.
 \label{eq:reference-state}
\end{equation}}%
Then {\color{black}$\rho_x=\rho_0+\beta A_x$}, {\color{black}$\tr A_x=0$}, and
{\color{black}\begin{equation}
 \tr(A_x^2)={\rtail}^{-1}-d_\perp^{-1}\leq {\rtail}^{-1}.
 \label{eq:A-second-moment}
\end{equation}}%
The full rank state {\color{black}$\rho_0$} is only an information radius reference; it
need not belong to the decay class.

\subsection{One copy information after a shared rotation}

For probability densities $p$ and $q$ with respect to a common
measure $\mu$ on $\Omega$, define
\begin{equation}
 \begin{aligned}
 \chi^2(p\|q)&:=\int_\Omega\frac{(p(z)-q(z))^2}{q(z)}\,d\mu(z),\\
 D_{\KL}(p\|q)&:=\int_\Omega p(z)\ln\frac{p(z)}{q(z)}\,d\mu(z).
 \end{aligned}
 \label{eq:information-divergences}
\end{equation}
These divergences satisfy $D_{\KL}(p\|q)\leq\chi^2(p\|q)$.

\begin{lemma}[Haar averaged one copy chi square bound]
\label{lem:one-copy-chi-square}
Fix any one copy POVM on a standard Borel outcome space. For every packing
label $x$, let $p_{x,W}$ and $p_{0,W}$ be its outcome densities under
$W\rho_xW^\dagger$ and $W\rho_0W^\dagger$, respectively, relative to a
common dominating measure. Then
\begin{equation}
 \E_W\chi^2(p_{x,W}\|p_{0,W})
 \leq\frac{2\beta^2}{1-\beta}\tr(A_x^2)
 \leq\frac{8\beta^2}{3{\rtail}}.
 \label{eq:one-copy-chi-square}
\end{equation}
\end{lemma}
The proof is given in Appendix~\ref{sec:one-copy-details}.

\subsection{Information reduction and testing lower bound}
\label{sec:lower-fano}

The following lemma combines the one copy bound of
Lemma~\ref{lem:one-copy-chi-square} with Fano's inequality to lower bound
the probability of misidentifying the packing label from $T$
nonadaptive single copy measurements.

\begin{lemma}[Nonadaptive Fano reduction]
\label{lem:nonadaptive-fano}
Let $\{\rho_x\}_{x=1}^N$, $\rho_0$, and the shared Haar rotation $W$ be as
in Section~\ref{sec:lower-packing}, with $N\geq2$. Choose the hidden label
$X$ uniformly from $\{1,\ldots,N\}$, with $X$, $W$, and the design seed
mutually independent. Suppose $\kappa\geq0$ satisfies, for every label $x$ and every
one copy POVM, the bound
\begin{equation}
 \E_W\chi^2(p_{x,W}\|p_{0,W})\leq\kappa,
 \label{eq:nonadaptive-chi-budget}
\end{equation}
with the densities defined in Lemma~\ref{lem:one-copy-chi-square}.
For every randomized nonadaptive single copy design using $T\geq1$
copies, every possibly randomized decoder $\widehat X$ based on the
outcomes, the public design seed, and $W$ revealed after measurement
has error probability, averaged over all randomness, satisfying
\begin{equation}
 \Prob\{\widehat X\neq X\}
 \geq1-\frac{T\kappa+\ln2}{\ln N}.
 \label{eq:nonadaptive-fano-bound}
\end{equation}
In particular, if $T\kappa\leq(\ln N)/16$ and $\ln N\geq4\ln2$, then
\begin{equation}
 \Prob\{\widehat X\neq X\}\geq\frac{11}{16}.
 \label{eq:fano-error}
\end{equation}
\end{lemma}
The proof is given in Appendix~\ref{sec:nonadaptive-fano-proof}.

Apply Lemma~\ref{lem:nonadaptive-fano} with
$\kappa=8\beta^2/(3{\rtail})$, supplied by
Lemma~\ref{lem:one-copy-chi-square}, and put
\begin{equation}
 c_I:=\sqrt{\frac{3\gamma_0}{128}}.
 \label{eq:fano-ci}
\end{equation}
If $\beta\leq c_I{\rtail}\sqrt{d_\perp/T}$, then
$T\kappa\leq\gamma_0{\rtail}d_\perp/16\leq(\ln N)/16$.
For $d\geq514$, the packing bound also gives $\ln N\geq4\ln2$.
Thus \eqref{eq:fano-error} applies.

Given $\widehat\rho$, decode by its nearest rotated packing state,
breaking ties by the smallest index. The separation
\eqref{eq:hard-separation} makes this decoder correct whenever
$\|\widehat\rho-W\rho_XW^\dagger\|_{\mathrm{tr}}<\beta/4$.
Consequently, a label error implies trace norm error at least $\beta/4$.
Since $X$ is uniform and independent of $W$,
\[
 \frac1N\sum_{x=1}^N\E_W\Prob\{\widehat X\neq x\mid X=x,W\}
 =\Prob\{\widehat X\neq X\}\geq\frac{11}{16}.
\]
Here the conditional probabilities average over the design seed,
outcomes, and estimator randomness. Hence some fixed pair $(x_*,w_*)$
has conditional error probability at least $11/16$; it may depend on the
design and estimator, but not their realized randomness. For
$\rho_*:=w_*\rho_{x_*}w_*^\dagger\in\cC_{\alpha,d}(L)$,
\begin{align}
 \Prob_{\rho_*}\!\left\{
 \|\widehat\rho-\rho_*\|_{\mathrm{tr}}\geq\frac\beta4
 \right\}&\geq\frac{11}{16},
 \label{eq:fixed-hard-state-probability}\\
 \E_{\rho_*}\|\widehat\rho-\rho_*\|_{\mathrm{tr}}
 &\geq\frac\beta4\,\frac{11}{16}
 =\frac{11\beta}{64}.
 \label{eq:fixed-hard-state-expectation}
\end{align}

\subsection{From packing separation to the minimax rate}
\label{sec:lower-envelope}

At this stage the statistical part is finished.  Three restrictions remain
on the packing separation {\color{black}$\beta$}: it cannot exceed a constant; class
membership requires {\color{black}$\beta\lesssim L{\rtail}^{1-\alpha}$}; and Fano requires
{\color{black}$\beta\lesssim {\rtail}\sqrt{d/T}$}.  Taking the minimum enforces all three
simultaneously.

Now write {\color{black}$M=\lfloor d_\perp/3\rfloor$} and set
{\color{black}\begin{equation}
 \begin{aligned}
 c_\alpha^*&:=\min\{1/4,2^{1-\alpha},c_I\},\\
 \beta_{\rtail}&:=c_\alpha^*
 \min\left\{1,L{\rtail}^{1-\alpha},{\rtail}\sqrt{d_\perp/T}\right\}.
 \end{aligned}
 \label{eq:explicit-bm}
\end{equation}}%
For every {\color{black}$1\leq {\rtail}\leq M$}, this choice satisfies the spectral ordering,
class membership, and information restrictions.  Maximizing
\eqref{eq:fixed-hard-state-expectation} over this finite set gives the
primitive lower bound
{\color{black}\begin{equation}
 \begin{aligned}
 \mathfrak R_T(\cC_{\alpha,d}(L))
 &\geq\frac{11c_\alpha^*}{64}\max_{1\leq {\rtail}\leq M}\\[-2pt]
 &\quad\min\left\{1,\,L{\rtail}^{1-\alpha},\,{\rtail}\sqrt{d_\perp/T}\right\}.
 \end{aligned}
 \label{eq:primitive-lower}
\end{equation}}%
The probability statement \eqref{eq:fixed-hard-state-probability} proves the
claimed constant probability version with the same choice of {\color{black}${\rtail}$}.

The discrete envelope in Lemma~\ref{lem:discrete-lower-envelope}, proved in
Appendix~\ref{app:lower-details}, completes the integer optimization.

Apply the lemma with {\color{black}$\eta=\sqrt{d_\perp/T}$}. The stated hypothesis
{\color{black}$d\geq514$} gives {\color{black}$M\geq(85/257)d$} and
{\color{black}$\sqrt{d_\perp/T}\geq(499/500)\sqrt{d/T}$}.
Comparing the three branches before the integer optimization factor {\color{black}$1/2$}
gives a factor {\color{black}$33/100$}, so
{\color{black}\begin{equation}
 G_M(L,\sqrt{d_\perp/T})\geq\frac{33}{200}
 \min\left\{1,
 L^{1/\alpha}\left(\frac dT\right)^{\frac{\alpha-1}{2\alpha}},
 \sqrt{\frac{d^3}{T}}\right\}.
 \label{eq:lower-D-comparison}
\end{equation}}%
Combining \eqref{eq:primitive-lower} and
\eqref{eq:lower-D-comparison} proves Theorem~\ref{thm:lower}.  For example,
{\color{black}$363c_\alpha^*/12800$} is a valid expected risk constant.
\end{revision}

\section{Dense Matrix Computation}
\label{sec:implementation}\label{sec:polynomial-forward-fit}
The statistical results require feasible fits at the stated objective or
first order tolerances. PLS averages the projectors, applies the inverse
channel once, and projects the eigenvalues onto the probability simplex.
OMD admits a projected subgradient construction. Weighted PLS uses
Frank--Wolfe updates with a smallest eigenvalue vector of the gradient
and stops when its first order gap passes the prescribed tolerance.

In the stated dense arithmetic model, including streamed input, the
sufficient total costs are $O(Td^2+d^3)$ for PLS and weighted PLS and
$O(Td^3+d^4)$ for projected OMD. These bounds suppress spectral precision
factors. The equal total orders for PLS and weighted PLS include data
processing; their costs after aggregation differ. Supplement
Section~\ref*{supp:computational-proofs} gives the exact iteration budgets,
accounting model, storage analysis, finite precision requirements, and
rational SDP construction.

Our numerical OMD solver uses capped ADMM for the same fitting objective.
A run that reaches its cap without passing the gap tests reports
nonconvergence and does not supply an accepted fit for the statistical
theorem. The projected construction supplies the guaranteed optimization
procedure. Supplement Section~\ref*{app:numerical-optimization} gives the
ADMM updates and numerical checks. Appendix~\ref{app:formal-verification}
states the formal verification scope.

\section{Numerical Experiments}
\label{sec:numerical-experiments}
We compare PLS, OMD with statistical stopping, and weighted PLS with Frank--Wolfe
line search on the same measurement data, varying measurement count and dimension. These small block experiments lie outside
\eqref{eq:headline-block-condition} and assess the chosen state's empirical
accuracy and computation.

\subsection{Setup and reported quantities}
The true state is $\rho=\operatorname{diag}(1/8,\ldots,1/8,0,\ldots,0)$,
with eight positive eigenvalues and computational eigenbasis. Each
setting has 100 independent datasets shared by the three procedures.
No solver receives the true state, rank, or eigenbasis.
\reviewchange{These experiments illustrate the selected implementations
on one spectrum and eigenbasis; they do not establish uniform spectral
adaptation or validate the sufficient block size threshold.}
Table~\ref{tab:simulation-design} gives the two sweeps retained here.

\begin{table}[!b]
\caption{Simulation settings; the true rank is eight throughout.}
\label{tab:simulation-design}
\centering\small
\begin{tabular}{lll}
\hline
Sweep & Varied values & Fixed settings\\
\hline
Copies & $T=4^5,4^6,\ldots,4^9$ & $d=256,\ k=2$\\
Dimension & $d=16,64,256,1024$ & $k=2,\ T/d=64$\\
\hline
\end{tabular}
\end{table}

OMD starts from a fresh PLS estimate; weighted PLS starts from $I_d/d$ with
quadratic line search \cite{Jaggi2013}. Their distinct stopping certificates
are \eqref{eq:candidate-two-program} and \eqref{eq:mw-approximate-fit}.
The OMD and weighted PLS procedures are specified in Supplement
Sections~\ref*{supp:omd-procedure} and~\ref*{supp:mw-steps}, respectively.

Error curves show mean trace norm loss with pointwise 95\% paired bootstrap
intervals. Runtime shows medians and interquartile ranges across dataset
medians of three fresh solves after a discarded warmup, including initialization and stopping
checks. \reviewchange{Solver timing starts from the stored empirical
mean and a constructed channel; it excludes observation generation and
aggregation, input loading, channel setup, and independent validation.
The total work bounds include aggregation and setup.}
Memory uses fresh processes; Supplement Section~\ref*{app:simulation-reproducibility}
gives the environment and full protocol.

\reviewchange{Supplement Section~\ref*{supp:accuracy-refinement} also reports
an additional comparison on independent datasets, using the same three
estimators and their stated stopping tolerances.}

\subsection{Increasing the measurement count}
\begin{figure*}[!t]
\centering
\includegraphics[width=\textwidth]{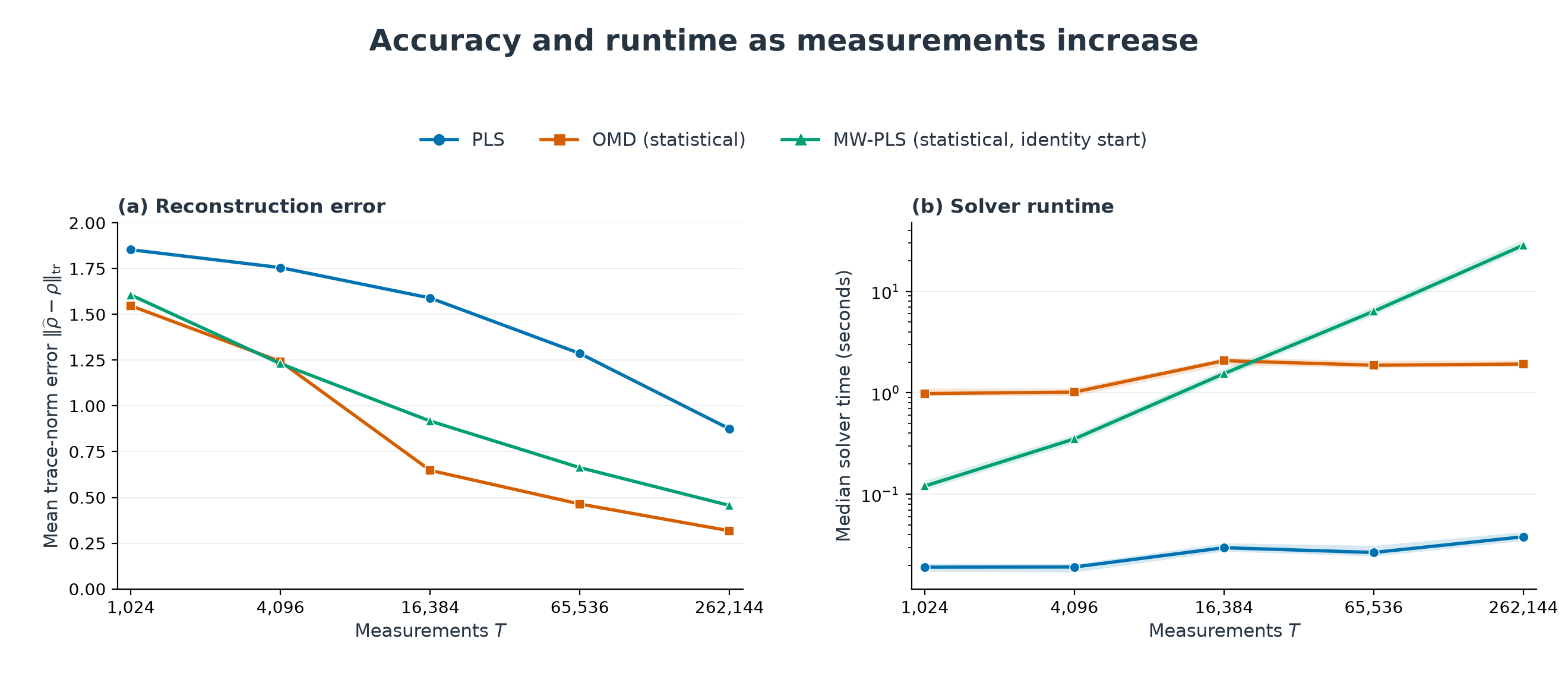}
\caption{Measurement count at $d=256$, $k=2$: mean trace norm error and
median solver time for the three selected procedures. Error bands are
pointwise 95\% bootstrap intervals; runtime bands are interquartile ranges. The
horizontal axes and runtime vertical axis are logarithmic; the error axis
is linear. \reviewchange{Timings include initialization and stopping checks,
but exclude data generation and aggregation, loading, channel setup, and
independent validation.}}
\label{fig:msi-measurements}
\end{figure*}

In Fig.~\ref{fig:msi-measurements}, more observations reduce error for
all three procedures. From $T=1024$ to $262144$, mean trace norm error
falls from $1.853$ to $0.874$ for PLS, from $1.547$ to $0.318$ for OMD,
and from $1.606$ to $0.456$ for weighted PLS. PLS remains fastest, while
OMD and weighted PLS give lower error on these states.

The OMD and weighted PLS errors nearly coincide at $T=4096$, but separate
as $T$ increases. Their objectives differ. Weighted PLS runtime rises from $0.121$ to $28.510$ seconds,
whereas OMD takes about one to two seconds. Here the MW gap tolerance
shrinks as $d/T$ and the OMD tolerance as $\sqrt{d/T}$. The resulting
iteration demands help explain the runtime trend, together with the
different objectives and convergence behavior.

The additional block size sweep is reported in Supplement
Section~\ref*{supp:block-experiment}. \reviewchange{It illustrates
state dependent changes with the measurement design. At the global
endpoint, all three implementations use the common spectral solution,
giving identical estimates and reconstruction errors.}

\subsection{Increasing the dimension}
\begin{figure*}[!t]
\centering
\includegraphics[width=\textwidth]{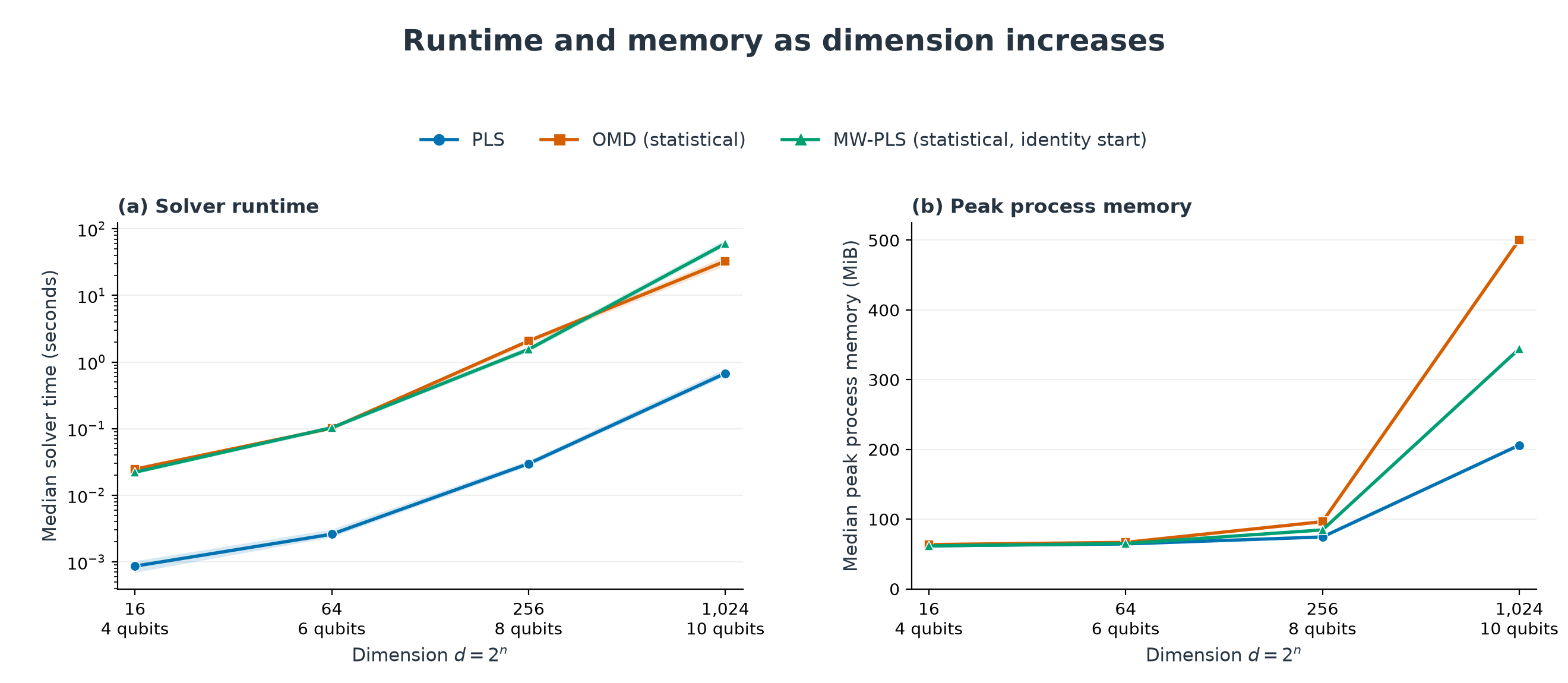}
\caption{Dimension at $k=2$, $T/d=64$: median runtime and peak process
memory, with interquartile ranges across datasets and fresh processes,
respectively. The horizontal axes and runtime vertical
axis are logarithmic; memory is linear. PLS uses one spectral projection;
OMD and weighted PLS perform repeated spectral calculations.
\reviewchange{Solver timings exclude data generation and aggregation,
loading, channel setup, and independent validation.}}
\label{fig:msi-dimension}
\end{figure*}

Figure~\ref{fig:msi-dimension} shows the cost of dense reconstruction.
At $d=1024$, median times for PLS, OMD, and weighted PLS are $0.677$,
$32.693$, and $60.054$ seconds; their mean trace norm errors are
$1.624$, $0.602$, and $0.823$. Median peak memory is $206.2$, $500.7$,
and $344.4$ MiB. Thus weighted PLS uses less memory than this OMD
implementation but takes longer at this setting. At small dimensions,
shared interpreter and library costs make the memory curves nearly
overlap.

The runtime vertical axes are logarithmic: a separation of {\color{black}${\Delta_{\mathrm{axis}}}$}
base ten units represents a time ratio {\color{black}$10^{\Delta_{\mathrm{axis}}}$}. At $d=1024$,
OMD and weighted PLS take about $48$ and $89$ times the PLS time.
The error and memory axes are linear. A logarithmic axis alone does
not establish exponential growth in $d$; polynomial dense matrix
costs in $d=2^n$ are nevertheless exponential in $n$.

These rankings depend on the procedures and tolerances; minimax guarantees
do not order performance on each state. Acceleration and its required gap
checks are discussed in Supplement Sections~\ref*{app:certified-acceleration}
and~\ref*{supp:mw-steps}.

\section{Discussion}
\label{sec:discussion}
OMD and weighted PLS attain adaptive minimax trace risk over spectral tail classes
with the periodic measurements of Cho and Kim.

On the tested states outside the sufficient theorem regime, PLS is fastest;
OMD and weighted PLS give lower error at the chosen tolerances.

\paragraph{Sharp spectral guarantees for PLS}
PLS already attains the rate over the full state space. For global Clifford measurements it also inherits the spectral class
rate from weighted PLS. Sharp adaptation for general nonglobal designs, under
trace norm or Hilbert--Schmidt loss, remains open.

\paragraph{Measurement assumptions}
Sharper versions of the block condition \eqref{eq:headline-block-condition}
would improve finite resource guarantees.
Section~\ref{sec:oracle} also applies to other ensembles satisfying its moment
and conditioning assumptions. Noise, basis reuse, and measurements adapted
to outcomes require separate analysis.

\paragraph{Computation and output representation}
\phantomsection\label{sec:scalable-open-problem}
For OMD, the computational target remains a feasible state
$\widetilde\rho\in\cD_d$ satisfying
\[
 f_T(\widetilde\rho)-f_T^*\leq\min\{1,\sqrt{d/T}\},
\]
without rank or decay inputs; this certificate preserves the risk guarantee.
Weighted PLS has the corresponding target in \eqref{eq:mw-approximate-fit};
PLS requires accurate channel inversion and Frobenius projection.

Writing a dense estimate requires $\Omega(d^2)$ entries, and an arbitrary
pure state has $2d-2$ real parameters. Polynomial scaling in $n$ thus requires extra structure and compressed
output or restricted access.

\appendices
\section{Local Fourth Copy Algebra and Numerical Bounds}
\label{app:local-algebra}
We use the standard four copy Clifford sector set $\Sigma_4$ and
normalization of \cite{GrossNezamiWalter2021,HelsenWalter2023}. Its definition
and count of 30 sectors are established in
\cite[Definition~4.1 and Example~4.27]{GrossNezamiWalter2021}; the row
multiplicities are recorded in \cite[pp.~7--8, arXiv~v2]{HelsenWalter2023}.
Concretely, $\Sigma_4$ consists of the four dimensional subspaces
$\Xi\subseteq\mathbb F_2^4\oplus\mathbb F_2^4$ over $\mathbb F_2$ containing
$(\mathbf1,\mathbf1)$ and satisfying {\color{black}$|{\mathbf b}|-|{\mathbf c}|\equiv0\pmod4$} for every
{\color{black}$({\mathbf b},{\mathbf c})\in\Xi$}, where $\mathbf1=(1,1,1,1)$ and $|\cdot|$ denotes Hamming weight.
We verify the count and intersections by explicit finite enumeration. The only one dimensional subspace containing $(\mathbf1,\mathbf1)$ is
$\operatorname{span}\{(\mathbf1,\mathbf1)\}$.
Starting from it, adjoining all vectors that preserve this condition and
removing duplicates gives
$35,105,30$ subspaces in dimensions $2,3,4$. Their intersections give
the row profile stated below; the accompanying integer checker tests all $900$ ordered pairs.

\begin{revision}
For completeness, we record the standard local averaging formula used in
Lemma~\ref{lem:ck-four-moment-expansion}:
{\color{black}\begin{equation}
 \begin{aligned}
 \twirlfour{}_{\dloc{}}({\oparg{}})&:=\E_{\cliffgate{}}({\cliffgate{}}^{[4]})^*{\oparg{}}{\cliffgate{}}^{[4]}\\
 &=\sum_{\Xi,\Xi'\in\Sigma_4}W_{\dloc{}}(\Xi,\Xi')
   \tr\bigl(R_{\dloc{}}(\Xi)^*{\oparg{}}\bigr)R_{\dloc{}}(\Xi'),
 \end{aligned}
 \label{eq:ck-local-average}
\end{equation}}%
where {\color{black}${\dloc{}}=2^\ell$}, $\ell\geq3$, {\color{black}${\cliffgate{}}$} is uniform on the projective
$\ell$ qubit Clifford group, and {\color{black}${\oparg{}}$} is any four copy operator in site order. Here {\color{black}${\cliffgate{}}^{[4]}=J_D {\cliffgate{}}^{\otimes4}J_D^*$}, as in Section~\ref{sec:fourth-moment-proof}.
This average, also called the Clifford twirl, is the orthogonal projection
for the Hilbert--Schmidt inner product onto the span of the 30 operators
{\color{black}$R_{\dloc{}}(\Xi)$} \cite[Theorem~4.3]{GrossNezamiWalter2021}.
Formula~\eqref{eq:ck-local-average} is the projection formula in a
nonorthogonal basis \cite[Eq.~(18), arXiv~v2]{HelsenWalter2023}.
Its coefficients may have either sign. The signed row sum {\color{black}$\sgnrow{}_\dloc{}$} is the
scalar coefficient in the fourth moment expansion of a Clifford randomized
stabilizer state \cite[Eq.~(19), arXiv~v2]{HelsenWalter2023};
the absolute row sum follows from the explicit inverse below. In our periodic circuit proof,
the signed row sum determines the $U_2$ average
\eqref{eq:ck-four-first-layer}, while the absolute row sum controls the
two layer fourth moment bound.
\end{revision}
The following lemma records an explicit finite calculation within the
standard Clifford framework with 30 sectors in
\cite{GrossNezamiWalter2021,HelsenWalter2023}. The Gram structure and row
multiplicities are known; we include the closed form inverse and absolute
row sum calculation because they are needed for the periodic contraction.

\begin{lemma}[The Gram inverse for 30 sectors]
\label{lem:ck-four-copy-data}
Let $\Sigma_4$ be the set of 30 sectors and write
$h(\Xi,\Xi')=\dim(\Xi\cap\Xi')\in\{1,2,3,4\}$.  On a block of Hilbert space dimension
{\color{black}${\dloc{}}>4$}, the sector Gram matrix and its inverse have the forms
{\color{black}\begin{equation}
\begin{aligned}
 G_{\dloc}(\Xi,\Xi')&=\dloc^{h(\Xi,\Xi')},\\
 W_{\dloc}(\Xi,\Xi')&=w_{h(\Xi,\Xi')}(\dloc).
 \end{aligned}
 \label{eq:ck-four-copy-gram}
\end{equation}}%
where, with
{\color{black}$\Delta_{\dloc{}}=({\dloc{}}^2-1)({\dloc{}}^2-4)({\dloc{}}^2-16)$},
{\color{black}\begin{equation}
\begin{array}{c|c}
 h&w_h(\dloc)\\\hline
 1&-8/(\dloc\Delta_{\dloc})\\
 2&2/\Delta_{\dloc}\\
 3&-(\dloc^2-8)/(\dloc\Delta_{\dloc})\\
 4&(\dloc^2-14)/\Delta_{\dloc}
 \end{array}
 \label{eq:ck-four-copy-inverse}
\end{equation}}%
The scalar coefficient {\color{black}$\sgnrow{}_\dloc{}$} follows from the fourth moment expansion in
\cite[Eq.~(19), arXiv~v2]{HelsenWalter2023}; the subsequent row
calculation there gives the row profile.
Every row contains respectively $8,14,7,1$ entries with
$h=1,2,3,4$. Consequently, for every $\Xi\in\Sigma_4$,
{\color{black}\begin{align}
 &\sgnrow_{\dloc}:=\sum_{\Xi'\in\Sigma_4}W_{\dloc}(\Xi,\Xi')\notag\\
 &\qquad=\frac1{\dloc(\dloc+1)(\dloc+2)(\dloc+4)},
 \label{eq:ck-four-signed-row}\\
 &\absrow_{\dloc}:=\sum_{\Xi'\in\Sigma_4}|W_{\dloc}(\Xi,\Xi')|\notag\\
 &\qquad=\frac1{\dloc(\dloc-1)(\dloc-2)(\dloc-4)}.
 \label{eq:ck-four-absolute-row}
\end{align}}
\end{lemma}

\begin{proof}
With $N_h(i,j)$ counting intermediate sectors with intersection dimensions
$i,j$, multiplication reduces {\color{black}$W_{\dloc{}}G_{\dloc{}}=I$} to
{\color{black}$\sum_{i,j}N_h(i,j)w_i({\dloc{}}){\dloc{}}^j=\mathbf1\{h=4\}$}.
The finite table and all four polynomial checks are given in Supplement
\reviewchange{Section~\ref*{supp:inverse}}. The row profile and the signs $-,+,-,+$
then give both row sums by factorization.
\end{proof}

\subsection{Natural logarithm estimate}
For $q\geq8$, the inequalities
{\color{black}$\ln(1+{t_{\log}})\leq {t_{\log}}-{t_{\log}}^2/2+{t_{\log}}^3/3$} for {\color{black}${t_{\log}}\geq0$} and
{\color{black}$-\ln(1-{u_{\log}})\leq {u_{\log}}/(1-{u_{\log}})$} for {\color{black}$0\leq {u_{\log}}<1$} yield
\begin{equation}
 \ln\mathcal A(q)\leq\frac{14}{q}-\frac{21}{q^2}
 +\frac{146}{3q^3}+\frac{21}{q^4(1-16/q^4)}.
 \label{eq:ck-four-log-remainder}
\end{equation}
The last two positive terms are at most $21/q^2$: after multiplication
by $q^2$, their sum is bounded by
$146/24+21/[64(1-1/256)]<21$. This proves
\eqref{eq:ck-four-local-inflation}. All logarithms in this calculation are
natural.

\section{Lower Bound Details}
\label{app:lower-details}
\begin{revision}
\subsection{Trace density representation of a POVM}
\label{sec:povm-density}
Let $\mathsf M$ be a POVM on a standard Borel outcome space
$(\Omega,\mathcal F)$ with $\mathsf M(\Omega)=I_d$. For
$\mathcal E\in\mathcal F$, define the finite positive measure
$\mu(\mathcal E)=\tr\mathsf M(\mathcal E)$; then $\mu(\Omega)=d$.
If $\mu(\mathcal E)=0$, positivity forces $\mathsf M(\mathcal E)=0$.
Each matrix entry measure of $\mathsf M$ is therefore absolutely
continuous with respect to $\mu$. Its Radon--Nikodym derivative defines
a measurable matrix $F(z)$ such that
\[
 \mathsf M(\mathcal E)=\int_{\mathcal E}F(z)\,d\mu(z)
 \qquad(\mathcal E\in\mathcal F).
\]
Hermitian symmetry of the entry measures gives $F=F^*$ almost everywhere.
For each vector with rational real and imaginary coordinates, positivity
of $v^*\mathsf M(\cdot)v$ implies $v^*F(z)v\geq0$ almost everywhere.
Intersecting these countably many full measure sets and using continuity
in $v$ proves $F(z)\succeq0$ almost everywhere. Finally,
$\tr\mathsf M=\mu$ implies $\tr F=1$ almost everywhere. Redefine
$F=I_d/d$ on the null exceptional set if necessary. This proves the
trace density representation used below without a measurable choice of
eigenvectors.
\end{revision}

\begin{revision}
\subsection{Proof of Lemma~\ref{lem:one-copy-chi-square}}
\label{sec:one-copy-details}
\begin{proof}
Recall the subspaces $H$ and $K$ from Section~\ref{sec:lower-packing},
with dimensions $2$ and $d_\perp=d-2$, and the matrices $\rho_0,A_x$
from \eqref{eq:reference-state}. They satisfy $\rho_0\succ0$,
$\rho_x=\rho_0+\beta A_x$, and $\tr A_x=0$.
For the fixed POVM $\mathsf M$, use Appendix~\ref{sec:povm-density}
to write $\mathsf M(\mathcal E)=\int_{\mathcal E}F\,d\mu$, where
$\mu(\mathcal E)=\tr\mathsf M(\mathcal E)$, $F(z)\succeq0$,
and $\tr F(z)=1$ almost everywhere. In particular, $\mu(\Omega)=d$.
By Born's rule, the two outcome densities are
\[
 \begin{aligned}
 p_{x,W}(z)&=\tr\bigl(F(z)W\rho_xW^\dagger\bigr),\\
 p_{0,W}(z)&=\tr\bigl(F(z)W\rho_0W^\dagger\bigr).
 \end{aligned}
\]
Thus $\int_{\mathcal E}p_{x,W}\,d\mu$ is the probability of an outcome in
$\mathcal E$ under $W\rho_xW^\dagger$, and similarly for $p_{0,W}$.
Since $\rho_0\succ0$, $p_{0,W}(z)>0$ almost everywhere. Throughout,
$\E_W$ averages over the Haar rotation $W$.

First bound a fixed rank one effect
{\color{black}$E_z=w_z|v_z\rangle\!\langle v_z|$}, with {\color{black}$w_z>0$} and {\color{black}$\langle v_z|v_z\rangle=1$}.
Put {\color{black}$|\psi\rangle=W^\dagger|v_z\rangle$} and define
{\color{black}\[
 S=\bigl\|\Pi_H|\psi\rangle\bigr\|_2^2,
 \qquad |u\rangle=\frac{\Pi_K|\psi\rangle}{\sqrt{1-S}}.
\]}%
Then $S\sim\operatorname{Beta}(2,d_\perp)$, and $|u\rangle$ is uniform on the
unit sphere of $K$, independently of $S$. The vector $|u\rangle$ is defined
almost surely; write $\E_S$ and $\E_u$ for expectations over these two
variables. Define
{\color{black}$q_{0,z}:=\langle\psi|\rho_0|\psi\rangle$} and
{\color{black}$\delta_{x,z}:=\beta\langle\psi|A_x|\psi\rangle$}.
The Born probabilities (or outcome densities) under {\color{black}$W\rho_0W^\dagger$}
and {\color{black}$W\rho_xW^\dagger$} are {\color{black}$w_zq_{0,z}$} and
{\color{black}$w_z(q_{0,z}+\delta_{x,z})$}, respectively.
Since {\color{black}$A_x$} is supported on {\color{black}$K$}, these definitions give
{\color{black}\begin{align}
 q_{0,z}&=\frac{1-\beta}{2}S+\frac \beta{d_\perp}(1-S),
 \label{eq:q0}\\
\delta_{x,z}&=\beta(1-S)\langle u|A_x|u\rangle.
\label{eq:delta-x}
\end{align}}%
Since $\tr A_x=0$, the Haar projective second moment \cite{Mele2024} and integration against
the beta density {\color{black}$d_\perp(d_\perp+1)s(1-s)^{d_\perp-1}$} give
{\color{black}\begin{align}
 \E_u\langle u|A_x|u\rangle^2
 &=\frac{\tr(A_x^2)}{d_\perp(d_\perp+1)},
 \label{eq:spherical-moment}\\
 \E\frac{(1-S)^2}{S}&=\frac{d_\perp(d_\perp+1)}d.
 \label{eq:beta-moment}
\end{align}}%
The two dimensional head supplies the factor $s$ in this beta density,
which cancels the inverse factor $1/s$ in \eqref{eq:beta-moment}.
Since $q_{0,z}\geq(1-\beta)S/2$, independence and the two moment
identities give
\begin{equation}
 \begin{aligned}
 &\E_W\frac{(w_z\delta_{x,z})^2}{w_zq_{0,z}}\\
 &\quad\leq\frac{2w_z\beta^2}{1-\beta}\,
     \E_S\frac{(1-S)^2}{S}\,\E_u\langle u|A_x|u\rangle^2\\
 &\quad=\frac{2w_z\beta^2}{1-\beta}\,
     \frac{d_\perp(d_\perp+1)}d\,
     \frac{\tr(A_x^2)}{d_\perp(d_\perp+1)}\\
 &\quad=\frac{2w_z\beta^2}{(1-\beta)d}\tr(A_x^2).
 \end{aligned}
 \label{eq:effect-bound}
\end{equation}
The second inequality in Lemma~\ref{lem:one-copy-chi-square} follows from
$\beta\leq1/4$ and $\tr(A_x^2)\leq1/{\rtail}$ in \eqref{eq:A-second-moment}:
\[
 \frac{2\beta^2}{1-\beta}\tr(A_x^2)
 \leq\frac{2\beta^2}{3/4}\,\frac1{\rtail}
 =\frac{8\beta^2}{3{\rtail}}.
\]
To prove the first inequality, fix a pointwise spectral decomposition
{\color{black}\[
 F(z)=\sum_i\lambda_i|v_i\rangle\!\langle v_i|,
 \qquad\lambda_i\geq0,\quad\sum_i\lambda_i=1.
\]}%
The vectors {\color{black}$v_i$} are orthonormal, and their dependence on {\color{black}$z$} is
suppressed. For each unit rank one effect $|v_i\rangle\langle v_i|$, define
the reference probability and perturbation by
\[
 \begin{aligned}
 q_{i,z}&:=\langle v_i|W\rho_0W^\dagger|v_i\rangle>0,\\
 \delta_{i,z}&:=\beta\langle v_i|WA_xW^\dagger|v_i\rangle.
 \end{aligned}
\]
Their dependence on $W$ is suppressed as well.
Thus {\color{black}$p_{0,W}(z)=\sum_i\lambda_iq_{i,z}$} and
{\color{black}$p_{x,W}(z)-p_{0,W}(z)=\sum_i\lambda_i\delta_{i,z}$}.
Weighted Cauchy--Schwarz gives, at each fixed {\color{black}$z$},
{\color{black}\[
 \frac{(\sum_i\lambda_i\delta_{i,z})^2}
      {\sum_i\lambda_iq_{i,z}}
 \leq\sum_i\lambda_i\frac{\delta_{i,z}^2}{q_{i,z}}.
\]}%
No measurable choice of eigenvectors is needed: by Haar invariance,
{\color{black}$\E_W(\delta_{i,z}^2/q_{i,z})$} has the same value for every unit {\color{black}$v_i$}.
Thus its sum weighted by {\color{black}$\lambda_i$} is a constant multiple of the
measurable function {\color{black}$\tr F(z)$}.
Using the definition \eqref{eq:information-divergences}, Tonelli's
theorem, the rank one bound, and {\color{black}$\mathsf M(\Omega)=I_d$}, we conclude
{\color{black}\[
 \begin{aligned}
 &\E_W\chi^2(p_{x,W}\|p_{0,W})\\
 &\quad=\E_W\int_\Omega
       \frac{(p_{x,W}(z)-p_{0,W}(z))^2}{p_{0,W}(z)}\,d\mu(z)\\
 &\quad\leq\int_\Omega\sum_i\lambda_i\E_W
       \frac{\delta_{i,z}^2}{q_{i,z}}\,d\mu(z)\\
 &\quad\leq\int_\Omega\sum_i\lambda_i
       \frac{2\beta^2}{(1-\beta)d}\tr(A_x^2)\,d\mu(z)\\
 &\quad=\frac{2\beta^2}{(1-\beta)d}\tr(A_x^2)
       \int_\Omega\tr F(z)\,d\mu(z)\\
 &\quad=\frac{2\beta^2}{(1-\beta)d}\tr(A_x^2)
       \tr\mathsf M(\Omega)\\
 &\quad=\frac{2\beta^2}{1-\beta}\tr(A_x^2).
 \end{aligned}
\]}
\end{proof}
\end{revision}

\subsection{Proof of Lemma~\ref{lem:nonadaptive-fano}}
\label{sec:nonadaptive-fano-proof}
\begin{proof}
We use the randomized nonadaptive protocol of Section~\ref{sec:risks}.
The hidden label $X$ is uniform on $\{1,\ldots,N\}$, the shared rotation
$W$ has Haar law, and the public design seed $R$ has its prescribed law
$\pi$. These three random variables are mutually independent. For every
copy index $t\in\{1,\ldots,T\}$ and seed value $r$, the selected POVM is
$M_t(r)$ and its outcome is denoted by $Y_t$. Write
$Y^t=(Y_1,\ldots,Y_t)$, with $Y^0$ the empty record.
The measurement choices depend on $R$, and $W$ is revealed only after
all measurements have been completed.

\emph{Outcome laws and conditional information.}
Write $\Omega$ for the outcome space and $\mathcal E$ for a measurable
subset. For fixed $r,t$, use the common trace measure
$\mu_{r,t}(\mathcal E)=\tr M_t(r)(\mathcal E)$ from
Appendix~\ref{sec:povm-density}. Let $p_{x,w,r,t}$ and $p_{0,w,r,t}$
be the outcome densities of $M_t(r)$ under $w\rho_xw^\dagger$ and
$w\rho_0w^\dagger$, respectively. The latter is positive almost
everywhere because $\rho_0\succ0$. Conditional on $(W,R)$, the
uniform mixture density is
\[
 \bar p_{W,R,t}(z):=\frac1N\sum_{x=1}^N p_{x,W,R,t}(z).
\]
All logarithms are natural. For random variables $V_1,V_2,Z$ on
standard Borel spaces, conditional mutual information means
\[
 I(V_1;V_2\mid Z)
 :=\E_Z D_{\KL}\!\left(
 \Prob_{V_1,V_2\mid Z}\,\middle\|\,
 \Prob_{V_1\mid Z}\otimes \Prob_{V_2\mid Z}\right),
\]
where $\Prob_{\cdot\mid Z}$ denotes a regular conditional law. KL divergence
between measures is computed using their densities relative to a common
dominating measure, as in \eqref{eq:information-divergences}. Without
conditioning, $I(V_1;V_2):=D_{\KL}(\Prob_{V_1,V_2}\|\Prob_{V_1}\otimes\Prob_{V_2})$.
For our uniform label, this specializes to
\[
 \begin{aligned}
 &I(X;Y_t\mid W,R)\\
 &\quad=\frac1N\sum_{x=1}^N\E_{W,R}
 D_{\KL}(p_{x,W,R,t}\|\bar p_{W,R,t}).
 \end{aligned}
\]
The same definition with the conditional joint density of $Y^T$ gives
$I(X;Y^T\mid W,R)$. The relevant conditional laws and their
Radon--Nikodym derivatives can be chosen measurably on the standard Borel
spaces; no measurable selection of POVM eigenvectors is used.

\emph{From one copy chi square to information.}
For fixed $(w,r,t)$, expanding the logarithm gives the information radius
identity
\[
 \begin{aligned}
 &\frac1N\sum_x D_{\KL}(p_{x,w,r,t}\|p_{0,w,r,t})\\
 &=\frac1N\sum_x D_{\KL}(p_{x,w,r,t}\|\bar p_{w,r,t})
   +D_{\KL}(\bar p_{w,r,t}\|p_{0,w,r,t}).
 \end{aligned}
\]
Nonnegativity of the last term, $D_{\KL}\leq\chi^2$, and
\eqref{eq:nonadaptive-chi-budget}, applied to each fixed POVM $M_t(r)$,
yield
\begin{equation}
 \begin{aligned}
 I(X;Y_t\mid W,R)
 &\leq\frac1N\sum_x\E_{W,R}
       D_{\KL}(p_{x,W,R,t}\|p_{0,W,R,t})\\
 &\leq\frac1N\sum_x\E_R\E_W
       \chi^2(p_{x,W,R,t}\|p_{0,W,R,t})\leq\kappa.
 \end{aligned}
 \label{eq:one-copy-information}
\end{equation}
Independence of $W$ and $R$ is used here: after fixing $R=r$, the
rotation still has Haar law. The finite bound also justifies all
information radius integrals.

\emph{Accumulating information over the copies.}
Conditional on $(X,W,R)$, nonadaptivity and the independent prepared
copies give the product density
\[
 p_{x,W,R}^{(T)}(y_1,\ldots,y_T)
 =\prod_{t=1}^T p_{x,W,R,t}(y_t)
\]
relative to $\bigotimes_{t=1}^T\mu_{R,t}$. In particular,
$I(Y_t;Y^{t-1}\mid X,W,R)=0$. The interaction information identity
\cite[Chapter~2]{CoverThomas2006} reads
\[
 \begin{aligned}
 I(X;Y_t\mid W,R,Y^{t-1})
 &=I(X;Y_t\mid W,R)\\
 &\quad-I(Y_t;Y^{t-1}\mid W,R)\\
 &\quad+I(Y_t;Y^{t-1}\mid X,W,R).
 \end{aligned}
\]
These terms are finite: conditional dependence between the copies is
mediated by the finite label $X$, whose entropy is at most $\ln N$.
The last term vanishes and the middle term is nonnegative. The chain
rule therefore gives
\begin{equation}
 \begin{aligned}
 I(X;Y^T\mid W,R)
 &=\sum_{t=1}^T I(X;Y_t\mid W,R,Y^{t-1})\\
 &\leq\sum_{t=1}^T I(X;Y_t\mid W,R)\leq T\kappa.
 \end{aligned}
 \label{eq:information-tensorization}
\end{equation}
Since $X$ is independent of $(W,R)$,
$I(X;Y^T,W,R)=I(X;Y^T\mid W,R)$. This uses the conditional independence
principle in \cite[Corollary~2.5 and Section~4]{LoweNayak2025}; here we
condition on the shared rotation and seed and allow standard Borel
outcome spaces.

The same bound follows by applying the information radius inequality
directly to the full record:
\[
 \begin{aligned}
 &I(X;Y^T,W,R)\\
 &\quad\leq\frac1N\sum_x\E_{W,R}
 D_{\KL}\!\left(p_{x,W,R}^{(T)}\|p_{0,W,R}^{(T)}\right)\\
 &\quad=\frac1N\sum_x\sum_{t=1}^T\E_{W,R}
 D_{\KL}(p_{x,W,R,t}\|p_{0,W,R,t})\\
 &\quad\leq T\kappa.
 \end{aligned}
\]
Here $p_{0,W,R}^{(T)}:=\prod_{t=1}^T p_{0,W,R,t}$ is the conditional
reference density. The product identity is applied before averaging
over the shared rotation.

\emph{Fano's inequality.}
For a finite label, write
\[
 H(X\mid V):=-\E_V\sum_x \Prob\{X=x\mid V\}\ln \Prob\{X=x\mid V\},
\]
with $0\ln0=0$. Every possibly randomized decoder of the record
$V=(Y^T,W,R)$ gives a Markov chain $X\to V\to\widehat X$.
Data processing and conditioning on the decoder's error indicator give
\[
 \begin{aligned}
 H(X\mid V)&\leq H(X\mid\widehat X)\\
 &\leq\ln2+\Prob\{\widehat X\neq X\}\ln(N-1)\\
 &\leq\ln2+\Prob\{\widehat X\neq X\}\ln N.
 \end{aligned}
\]
Indeed, the binary error indicator has entropy at most $\ln2$; given
$\widehat X$, a correct decision determines the label, whereas an
incorrect decision leaves at most $N-1$ possibilities. Since $H(X)=\ln N$ and
$I(X;Y^T,W,R)=\ln N-H(X\mid Y^T,W,R)$, rearrangement proves
\begin{equation}
 \begin{aligned}
 \Prob\{\widehat X\neq X\}
 &\geq1-\frac{I(X;Y^T,W,R)+\ln2}{\ln N}\\
 &\geq1-\frac{T\kappa+\ln2}{\ln N}.
 \end{aligned}
 \label{eq:fano-reader-form}
\end{equation}
Under $T\kappa\leq(\ln N)/16$ and $\ln N\geq4\ln2$, the last
expression is at least $1-1/16-1/4=11/16$.
\end{proof}

\emph{Application and conversion to a fixed tomography state.}
Lemma~\ref{lem:one-copy-chi-square} supplies $\kappa=8\beta^2/(3{\rtail})$.
With $c_I$ from \eqref{eq:fano-ci}, the condition
$\beta\leq c_I{\rtail}\sqrt{d_\perp/T}$ gives
\begin{equation}
 T\kappa=\frac{8T\beta^2}{3{\rtail}}
 \leq\frac{8c_I^2{\rtail}d_\perp}{3}
 =\frac{\gamma_0{\rtail}d_\perp}{16}\leq\frac{\ln N}{16}.
 \label{eq:fano-information}
\end{equation}
Also, $d\geq514$ and ${\rtail}\geq1$ give
$\ln N\geq\gamma_0{\rtail}d_\perp\geq25/9>4\ln2$.
For each realized $W=w$, choose the smallest index minimizing
$\|\widehat\rho-w\rho_xw^\dagger\|_{\mathrm{tr}}$.
Continuity of the trace norm and finiteness of the packing make this
choice measurable. If its output is $x'\neq X$, the triangle inequality
and the minimizing property imply
\[
 \frac\beta2
 \leq\|w\rho_{x'}w^\dagger-w\rho_Xw^\dagger\|_{\mathrm{tr}}
 \leq2\|\widehat\rho-w\rho_Xw^\dagger\|_{\mathrm{tr}}.
\]
Thus a label error forces trace loss at least $\beta/4$.
Let $e(x,w)$ be the probability of this trace loss event at the fixed
state $w\rho_xw^\dagger$, including the design seed, outcomes, and
estimator randomness. The preceding bound gives
\[
 \frac1N\sum_x\E_W e(x,W)\geq\frac{11}{16}.
\]
The measurable function $e$ takes values in $[0,1]$. If it were strictly
below $11/16$ for every pair $(x,w)$, its average would be strictly below
$11/16$. Hence some fixed pair $(x_*,w_*)$ satisfies
$e(x_*,w_*)\geq11/16$. Its state $\rho_*=w_*\rho_{x_*}w_*^\dagger$
belongs to the spectral class by unitary invariance and is fixed before
the design seed and outcomes are sampled. Finally, nonnegativity of the
trace loss gives
\[
 \E_{\rho_*}\|\widehat\rho-\rho_*\|_{\mathrm{tr}}
 \geq\frac\beta4\Prob_{\rho_*}\!\left\{
 \|\widehat\rho-\rho_*\|_{\mathrm{tr}}\geq\frac\beta4\right\}
 \geq\frac{11\beta}{64}.
\]
This proves both \eqref{eq:fixed-hard-state-probability} and
\eqref{eq:fixed-hard-state-expectation} for the same fixed state.

\subsection{Packing and integer optimization}
\label{sec:lower-packing-details}

\begin{lemma}[Dimension agnostic Grassmann packing]
\label{lem:grassmann-packing}
For integers {\color{black}$d_\perp\geq3$} and {\color{black}$1\leq {\rtail}\leq d_\perp/3$}, there are rank {\color{black}${\rtail}$}
orthogonal projectors {\color{black}$P_1,\ldots,P_N$} on {\color{black}$\C^{d_\perp}$} such that
{\color{black}\begin{equation}
 \begin{aligned}
 \ln N&\geq\gamma_0{\rtail} d_\perp,\\
 \|P_x-P_{x'}\|_{\mathrm{tr}}&\geq\frac {\rtail}2\quad(x\neq x'),\\
 \gamma_0&:=\frac{25}{4608}.
 \end{aligned}
 \label{eq:projector-packing}
\end{equation}}
\end{lemma}

\begin{proof}
We use only the metric half of
{\color{black}\cite[Lemma~2, Eq.~(19)]{FlammiaGrossLiuEisert2012}}; its Pauli flatness conclusion is
irrelevant.  That proof, applied to normalized rank {\color{black}${\rtail}$} projectors on
{\color{black}$\C^{d_\perp}$}, says that from any list of
{\color{black}$s<\exp\{c(\varepsilon){\rtail} d_\perp\}$} such projectors one can choose another one
whose half trace norm distance from every member is at least {\color{black}$\varepsilon$},
where
{\color{black}\[
 c(\varepsilon)=\frac{\ln(8/\pi)}{2{\rtail} d_\perp}
 +\frac1{32}\left[(1-{\rtail}/d_\perp)-\varepsilon\right]^2.
\]}%
The metric proof samples real orthogonal conjugates using Haar measure on
{\color{black}$\mathsf{SO}(d_\perp)$}; these are also complex Hermitian projectors, so this
part is valid for every integer {\color{black}$d_\perp$}.

Take {\color{black}$\varepsilon=1/4$}.  Since {\color{black}${\rtail}/d_\perp\leq1/3$},
{\color{black}\[
 c(1/4)\geq\frac1{32}\left(\frac23-\frac14\right)^2
 =\frac{25}{4608}=\gamma_0.
\]}%
Set {\color{black}$N=\lceil e^{\gamma_0{\rtail} d_\perp}\rceil$}.  Construct the projectors
successively.  Before selecting the last one, the number already chosen is
at most
{\color{black}\[
 N-1<e^{\gamma_0{\rtail} d_\perp}\leq e^{c(1/4){\rtail} d_\perp},
\]}%
so the strict cardinality hypothesis of the quoted metric lemma holds at
every step.  It gives
{\color{black}$(2{\rtail})^{-1}\|P_x-P_{x'}\|_{\mathrm{tr}}\geq1/4$}, while the definition of {\color{black}$N$} gives
{\color{black}$\ln N\geq\gamma_0{\rtail} d_\perp$}.  This proves \eqref{eq:projector-packing}.
\end{proof}

\begin{lemma}[Discrete lower envelope]
\label{lem:discrete-lower-envelope}
For an integer {\color{black}$M\geq1$}, {\color{black}$L\geq1$}, {\color{black}$\eta>0$}, and {\color{black}$\alpha>1$}, put
{\color{black}\[
 G_M(L,\eta):=\max_{1\leq {\rtail}\leq M}
 \min\{1,L{\rtail}^{1-\alpha},{\rtail}\eta\}
\]}%
Then
{\color{black}\begin{equation}
 \begin{aligned}
 G_M(L,\eta)&\geq\frac12\min\{1,L^{1/\alpha}\eta^{(\alpha-1)/\alpha},M\eta\},\\
 G_M(L,\eta)&\leq\min\{1,L^{1/\alpha}\eta^{(\alpha-1)/\alpha},M\eta\}.
 \end{aligned}
 \label{eq:discrete-optimization}
\end{equation}}
\end{lemma}

\begin{proof}
For every {\color{black}${\rtail}$}, the weighted geometric mean inequality gives
{\color{black}\[
 \begin{aligned}
 \min\{L{\rtail}^{1-\alpha},{\rtail}\eta\}
 &\leq (L{\rtail}^{1-\alpha})^{1/\alpha}\\[-2pt]
 &\qquad{}\cdot({\rtail}\eta)^{(\alpha-1)/\alpha}\\
 &=L^{1/\alpha}\eta^{(\alpha-1)/\alpha}.
 \end{aligned}
\]}%
Together with {\color{black}${\rtail}\eta\leq M\eta$} and the cap one, this proves the upper
comparison.

For the reverse comparison, set
{\color{black}\begin{equation}
 {\packbal}=(L/\eta)^{1/\alpha}.
 \label{eq:effective-rank}
\end{equation}}%
If {\color{black}${\packbal}\leq1$}, then {\color{black}$\eta\geq L\geq1$}, and {\color{black}${\rtail}=1$} gives
{\color{black}$G_M=1$}, matching both bounds. If {\color{black}$1<{\packbal}\leq M$}, take {\color{black}${\rtail}=\lfloor {\packbal}\rfloor$}; then
{\color{black}${\packbal}/2\leq {\rtail}\leq {\packbal}$},
{\color{black}\[
 {\rtail}\eta\geq\frac12{\packbal}\eta,
 \qquad L{\rtail}^{1-\alpha}\geq L{\packbal}^{1-\alpha}={\packbal}\eta,
\]}%
and {\color{black}${\packbal}\eta=L^{1/\alpha}\eta^{(\alpha-1)/\alpha}\leq M\eta$}.
This gives at least half the displayed minimum. Finally, if {\color{black}${\packbal}>M$}, take {\color{black}${\rtail}=M$}.
Then {\color{black}$LM^{1-\alpha}>M\eta$} and
{\color{black}$L^{1/\alpha}\eta^{(\alpha-1)/\alpha}={\packbal}\eta>M\eta$}, so the value is
exactly {\color{black}$\min\{1,M\eta\}$}. These cases prove \eqref{eq:discrete-optimization}.
\end{proof}

\subsection{Proof of Corollary~\ref{cor:exact-rank}}
\label{app:rank-proof}
\begin{proof}
\emph{Upper bound and admissibility.}
For {\color{black}$\rho\in\cD_{d,r}$}, {\color{black}$\tau_r(\rho)=\sum_{j>r}\lambda_j(\rho)=0$}.
Thus, for {\color{black}$T\geq2d$}, \eqref{eq:polynomial-expected-oracle} at {\color{black}$s=r$} gives
{\color{black}\[
 \E_\rho\|\rhoCF-\rho\|_{\mathrm{tr}}
 \leq4\tau_r(\rho)+B_0r\sqrt{d/T}
 =B_0r\sqrt{d/T}.
\]}%
Set {\color{black}$C_{\rm rank}=\max\{4,B_0\}$}. For {\color{black}$1\leq T<2d$},
{\color{black}$r\sqrt{d/T}\geq1/2$}, so
{\color{black}$2\leq4r\sqrt{d/T}\leq C_{\rm rank}r\sqrt{d/T}$}.
Combining these two cases with the loss cap {\color{black}$2$} and taking the supremum
over {\color{black}$\cD_{d,r}$} gives the upper bound in \eqref{eq:exact-rank-risk}.
For weighted PLS, \eqref{eq:mw-expected-oracle} at {\color{black}$s=r$} gives
{\color{black}$\E_\rho\|\rhoMW-\rho\|_{\mathrm{tr}}\leq B_{\rm MW}r\sqrt{d/T}$}
when {\color{black}$T\geq2d$}. Repeating the same argument with
{\color{black}$C_{\rm rank}^{\rm MW}=\max\{4,B_{\rm MW}\}$} proves its upper bound.
The middle inequality follows from admissibility of each estimator
in the periodic experiment.

\begin{revision}
\emph{Lower bound for {\color{black}$r\geq3$}.}
For the lower bound, fix any admissible randomized design and reconstruction
kernel in \eqref{eq:minimax-risk}. Put {\color{black}$\gamma=\min\{1/4,c_I\}$}, where
{\color{black}$c_I$} is defined in \eqref{eq:fano-ci}. We first suppose {\color{black}$r\geq3$} and set
{\color{black}\[
 \begin{gathered}
 d_\perp=d-2,\qquad
 {\rtail}=\min\{r-2,\lfloor d_\perp/3\rfloor\}.
 \end{gathered}
\]}%
Condition \eqref{eq:headline-block-condition} gives {\color{black}$d\geq65536$}, so
{\color{black}$1\leq {\rtail}\leq d_\perp/3$}, {\color{black}${\rtail}+2\leq r$}, and {\color{black}$r\leq6{\rtail}$}.
Use \eqref{eq:hard-states} with
{\color{black}$\beta=\gamma\min\{1,{\rtail}\sqrt{d_\perp/T}\}$}. This state has rank {\color{black}${\rtail}+2$};
the construction therefore applies only from rank three onward.
The information and testing argument through
\eqref{eq:fixed-hard-state-expectation} uses only {\color{black}$\beta\leq1/4$} and
{\color{black}$\beta\leq c_I{\rtail}\sqrt{d_\perp/T}$}; the additional spectral tail constraint
in \eqref{eq:b-class-condition} is unnecessary for a rank class.
Consequently every admissible procedure has worst case expected loss
at least {\color{black}$11\beta/64$}. Since {\color{black}$\sqrt d\leq2\sqrt{d_\perp}$},
{\color{black}\[
 \min\{1,r\sqrt{d/T}\}
 \leq12\min\{1,{\rtail}\sqrt{d_\perp/T}\}.
\]}%
Taking the infimum over procedures proves the lower bound for {\color{black}$r\geq3$}
with coefficient {\color{black}$11\gamma/768$}.

\end{revision}
\emph{Rank one.}
The finite pairing below is an explicit application of Assouad's method
\cite[Lemmas~1--2]{Yu1997}; its measurement step is proved here for the randomized
nonadaptive protocol class in \eqref{eq:minimax-risk}.
This rank one bound holds for every {\color{black}$d\geq2$} and {\color{black}$T\geq1$}.

Now put {\color{black}${\dcube}=d-1$}, {\color{black}$K=\max\{{\dcube},T\}$}, {\color{black}$a=(8\sqrt K)^{-1}$}, and
{\color{black}$h=\sqrt{1-{\dcube}a^2}$}. Then {\color{black}${\dcube}a^2,Ta^2\leq1/64$} and {\color{black}$h\geq1/2$}.
For {\color{black}$\theta\in\{-1,1\}^{\dcube}$}, define the pure states
{\color{black}\begin{equation}
 |\psi_\theta\rangle=h|0\rangle+a\sum_{j=1}^{\dcube}\theta_j|j\rangle,
 \qquad \rho_\theta=|\psi_\theta\rangle\langle\psi_\theta|.
 \label{eq:pure-cube-states}
\end{equation}}%
These vectors have norm one. If {\color{black}$\theta^{(j)}$} flips bit {\color{black}$j$}, then
{\color{black}$\langle\psi_\theta|\psi_{\theta^{(j)}}\rangle=1-2a^2$}.
For {\color{black}$|\Psi_\theta\rangle=|\psi_\theta\rangle^{\otimes T}$}, Bernoulli's
inequality therefore gives
\begin{equation}
 \big\||\Psi_\theta\rangle-|\Psi_{\theta^{(j)}}\rangle\big\|_2^2
 =2\bigl[1-(1-2a^2)^T\bigr]\leq4Ta^2.
 \label{eq:pure-cube-tensor-distance}
\end{equation}

Fix a seed value {\color{black}$R=s$} and a measurable randomized binary test {\color{black}$g_s$}
of the transcript, with {\color{black}$0\leq g_s\leq1$}. Let {\color{black}$N_s$} be the product of the
{\color{black}$T$} selected single copy POVMs. The entrywise integral
{\color{black}$E_s=\int g_s(y)N_s(dy)$} satisfies {\color{black}$0\preceq E_s\preceq I$} and gives
the test's acceptance probability as
{\color{black}$\langle\Psi_\theta|E_s|\Psi_\theta\rangle$}.
This follows from positivity, normalization, and finite dimensional
integration, for arbitrary standard Borel outcomes. For any unit vectors
{\color{black}$|u\rangle,|v\rangle$}, expanding the difference of quadratic forms and
using {\color{black}$\|E_s\|_{\mathrm{op}}\leq1$} yields
{\color{black}\begin{equation}
 \begin{aligned}
 &\big|\langle u|E_s|u\rangle-\langle v|E_s|v\rangle\big|\\
 &\quad\leq
 \big|\langle u-v|E_s|u\rangle\big|
 +\big|\langle v|E_s|u-v\rangle\big|\\
 &\quad\leq2\big\||u\rangle-|v\rangle\big\|_2.
 \end{aligned}
 \label{eq:pure-cube-binary-effect}
\end{equation}}%
Thus neighboring acceptance probabilities differ by at most
{\color{black}$4a\sqrt T\leq1/2$}. This bound is uniform in the seed and the test.

For an arbitrary density matrix estimate {\color{black}$\sigma$}, decode bit {\color{black}$j$} as
{\color{black}$+1$} if {\color{black}$\operatorname{Re}\sigma_{j0}\geq0$} and as {\color{black}$-1$} otherwise.
This is a measurable decoder, including the stated tie rule.
Write {\color{black}$H_\theta(\sigma)$} for its number of wrong bits.
Because {\color{black}$(\rho_\theta)_{j0}=ha\theta_j$}, each wrong bit forces
{\color{black}$|\sigma_{j0}-(\rho_\theta)_{j0}|\geq ha$}. Cauchy--Schwarz gives
{\color{black}\begin{equation}
 \begin{aligned}
 haH_\theta(\sigma)
 &\leq\sum_{j=1}^{\dcube}|\sigma_{j0}-(\rho_\theta)_{j0}|\\
 &\leq\sqrt {\dcube}\,\|\sigma-\rho_\theta\|_{\mathrm F}
 \leq\sqrt {\dcube}\,\|\sigma-\rho_\theta\|_{\mathrm{tr}}.
 \end{aligned}
 \label{eq:pure-cube-hamming-loss}
\end{equation}}%
In particular, the estimate may be mixed.

Let {\color{black}$g_{s,j}(y)$} be the reconstruction kernel's probability that the
estimated bit {\color{black}$j$} equals {\color{black}$+1$}. It is measurable, lies in {\color{black}$[0,1]$}, and is
independent of the true state. If {\color{black}$\theta_j=+1$}, the conditional bit error
probabilities for {\color{black}$\theta$} and {\color{black}$\theta^{(j)}$} sum to
{\color{black}$1-\E_\theta[g_{s,j}(Y)\mid R=s]
 +\E_{\theta^{(j)}}[g_{s,j}(Y)\mid R=s]\geq1/2$}.
Pairing the {\color{black}$2^{\dcube}$} vertices along each coordinate gives average bit error
at least {\color{black}$1/4$}. Integrate over the state independent seed law and sum over
coordinates to obtain
{\color{black}\begin{equation}
 2^{-{\dcube}}\sum_\theta\E_\theta H_\theta(\widehat\rho)\geq {\dcube}/4.
 \label{eq:pure-cube-average-hamming}
\end{equation}}%
These expectations include both design and reconstruction randomness.
Only the original measurable error probabilities are integrated over
seeds; no jointly measurable choice of effect densities is required.
Combining \eqref{eq:pure-cube-hamming-loss} and
\eqref{eq:pure-cube-average-hamming} yields
{\color{black}\begin{equation}
 \begin{aligned}
 2^{-{\dcube}}\sum_\theta\E_\theta\|\widehat\rho-\rho_\theta\|_{\mathrm{tr}}
 &\geq\frac{ha\sqrt {\dcube}}{4}\\
 &\geq\frac1{64}\sqrt{\frac {\dcube}{\max\{{\dcube},T\}}}\\
 &\geq\frac1{128}\min\{1,\sqrt{d/T}\},
 \end{aligned}
 \label{eq:pure-cube-risk}
\end{equation}}%
where the last comparison uses {\color{black}${\dcube}=d-1\geq d/2$}.
For each procedure, its supremum over pure states dominates this finite
average. Taking the infimum over procedures proves the rank one bound.

\emph{Rank two and completion.}
Since {\color{black}$\cD_{d,1}\subseteq\cD_{d,2}$} and
{\color{black}$\min\{1,2x\}\leq2\min\{1,x\}$} for {\color{black}$x\geq0$}, we obtain
{\color{black}\begin{equation}
 \mathfrak R_T(\cD_{d,2})
 \geq\mathfrak R_T(\cD_{d,1})
 \geq\frac1{256}\min\{1,2\sqrt{d/T}\}.
 \label{eq:pure-cube-rank-two}
\end{equation}}%
Combining these lower bounds with the upper bound, and taking
{\color{black}$c_{\rm rank}=\min\{11\gamma/768,1/256\}$}, proves
\eqref{eq:exact-rank-risk} for every {\color{black}$1\leq r\leq d$}.
\end{proof}

\section{Spectral Projection and Channel Evaluation}
\label{app:channel-algebra}
\subsection{Dense spectral projection and PLS}
\label{sec:shared-spectral-projection}
All density projections use a full Hermitian eigendecomposition.
Writing {\color{black}${\oparg{}}={\eigbasis{}}\operatorname{diag}(\lambda_i){\eigbasis{}}^*$} gives
{\color{black}\[
 \Pi_{\cD_d}^{\F}({\oparg{}})
 ={\eigbasis{}}\operatorname{diag}\bigl((\lambda_i-\theta)_+\bigr)_{i=1}^d{\eigbasis{}}^*,
 \]}%
where {\color{black}$\sum_{i=1}^d(\lambda_i-\theta)_+=1$}.
Unitary invariance reduces the minimization to probability simplex projection
\cite{GutaKahnKuengTropp2020,DuchiEtAl2008}. With eigenvalues in decreasing
order, compute
\begin{revision}
\[
 j_*=\max\left\{j:\lambda_j>
 \frac{\sum_{i=1}^j\lambda_i-1}{j}\right\},\qquad
 \theta=\frac{\sum_{i=1}^{j_*}\lambda_i-1}{j_*}.
\]
\end{revision}
The data determine both $\theta$ and the output rank. The implementation
and its dense arithmetic costs are given in Supplement
Section~\ref*{sec:pls-dense-computation}.
For trace one {\color{black}${\oparg{}}$}, $\theta\geq0$, with equality exactly when {\color{black}${\oparg{}}\succeq0$}.
For $\rho\in\cD_d$, Weyl's inequality gives
{\color{black}$\sum_i(\lambda_i-\|{\oparg{}}-\rho\|_{\op})_+\leq1$} and
{\color{black}$\lambda_{\min}({\oparg{}})\geq-\|{\oparg{}}-\rho\|_{\op}$}.
Hence {\color{black}$\theta\leq\|{\oparg{}}-\rho\|_{\op}$} and
{\color{black}\[
 \begin{aligned}
 \|\Pi_{\cD_d}^{\F}({\oparg{}})-{\oparg{}}\|_{\op}&\leq\|{\oparg{}}-\rho\|_{\op},\\
 \|\Pi_{\cD_d}^{\F}({\oparg{}})-\rho\|_{\op}&\leq2\|{\oparg{}}-\rho\|_{\op}.
 \end{aligned}
\]}%
Frobenius nonexpansiveness gives the other projection bound used in
Appendix~\ref{sec:pls-comparison}.

\begin{revision}
\subsection{Known channel transfer formula}
For phase free Pauli matrices $P$, the channel satisfies
$\cM(P)=\measeig{P}P$, $\cL(P)=(d+1)\measeig{P}P$ for $P\ne I$, and $\cL(I)=I$.
The next calculation reproduces the known transfer representation; no
novel channel computation claim is made. Earlier explicit brickwork shadow
channels are studied in \cite{ArienzoEtAl2023}.

\begin{lemma}[Cho and Kim's transfer formula]
\label{lem:forward-two-state-transfer}
The following is \cite[Appendix~G, Eqs.~(149)--(152)]{ChoKim2026},
written in local activity variables.
Put $m=n/k$ and $Q=2^k$. Mark a first layer block
as active when the restriction of $P$ to that block is nonidentity, and let
{\color{black}${\activity{}}_i\in\{0,1\}$} be its activity indicator. For
{\color{black}${b_{\mathrm L}},{b_{\mathrm R}}\in\{0,1\}$} define
{\color{black}\begin{equation}
\begin{aligned}
 &\actweights_i(b_{\mathrm L},b_{\mathrm R})\\
 &\quad=\begin{cases}
 \mathbf1\{b_{\mathrm L}=b_{\mathrm R}=0\},&\activity_i=0,\\
 0,&\activity_i=1,\ (b_{\mathrm L},b_{\mathrm R})=(0,0),\\
 (Q+1)^{-1},&\activity_i=1,\ b_{\mathrm L}+b_{\mathrm R}=1,\\
 (Q-1)(Q+1)^{-1},&\activity_i=1,\ (b_{\mathrm L},b_{\mathrm R})=(1,1).
 \end{cases}
 \end{aligned}
 \label{eq:forward-local-weights}
\end{equation}}%
Let the $2$ by $2$ rational matrix {\color{black}${\transfermat{}}_i$} have entries
{\color{black}\begin{equation}
 {\transfermat{}}_i({b_{\mathrm L}},{b_{\mathrm L}}')=
 \sum_{{b_{\mathrm R}}\in\{0,1\}}{\actweights{}}_i({b_{\mathrm L}},{b_{\mathrm R}})
 (Q+1)^{-\mathbf1\{{b_{\mathrm R}}\lor{b_{\mathrm L}}'\}}.
 \label{eq:forward-transfer-matrix}
\end{equation}}%
In particular the two possible matrices are exactly
{\color{black}\begin{equation}
 \begin{aligned}
 {\transfermat{}}^{(0)}&=\begin{pmatrix}1&(Q+1)^{-1}\\0&0\end{pmatrix},\\
 {\transfermat{}}^{(1)}&=\frac1{(Q+1)^2}\begin{pmatrix}1&1\\2Q&Q\end{pmatrix}.
 \end{aligned}
 \label{eq:known-transfer-matrices}
\end{equation}}%
Set {\color{black}${\transfermat{}}_i={\transfermat{}}^{({\activity{}}_i)}$}. With cyclic indices,
{\color{black}\begin{equation}
 \measeig{P}=\tr({\transfermat{}}_1{\transfermat{}}_2\cdots {\transfermat{}}_m).
 \label{eq:forward-pauli-multiplier}
\end{equation}}%
Consequently all $d^2$ multipliers are computable using
$O(d^2n/k)$ rational arithmetic operations, in addition to the
$O(d^2n)$ bit inspections needed to read the Pauli labels, and polynomial
bit length.
\end{lemma}

\begin{proof}
On an active block the nonidentity Pauli is uniform over $Q^2-1$ labels.
The three half block activity probabilities give
\eqref{eq:forward-local-weights}; an active second layer block is diagonal
with probability $(Q+1)^{-1}$. Summing the cyclic activities gives the
matrix trace. Supplement Section~\ref*{supp:transfer} gives the full count.
\end{proof}
\end{revision}
\section{Quantitative Covariance and Oracle Bounds}
\label{app:quantitative}
\begin{revision}
This appendix supplies the sample threshold argument behind the covariance
rate and then records explicit finite sample refinements. The sharp rate
constant below is universal but is not assigned a numerical value; the
separate bounds in Proposition~\ref{prop:explicit-raw-constants} have explicit coefficients.

\subsection{Fixed norm covariance and the sample threshold}
Let $Z$ be a centered real random vector with covariance $\Sigma\ne0$,
$\|Z\|^2=\tr\Sigma$ almost surely, and $L_6$ to $L_2$ ratio at most {\color{black}${\kappa_{6,2}}$}.
Put {\color{black}${r_{\mathrm{eff}}}=\tr\Sigma/\|\Sigma\|_{\op}\geq1$} and
$\widehat\Sigma_t=t^{-1}\sum_{i=1}^tZ_iZ_i^{\mathsf T}$.
Choose $c_{\rm BY},C_{\rm BY}\geq1$, depending only on {\color{black}${\kappa_{6,2}}$}, so that
\cite[Theorem~2, arXiv v3]{AbdallaZhivotovskiy2023} gives
{\color{black}\[
 \E\|\widehat\Sigma_{N_{\mathrm{pad}}}-\Sigma\|_{\op}
 \leq C_{\rm BY}\|\Sigma\|_{\op}\sqrt{{r_{\mathrm{eff}}}/{N_{\mathrm{pad}}}}
 \]}%
for {\color{black}$N_{\mathrm{pad}}\geq c_{\rm BY}r_{\mathrm{eff}}$}.
The source theorem uses radial truncation. Here its cutoff is at least
$\|Z\|$ whenever {\color{black}${N_{\mathrm{pad}}}\geq {r_{\mathrm{eff}}}$}, since
{\color{black}$(\tr\Sigma)^2\leq {N_{\mathrm{pad}}}\tr\Sigma\|\Sigma\|_{\op}$}, so it returns the raw average.

To cover every integer {\color{black}$t\geq {r_{\mathrm{eff}}}$}, append independent observations up to
{\color{black}${N_{\mathrm{pad}}}=\max\{t,\lceil c_{\rm BY}{r_{\mathrm{eff}}}\rceil\}$}. Conditional expectation gives
{\color{black}\[
 \E[\widehat\Sigma_{N_{\mathrm{pad}}}-\Sigma\mid Z_1,\ldots,Z_t]
 =\frac t{N_{\mathrm{pad}}}(\widehat\Sigma_t-\Sigma).
\]}%
By conditional Jensen and {\color{black}${N_{\mathrm{pad}}}/t\leq\max\{1,c_{\rm BY}+1\}$},
{\color{black}\begin{equation}
 \begin{aligned}
 \E\|\widehat\Sigma_t-\Sigma\|_{\op}
 &\leq C_{\rm fn}({\kappa_{6,2}})\|\Sigma\|_{\op}\sqrt{{r_{\mathrm{eff}}}/t},\\
 C_{\rm fn}({\kappa_{6,2}})&:=C_{\rm BY}\sqrt{\max\{1,c_{\rm BY}+1\}}.
 \end{aligned}
 \label{eq:fixed-norm-padding}
\end{equation}}%
Appending observations is only a proof device; the estimator still uses
the original $t$ samples.

For the periodic ensemble,
Proposition~\ref{prop:periodic-low-order} and Theorem~\ref{thm:ck-fourth-projector} give
{\color{black}\begin{equation}
 \begin{aligned}
 a&\geq a_*,\qquad {\momconst{}}=34,\\
 {\covlo{}}&=\frac{19627}{50000},\qquad {\covhi{}}=\frac{80373}{25000}.
 \end{aligned}
 \label{eq:periodic-primitive-constants}
\end{equation}}%
The real sixth power ratio in Proposition~\ref{prop:primitive-covariance}
is at most {\color{black}$136/{\covlo{}}^3\leq(29/8)^6$}. Its real covariance has trace $d+1$,
operator norm at most $80373/50000$, and effective rank at most $2d$.
For $d\geq65536$,
{\color{black}\[
 \|\Sigma_{\R}\|_{\op}\sqrt{{r_{\mathrm{eff}}}/T}
 =\sqrt{\|\Sigma_{\R}\|_{\op}(d+1)/T}
 \leq\frac{127}{100}\sqrt{d/T}.
\]}%
Realification costs a factor two. Thus \eqref{eq:forward-covariance-rate}
holds for every $T\geq2d$, with $c_0=2$ and
{\color{black}\begin{equation}
 {\covrate{}}=\covstar{}:=\frac{127}{50}C_{\rm fn}(29/8).
 \label{eq:audited-covariance-constant}
\end{equation}}%
The oracle coefficients may consequently be chosen as
{\color{black}\begin{equation}
 B_0=\frac{8\covstar{}+4}{a_*},\qquad C_{\rm rank}=\max\{4,B_0\}.
 \label{eq:oracle-risk-coefficients}
\end{equation}}
\end{revision}

\subsection{Integer optimization of the upper rate}
\label{sec:upper-rate-optimization}
For an estimator $\widehat\rho$, write
$\ell_{\mathrm{loss}}:=\E_\rho\|\widehat\rho-\rho\|_{\mathrm{tr}}$.

\begin{lemma}[Common spectral optimization]
\label{lem:common-spectral-optimization}
Let $d\geq1$ be an integer, $L\geq1$, $\alpha>1$, and $\eta>0$. Suppose a loss
{\color{black}${\ell_{\mathrm{loss}}}\leq2$} obeys {\color{black}${\ell_{\mathrm{loss}}}\leq4\tau_s+4\eta s$} for $1\leq s\leq d$,
where $\tau_s\leq Ls^{1-\alpha}$ and $\tau_d=0$. Then
{\color{black}\[
 {\ell_{\mathrm{loss}}}\leq10\min\{1,L^{1/\alpha}\eta^{(\alpha-1)/\alpha},d\eta\}.
\]}
\end{lemma}
\begin{proof}
Put {\color{black}${\rbal{}}=(L/\eta)^{1/\alpha}$}, so that
{\color{black}$\eta {\rbal{}}=L^{1/\alpha}\eta^{(\alpha-1)/\alpha}$}.
If {\color{black}$1\leq {\rbal{}}\leq d$}, choose {\color{black}$s=\lceil {\rbal{}}\rceil$}.
For {\color{black}${\rbal{}}\geq2$}, {\color{black}$s\leq3{\rbal{}}/2$} gives an error at most {\color{black}$10\eta {\rbal{}}$}.
For {\color{black}${\rbal{}}<2$} and $\alpha\leq2$, {\color{black}$\eta {\rbal{}}=L{\rbal{}}^{1-\alpha}\geq1/2$}, so the loss cap
suffices. For $\alpha\geq2$, use {\color{black}$1\leq s/{\rbal{}}\leq2$}: the oracle is at most
{\color{black}$4\eta {\rbal{}}({\rbal{}}/s+s/{\rbal{}})\leq10\eta {\rbal{}}$}.
If {\color{black}${\rbal{}}>d$}, take $s=d$ and use $\tau_d=0$. If {\color{black}${\rbal{}}<1$}, then {\color{black}$\eta {\rbal{}}\geq1$} and
the loss cap applies. Thus {\color{black}${\ell_{\mathrm{loss}}}\leq10\eta {\rbal{}}$} in every case; combining this
with {\color{black}${\ell_{\mathrm{loss}}}\leq2$} and {\color{black}${\ell_{\mathrm{loss}}}\leq4d\eta$} proves the claim.
\end{proof}

For OMD take $\eta=(B_0/4)\sqrt{d/T}$. Rescaling and including $T<2d$
gives the universal upper coefficient in Theorem~\ref{thm:headline-ck},
\begin{equation}
 C_*:=\max\left\{4,\;10\max\left\{1,\frac{2\covstar{}+1}{a_*}\right\}\right\}.
 \label{eq:audited-upper-constant}
\end{equation}
Weighted PLS uses the same lemma with $B_{\rm MW}$ in place of $B_0$.
The PLS comparison uses its logarithmic noise scale and the separate
Frobenius bound for the full state space. Thus integer optimization is
performed once; only the oracle coefficient and noise scale change.

\subsection{Direct finite sample refinements}
\begin{revision}
The general covariance proof gives the sharp order with a universal prefactor.
\end{revision} The same raw average estimator also has the following bounds,
which are useful when numerical constants matter.

\begin{proposition}[Explicit raw average bounds]
\label{prop:explicit-raw-constants}
Under the periodic block condition \eqref{eq:headline-block-condition}, for every $T\geq1$,
\begin{equation}
 \E_\rho\|\overline Y_T-\cL(\rho)\|_{\op}
 \leq\frac{82}{5}\sqrt{d/T}+\frac{8d(d+1)}{T}.
 \label{eq:explicit-two-term-covariance}
\end{equation}
If $T\geq d^3$, the right side can instead be $(33/2)\sqrt{d/T}$. Consequently the
same approximate OMD fit satisfies, for every $1\leq s\leq d$,
\begin{equation}
 \begin{aligned}
 &\E_\rho\|\rhoCF-\rho\|_{\mathrm{tr}}\\
 &\quad\leq\min\left\{2,\;4\tau_s(\rho)
          +136s\sqrt{d/T}+\frac{65sd(d+1)}{T}\right\}.
 \end{aligned}
 \label{eq:explicit-two-term-oracle}
\end{equation}
For $T\geq d^3$, the spectral class upper constant can be $341$, and
states of rank at most $r$ have expected loss at most $\min\{2,137r\sqrt{d/T}\}$.
\end{proposition}
\begin{proof}
Take a $1/16$-net $\mathcal N$ of the complex unit sphere with
$|\mathcal N|\leq33^{2d}$. Each directional score lies in $[-1,d]$, as does
its mean. Thus its centered absolute value is at most $d+1$; by
Proposition~\ref{prop:periodic-low-order}, \eqref{eq:scalar-variance},
its variance is at most $\kappa_*\leq49/5$.
For the maximum absolute empirical directional error $\eta_{\mathcal N}$,
independence and {\color{black}$e^{t_{\mathrm{mgf}}}\leq1+{t_{\mathrm{mgf}}}+(3/4){t_{\mathrm{mgf}}}^2$} for {\color{black}$|{t_{\mathrm{mgf}}}|\leq1$} give
{\color{black}\begin{equation}
 \begin{aligned}
 \E\eta_{\mathcal N}&\leq\frac{\ln(2|\mathcal N|)}{{\lambda_{\mathrm{Lap}}}}+\frac{147{\lambda_{\mathrm{Lap}}}}{20T},\\
 &0<{\lambda_{\mathrm{Lap}}},\qquad (d+1){\lambda_{\mathrm{Lap}}}\leq T.
 \end{aligned}
 \label{eq:explicit-net-laplace}
\end{equation}}%
Each signed average has exponential moment at most {\color{black}$e^{147{\lambda_{\mathrm{Lap}}}^2/(20T)}$}.
Sum over signs and directions, then apply Jensen's inequality.
Here $d\geq65536$ gives $\ln(2|\mathcal N|)\leq7d$.
Choose {\color{black}${\lambda_{\mathrm{Lap}}}=T/[d+1+\sqrt{T/d}]$}. The two terms in
\eqref{eq:explicit-net-laplace} are at most
$7d(d+1)/T+7\sqrt{d/T}$ and $(147/20)\sqrt{d/T}$, respectively.
The net inequality
$\|A\|_{\op}\leq(8/7)\max_{|v\rangle\in\mathcal N}|\langle v|A|v\rangle|$
proves \eqref{eq:explicit-two-term-covariance}.
If $T\geq d^3$, take {\color{black}${\lambda_{\mathrm{Lap}}}=(19/20)\sqrt{dT}$}. This is allowed because
$(19/20)(d+1)\leq d$, and gives
$\E\eta_{\mathcal N}\leq(140/19+2793/400)\sqrt{d/T}
\leq(231/16)\sqrt{d/T}$.
The net inequality now gives coefficient $33/2$.
Apply \eqref{eq:polynomial-forward-oracle}, with
$\delta_{\rm alg}\leq\sqrt{d/T}$, and round upward:
$676/(5a_*)<136$, $64/a_*<65$, $340/a_*<341$, and $136/a_*<137$.
For the spectral bound use the factor ten rate optimization above; for
the rank bound take $s=r$ and $\tau_r=0$.
This net is only a proof device; the raw average estimator is unchanged.
\end{proof}

\section{Projected Least Squares Comparison}
\label{sec:pls-comparison}
This appendix proves the risk bounds reported in the PLS row of
Table~\ref{tab:estimator-comparison} for the common periodic measurement records.

\subsection{Statistical comparison}
\begin{finalrevision}
For the PLS estimator \eqref{eq:candidate-one}, we prove the inverse channel
concentration bound
\begin{equation}
 \E_\rho\|A_T-\rho\|_{\op}
 \leq C\left(\sqrt{\frac{d\ln d}{T}}+\frac{d\ln d}{T}\right)
 \label{eq:pls-operator-comparison}
\end{equation}
under \eqref{eq:headline-block-condition}. The constant is universal and
the bound is uniform over all density matrices. We follow the Pauli
expansion in \cite[Appendix~D]{ChoKim2026}, giving the local counting
argument explicitly below, and then use the expectation form of matrix
Bernstein \cite[Theorem~6.1.1, Eq.~(6.1.3)]{Tropp2015}.

\emph{Inverse snapshot second moment.}
Write $X=\cM^{-1}(B)$ and retain $q=2^{k/2}$, $Q=q^2$, $m=n/k$,
and $d=Q^m$ from Section~\ref{sec:fourth-moment}.
For a phase free Pauli label $p$, let $w_j(p)$ count its nonidentity
blocks in layer $j$. Write $\lambda_p=\measeig{P_p}$ for the eigenvalue
of $\cM$ on the Pauli matrix $P_p$. A nonidentity Pauli becomes diagonal under a uniform
$k$ qubit Clifford with probability $1/(Q+1)$: among the $Q^2-1$
nonidentity Paulis, exactly $Q-1$ are diagonal. Independence across
second layer blocks therefore gives
\[
 \lambda_p\geq(Q+1)^{-m}.
\]
Moreover, $\lambda_p$ depends only on the first layer block support of
$p$. Thus channel inversion commutes with first layer Clifford
conjugation.

For a fixed second layer circuit, let $S$ be the phase free Pauli
subgroup that it maps to diagonal Paulis. It has $|S|=d$. Put
\[
 \mathcal T_{1,\rho}(P_r)
 =\E_{U_1}\!\left[
   \tr(\rho U_1^\dagger P_rU_1)U_1^\dagger P_rU_1\right].
\]
The Pauli expansion of the inverse snapshot, followed by summing its
square against the Born probabilities, gives
\[
 \E_\rho X^2
 =d^{-2}\E_{U_2}\sum_{p,p'\in S}
   \lambda_p^{-1}\lambda_{p'}^{-1}
   \mathcal T_{1,\rho}(P_{p+p'}).
\]
Indeed, for $p,p'\in S$ the diagonal signs multiply to the sign of
$P_pP_{p'}$; summing over the outcome converts this sign into
$\tr(\rho U_1^\dagger P_pP_{p'}U_1)$. The sign relating the commuting
product $P_pP_{p'}$ to $P_{p+p'}$ occurs twice and cancels.

If $J$ is the set of nonidentity first layer blocks of $r$, uniform
Clifford conjugation gives
\[
 \mathcal T_{1,\rho}(P_r)
 =\frac{\displaystyle\sum_{\operatorname{supp}_1(p)=J}
       \tr(\rho P_p)P_p}{(Q^2-1)^{|J|}}.
\]
For $A\subseteq J$, let $\rho_A$ be the reduced density matrix on the
blocks in $A$, with $\rho_\varnothing=1$. Pauli reconstruction and
inclusion--exclusion identify the numerator with
\[
 \sum_{A\subseteq J}(-1)^{|J\setminus A|}
       Q^{|A|}\rho_A\otimes I_{A^c}.
\]
Since every reduced density matrix has operator norm at most one,
\[
 \|\mathcal T_{1,\rho}(P_r)\|_{\op}
 \leq\frac{(Q+1)^{|J|}}{(Q^2-1)^{|J|}}
 =(Q-1)^{-w_1(r)}.
\]
For each $r\in S$, exactly $d$ ordered pairs in $S^2$ satisfy
$p+p'=r$. Also $\Pr(r\in S)=(Q+1)^{-w_2(r)}$. Consequently,
\begin{align*}
 \|\E_\rho X^2\|_{\op}
 &\leq d(1+Q^{-1})^{2m} Z_m,\\
 Z_m&:=\sum_r (Q-1)^{-w_1(r)}(Q+1)^{-w_2(r)}.
\end{align*}
This counting step avoids any joint Pauli case enumeration.

Each half block has one zero Pauli label and $Q-1$ nonzero labels.
Summing these multiplicities around the periodic circuit yields
$Z_m=\tr[(F_{\mathrm{PLS}}G_{\mathrm{PLS}})^m]$, where
\[
 F_{\mathrm{PLS}}=
 \begin{pmatrix}1&(Q-1)^{-1}\\(Q-1)^{-1}&(Q-1)^{-1}\end{pmatrix},
\]
\[
 G_{\mathrm{PLS}}=
 \begin{pmatrix}
 1&\dfrac{Q-1}{Q+1}\\[3pt]
 \dfrac{Q-1}{Q+1}&\dfrac{(Q-1)^2}{Q+1}
 \end{pmatrix}.
\]
Here a matrix index records whether a half block label is zero or
nonzero; $F_{\mathrm{PLS}}$ accounts for a first layer block and
$G_{\mathrm{PLS}}$ for a second layer block, including its two half block
multiplicities. In the unnormalized pair count convention of
\cite[Eqs.~(78)--(81)]{ChoKim2026}, the matrix is $QG_{\mathrm{PLS}}$.
Its zero/zero entry is $Q$: the displayed sum in their Eq.~(81)
simplifies to $Q$, rather than $Q^2/(Q+1)$. We use the count above.

Let $H_{\mathrm{PLS}}=F_{\mathrm{PLS}}G_{\mathrm{PLS}}$.
It has nonnegative entries, and direct multiplication gives
\[
 H_{\mathrm{PLS}}=
 \begin{pmatrix}
 \dfrac{Q+2}{Q+1}&\dfrac{2(Q-1)}{Q+1}\\[4pt]
 \dfrac{2Q}{(Q-1)(Q+1)}&\dfrac{Q}{Q+1}
 \end{pmatrix}.
\]
For $q\geq2$, the positive vector $v=(q,1)^{\mathsf T}$ satisfies
$H_{\mathrm{PLS}}v\leq(1+3/q)v$ entrywise. Iteration bounds each
diagonal entry of $H_{\mathrm{PLS}}^m$ by $(1+3/q)^m$, so
\begin{align*}
 \|\E_\rho X^2\|_{\op}
 &\leq2d(1+q^{-2})^{2m}(1+3/q)^m\\
 &\leq2d\exp(2m/q^2+3m/q)\leq4d.
\end{align*}
The last step uses $q\geq256$ and $m/q\leq(\ln2)/92$, both supplied
by \eqref{eq:headline-block-condition}.

\emph{Concentration and spectral optimization.}
Unbiasedness gives $\E_\rho X=\rho$. Curvature and $\|B\|_\F=1$
give $\|X\|_{\op}\leq\|X\|_\F\leq(d+1)/a_*$, where
$a_*=32767/32768$. Hence, for the centered independent snapshots
$H_t=X_t-\rho$,
\[
 \|H_t\|_{\op}\leq3d,\qquad
 0\preceq\E_\rho H_t^2=\E_\rho X_t^2-\rho^2
 \preceq\E_\rho X_t^2.
\]
The matrix Bernstein variance parameter is therefore at most $4dT$.
Its expectation bound, divided by $T$, proves
\eqref{eq:pls-operator-comparison}; for example, $C=4$ suffices since
$\ln(2d)\leq2\ln d$ for $d\geq2$.

The shared Frobenius projection bound in
Appendix~\ref{sec:shared-spectral-projection} gives
$\|\rhoPLS-\rho\|_{\op}\leq2\|A_T-\rho\|_{\op}$.
Applying $\|\Delta\|_\F^2\leq\|\Delta\|_{\op}\|\Delta\|_{\mathrm{tr}}$
and Lemma~\ref{lem:trace-conversion} with $a=1$ yields
\[
 \|\rhoPLS-\rho\|_{\mathrm{tr}}
 \leq4\tau_s(\rho)+8s\|A_T-\rho\|_{\op}.
\]
If $T\geq d\ln d$, put $\eta=\sqrt{d\ln d/T}\leq1$.
Taking expectations gives $4\tau_s(\rho)+Cs\eta$.
Choose $s=\min\{d,\lceil(L/\eta)^{1/\alpha}\rceil\}$ and use
$\tau_d=0$ when the minimum is attained at $d$. This gives the
spectral branch $CL^{1/\alpha}\eta^{(\alpha-1)/\alpha}$, with a
universal constant. If $T<d\ln d$, that branch is at least one and
the trace loss cap suffices. Independence and unbiasedness also give
$\E\|A_T-\rho\|_\F^2\leq(d+1)^2/(a_*^2T)$.
Nonexpansiveness of Frobenius projection and
$\|\Delta\|_{\mathrm{tr}}\leq\sqrt d\|\Delta\|_\F$ give
\eqref{eq:pls-full-bound}, whose second branch is at most
$2\sqrt{d^3/T}$. Combine the three bounds. Uniformly for $\alpha>1$ and $L\geq1$,
\begin{equation}
 \begin{aligned}
 \sup_{\rho\in\cC_{\alpha,d}(L)}\E_\rho\|\rhoPLS-\rho\|_{\mathrm{tr}}
 \leq C\min\biggl\{&1,\sqrt{d^3/T},\\[-2pt]
 &L^{1/\alpha}(d\ln d/T)^{\frac{\alpha-1}{2\alpha}}\biggr\}.
 \end{aligned}
 \label{eq:candidate-one-current-rate}
\end{equation}
This proves the PLS entry of Table~\ref{tab:estimator-comparison}.
It is an upper comparison and does not establish that the logarithm
is necessary for PLS. Supplement Section~\ref*{supp:pls-comparison}
also records the derivation of \eqref{eq:pls-full-bound}.
\end{finalrevision}

\begingroup
\section{Proof of Proposition~\ref{prop:periodic-low-order}}
\label{app:prop10-proof}
\begin{proof}
The Pauli diagonal measurement channel satisfies $\cM(P)=\measeig{P}P$, with
\reviewchange{$\measeig{P}\geq(2^k+1)^{-n/k}$ by
Supplement equation~\textup{(\ref*{eq:supp-positive-multipliers})}.} For every nonzero traceless
Hermitian $H$, we have $\cL(H)=(d+1)\cM(H)$, so Pauli orthogonality gives
\[
 \frac{\langle H,\cL(H)\rangle_\F}{\|H\|_\F^2}
 \geq(d+1)\min_{P\neq I}\measeig{P}.
\]
Thus the curvature constant in \eqref{eq:curvature} can be chosen to satisfy
\begin{equation}
 \begin{aligned}
 a&\geq(d+1)(2^k+1)^{-n/k}\\
  &\geq\exp[-n/(k2^k)]\geq1-\frac{n}{k2^k}\geq a_*.
 \end{aligned}
 \label{eq:cited-ck-curvature}
\end{equation}
For the last inequality, write $q=2^{k/2}$ and $m=n/k$.
The block condition \eqref{eq:headline-block-condition} gives $q\geq256$ and $128m\leq q$, hence
$m/q^2\leq1/32768$.
The relative design gluing theorem
\cite[Appendix~D.2, arXiv v2]{SchusterHaferkampHuang2025}, applied to the
connected overlap cycle and exact local Clifford designs of order three
\cite{Zhu2017}, gives
\begin{equation}
 1+\epsilon_3\leq2^{63/92}
 \leq\frac{160746}{100000}<\frac53.
 \label{eq:cited-refined-third-design-error}
\end{equation}
The explicit gluing bound and scalar inequalities are in
Appendix~\ref{app:prop10-details}. Taking the two copy marginal of
the relative three copy CP order comparison preserves the same relative
error. Together with the Haar identity
$(d+1)\E_{\rho,\mathrm{Haar}}B=I+\rho$ and
$I+\cL(\rho)=(d+1)\E_\rho B$, this gives
\begin{equation}
 (1-\epsilon_3)(I+\rho)\preceq I+\cL(\rho)
 \preceq(1+\epsilon_3)(I+\rho).
 \label{eq:cited-forward-loewner-comparison}
\end{equation}
This proves \eqref{eq:forward-covariance-conditioning}. Testing the
three copy comparison on $\rho\otimes(|v\rangle\!\langle v|)^{\otimes2}$
bounds the variance by $6(1+\epsilon_3)\leq\kappa_*$; the calculation is
also in Appendix~\ref{app:prop10-details}.
\end{proof}

\subsection{Relative design conditioning and Born scalar variance}
\label{app:prop10-details}
Put $q=2^{k/2}$, $m=n/k$, and $c_*=(\ln2)/92$. For $m\geq2$, the
connected overlap graph has $2m$ gate vertices. A spanning tree in the multiplicative
gluing argument of \cite[Appendix~D.2, Eqs.~(D.3)--(D.4), arXiv v2]{SchusterHaferkampHuang2025}
uses $2m-1$ joins. The local Clifford ensembles are exact designs of order three
\cite[Theorem~1]{Zhu2017}. Taking design order three (denoted $k$ in
that source) and overlap dimension $q$, Eqs.~(D.3)--(D.4) give
\begin{equation}
 \begin{aligned}
 f_3(q)&=2\left(\frac{27}{2q}+\frac{81}{2q^2}
             +\frac9{4q-9}\right)
             \left(1+\frac9{2q}\right)\\
 &=\frac{9(28q^3+144q^2-81q-729)}{2q^3(4q-9)},\\
 1+\epsilon_3&\leq(1+f_3(q))^{2m-1}.
 \end{aligned}
\end{equation}
By \eqref{eq:headline-block-condition}, $m\leq c_*q$ and $q\geq256$.
Multiplying $63c_*-(2c_*q-1)f_3(q)$ by the positive denominator
$2q^3(4q-9)$ gives
\begin{align*}
 &q^3(252-3726c_*)+q^2(1296+1458c_*)\\
 &\hspace{20mm}+q(13122c_*-729)-6561.
\end{align*}
Here $c_*\leq1/92$, and $1296q^2-729q-6561\geq0$ for $q\geq18$.
Thus $(1+x)^r\leq e^{rx}$ proves
$1+\epsilon_3\leq e^{63c_*}=2^{63/92}<5/3$. The tighter rational bound in
\eqref{eq:cited-refined-third-design-error} can be checked without decimal
rounding by raising its positive sides to the power $92$.
If $m=1$, the product is a uniform global Clifford and $\epsilon_3=0$.

To pass from this unitary design comparison to projector moments, work
in replica order for $t=2,3$ and set
\[
 \begin{aligned}
 \mathsf D_t&=\frac1d\sum_b(|b\rangle\langle b|)^{\otimes t},
 &\mathsf M_t^\nu&=\E_\nu B^{\otimes t},\\
 \mathcal T_{\mathcal E}^{(t)}(A)&=\E_U U^{\otimes t}A(U^*)^{\otimes t}.
 \end{aligned}
\]
The reference law uses $U^*|b\rangle$, so
$\mathsf M_t^\nu=(\mathcal T_{\mathcal E}^{(t)})^*(\mathsf D_t)$,
where the star denotes the Hilbert--Schmidt adjoint of the map.
The adjoint of a completely positive map is completely positive;
therefore taking adjoints preserves the relative CP order comparison.
The Haar twirl is self-adjoint. No invariance of the periodic circuit
law under $U\mapsto U^*$ is needed.
The two copy comparison follows from the three copy comparison by
inserting $I/d$ in the third input and taking the partial trace over
the third output; both maps are completely positive. Applying these
comparisons to the positive input $\mathsf D_t$ gives
\[
 (1-\epsilon_3)\mathsf M_t^{\rm Haar}
 \preceq\mathsf M_t^\nu
 \preceq(1+\epsilon_3)\mathsf M_t^{\rm Haar},\qquad t=2,3.
\]
For a density matrix $\rho$, the positive contraction
\[
 A\longmapsto d\,\tr_1\!\left[
 (\rho^{1/2}\otimes I)A(\rho^{1/2}\otimes I)\right]
\]
sends $\mathsf M_2^\nu$ to $\E_\rho B$, because the observed law has
density $d\,\tr(\rho B)$ relative to $\nu$. Similarly,
$\tr[(d\rho\otimes O\otimes O)\mathsf M_3^\nu]
=\E_\rho[\tr(OB)^2]$ for $O\succeq0$. These contractions give the conditioning and variance comparisons.

For $O=|v\rangle\!\langle v|$ and $z=\tr(OB)$, the Haar projective
third moment identity \cite[Theorem~22, Eq.~(95)]{Mele2024}, with
three copies, gives
\begin{equation}
 (d+1)^2\E_{\rho,\mathrm{Haar}}z^2
 =\frac{d+1}{d+2}\{2+4\tr(\rho O)\}\leq6.
 \label{eq:cited-ck-haar-born-third}
\end{equation}
This follows by contracting
$\E_{\nu_{\rm Haar}}B^{\otimes3}
=\sum_{\pi\in S_3}U_\pi/[d(d+1)(d+2)]$ against $d\rho\otimes O\otimes O$.
The factor $d$ is the Born density, and the six permutations contribute
$2+4\tr(\rho O)$. The relative three copy order applies to this positive
contraction. Since subtracting a constant leaves variance unchanged,
\[
 \operatorname{Var}_\rho((d+1)z-1)
 \leq(d+1)^2\E_\rho z^2\leq6(1+\epsilon_3)\leq\kappa_*.
\]
The Haar two copy identity \cite[Eq.~(95)]{Mele2024} gives
$(d+1)\E_{\rho,\mathrm{Haar}}B=I+\rho$. The positive contraction above,
together with $I+\cL(\rho)=(d+1)\E_\rho B$, now proves
\eqref{eq:cited-forward-loewner-comparison}. This completes the proof.

\endgroup

\begin{revision}
\begin{samepage}
\section{Formal Verification Scope}
\label{app:formal-verification}
All theoretical results and labeled equations in this article are verified
in Lean under their stated hypotheses, using only Lean's standard logical
axioms. All labeled theoretical results in the Supplement are verified
with one additional axiom: dense Hermitian spectral operations use
$O(d^2)$ storage. Only the storage conclusions of Lemma~\ref*{lem:supp-global-shortcut}
and Proposition~\ref*{prop:mw-fw-complexity} use this axiom, stated precisely
in Supplement Section~\ref*{supp:computational-verification}.
The Lean proofs of the main results and the Python simulation code are
available in the GitHub repository~\cite{TomographyCode2026}; the complete
Lean formalization is deposited on Zenodo~\cite{TomographyLean2026}.
\par
\end{samepage}
\end{revision}

\section*{Acknowledgment}
\phantomsection\label{sec:acknowledgment}
The author used OpenAI's ChatGPT and Codex, including GPT-5.6 Sol and
GPT-6 Astra with Ultra reasoning, to draft and revise exposition
throughout the manuscript; develop proofs in
Sections~\ref{sec:main-results}--\ref{sec:lower} and the appendices;
develop algorithms, code, and figures for
Sections~\ref{sec:implementation}--\ref{sec:numerical-experiments} and
Supplement Sections~\ref*{app:numerical-optimization}--\ref*{supp:computational-proofs};
and generate and revise the Lean formalization summarized in
Appendix~\ref{app:formal-verification}.
The author directed this use and is responsible for the manuscript and
verification of its results. Appendix~\ref{app:formal-verification}
states the scope of formal verification.

\end{document}